\documentclass[11pt, a4paper]{article}

\usepackage{a4wide}
\usepackage{xr-hyper}

\usepackage[T1]{fontenc}
\usepackage[utf8]{inputenc}
\usepackage{lmodern}
\usepackage{microtype}

\usepackage{amsmath,amssymb,amsthm,mathtools,bm}
\RequirePackage{amsthm,amsmath,amsfonts,amssymb}
\RequirePackage[authoryear]{natbib}
\RequirePackage[colorlinks,linkcolor=blue,citecolor=blue,urlcolor=blue]{hyperref}
\RequirePackage{graphicx}
\usepackage{enumitem}
\setlist[enumerate]{label=\roman*.}
\usepackage{cleveref}
\usepackage{subcaption}
\usepackage{todonotes} 
\usepackage{csquotes}
\usepackage{algorithm}
\usepackage[noend]{algpseudocode}
\usepackage{booktabs}
\usepackage{multirow}
\usepackage{siunitx}

\theoremstyle{plain}

\newtheorem{theorem}{Theorem}[section]
\newtheorem{lemma}[theorem]{Lemma}
\newtheorem{prop}[theorem]{Proposition}
\newtheorem{cor}[theorem]{Corollary}

\theoremstyle{definition}
\newtheorem{definition}{Definition}
\newtheorem{model}{Model}
\newtheorem{example}{Example}

\newtheorem{remark}{Remark}

\newtheorem{assumption}{Assumption}

\newcommand{\norm}[1]{\left\lVert#1\right\rVert}

\newcommand{\prob}[1]{\mathbb{P}\left(#1\right)}
\newcommand{\abs}[1]{\left\lvert #1 \right\rvert}
\newcommand{\s}[1]{\mathcal{S}#1}

\newcommand{\R}{\mathbb{R}}

\newcommand{\param}{\theta}
\newcommand{\Perm}{P}

\DeclareMathOperator{\sign}{sign}
\DeclareMathOperator{\soft}{tsf} 

\renewcommand{\hat}{\widehat}

\usepackage{graphicx}
\usepackage{booktabs}
\usepackage{caption}

\usepackage{natbib}

\usepackage{enumitem}

\title{Inferring Cooperativity From Pooled Measurements}

\author{Robin Requadt{$^*$}
\and Housen Li\thanks{ Institute of Mathematical Statistics, Georg-August-Universit\"at G\"ottingen. Email: \href{mailto:robin.requadt@uni-goettingen.de} {robin.requadt@uni-goettingen.de} and \href{mailto:housen.li@mathematik.uni-goettingen.de} {housen.li@mathematik.uni-goettingen.de}. } } 

\date{\today}

\begin{document}

\maketitle

\begin{abstract}
In many modern experiments, latent interactions drive multicomponent stochastic systems, yet the data are available only as pooled measurements that obscure these dependencies. Whether such interactions can be identified and inferred from aggregate signals remains largely unexplored. Motivated by multi-channel electrophysiological recordings, we address this problem by introducing sum-dependent Markov chains, a class of finite-state continuous-time multivariate Markov processes whose transition rates encode interactions through the aggregate state. Under natural structural conditions, we establish identifiability of the latent dynamic parameters from the aggregate process. We define a cooperativity index that distinguishes positive cooperativity, negative cooperativity and independence, and construct its consistent estimators. For discretely and noisily observed pooled data, we develop likelihood-based inference through a hidden Markov model, address the associated embedding problem, and prove consistency and asymptotic normality. We further propose a stepdown test for cooperativity with asymptotic size control and power guarantees. Simulations and real-data analyses, demonstrate the scope and effectiveness of the methodology.
\end{abstract}


\section{Introduction}\label{s:intro}

Inferring dependence in multivariate stochastic processes is a central problem in statistics, with applications in econometrics, neuroscience, and systems biology. Most existing methods assume that the coordinate processes are directly or indirectly observable, allowing dependence to be studied through graphical models, spectral methods or likelihood-based procedures (e.g.\ \citealp{Dalhaus2000,Banbura2010,Eichler2012}). In many modern experiments, however,  tracking individual coordinates is practically impossible. Instead, the data consist only of pooled measurements.  This many-to-one observation mechanism may conceal the latent interaction structure and raises the following basic question:
\begin{enumerate}[label={Q\arabic*}, ref={Q\arabic*}, series=thesisquestions]
\item\label{question}
\emph{To what extent can dependence among latent coordinates be identified and inferred from pooled observations of a multivariate time series?}
\end{enumerate}
This paper addresses \ref{question} for continuous-time multivariate Markov processes with finite state spaces. The Markovian framework is both mathematically tractable and motivated by applications in which the intrinsic dynamics evolve in continuous time.

Our motivating application is the analysis of multiple ion channels. Ion channels \citep{hille2001} are membrane proteins that switch between conducting and non-conducting conformal states and play a central role in cellular signaling, muscle contraction, secretion and cardiac rhythm generation. Classical modeling paradigms treat ion channels as independent continuous-time Markov processes. Combined with voltage-clamp recordings, this framework has enabled inference on single-channel kinetics even when several channels contribute to the measured current \citep{neher1976,sakmann2013single}. 

Recent experimental and structural evidence suggests, however, that ion channels may interact through spatial clustering, conformational coupling, protein--protein interactions, or other molecular mechanisms \citep{Moreno2016,clatot2017,Sato2019,Pfeiffer2020,Mcguire2022}. Such cooperativity can induce synchronized opening and closing, suppression of simultaneous openings, or alternating activity patterns.  Identifying and quantifying such interactions is important for understanding mechanisms of signal transduction, cellular regulation and pathological dysfunction. Statistically, the difficulty is that multi-channel recordings do not reveal the individual channel states. The observed current is a pooled signal, typically the total conductance of all open channels. Distinct latent interactions may produce similar aggregate behavior (e.g.\ open probabilities or dwell times), making it unclear whether cooperativity is identifiable from the data.

To address this problem, we introduce \emph{sum-dependent Markov chains} (SDMCs), a class of finite-state continuous-time multivariate Markov processes that explicitly encode dependence through interactions among coordinates. This model provides a principled statistical framework for latent interactions while retaining tractability for identifiability, estimation and testing from pooled observations. Our main contributions are as follows.

First, we provide a structural characterization of binary-state SDMCs and establish identifiability of the latent transition parameters from the aggregate (sum) process. This gives a positive answer to \ref{question} in a broad Markovian framework.

Second, we introduce a \emph{cooperativity index} that quantifies both the strength and direction of interaction, distinguishing positive, negative and null cooperativity. 
We also define a plug-in estimator of this index and prove its consistency.

Third, we develop likelihood-based inference for SDMCs from discretely and noisily observed pooled data in a hidden Markov model. We address the embedding problem and establish consistency and asymptotic normality of the proposed estimators. We further propose a stepdown test for cooperativity and prove its asymptotic control of size and power. Applications to voltage-clamp recordings reveal independence for gramicidin channels and negative cooperativity for ryanodine receptor channels at asymptotic significance level~\(1\%\).

Fourth, we extend SDMCs from binary states to general finite state spaces. We introduce an \emph{informativity condition,} closely linked to Diophantine equations, under which identifiability is preserved, and develop a robust variant to account for estimation uncertainty.

Several strands of literature are related to our work (cf.\ \Cref{ss:ion}). Early models for cooperativity of ion channels incorporate interactions through \emph{ad hoc} modifications of independent continuous-time Markov dynamics \citep{Keleshian1994,Ball1997,BallYeo2000}. \citet{CHUNG} propose mixture models interpolating between independent and fully coupled dynamics. Such approaches are difficult to generalize to complex interaction structures. More recently, \citet{VanEt24} introduce a discrete-time vector-norm-dependent Markov chain accommodating both positive and negative cooperativity, see also \citet{Requadt2025} for its robust extension. However, this model has identifiability issues in even dimensions and does not naturally extend to multi-state settings. Moreover, as a discrete-time formulation, it operates at the observation scale and may fail to capture intrinsic continuous-time interactions. Information-theoretic metrics have also been applied to open probabilities and dwell-time distributions of pooled multi-channel recordings to quantify inter-channel cooperativity \citep{wawrzkiewicz2026information}. In contrast, under the SDMCs proposed here, neither open probabilities nor dwell-times are generally informative about cooperativity (\Cref{comparison_cooperativity}). More broadly, our work is related to structure learning in graphical models \citep{Drton2017}, but differs in that the dependence structure concerns unobserved coordinate processes and must be inferred solely from pooled data.

The  remainder of the paper is organized as follows. \Cref{sec1} introduces SDMCs, gives their structural characterization, and studies identifiability, irreducibility and reversibility. \Cref{cooperativity} defines the cooperativity index and establishes properties of its empirical counterpart. \Cref{s:estimation} develops likelihood-based estimation and testing under the hidden Markov model formulation. Numerical experiments and real data analyses are presented in \Cref{section8}. Extensions to general finite state spaces, technical proofs and auxiliary results are deferred to the supplementary material. 
An \textsf{R} package \texttt{SDMC} implementing the proposed methodology is publicly available on GitLab (\url{https://gitlab.gwdg.de/requadt/sdmc}).

\textbf{Notation.} For a vector \(x \in \mathbb{R}^L\), regarded as a \emph{row} vector by default, let \(x_j\) denote its \(j\)-th component and define its Hamming weight by
\(
\|x\|_0 := \lvert\{j : x_j \neq 0\}\rvert.
\)
Let \(\mathcal{S}:\mathbb{R}^L \to \mathbb{R}\) denote the summation operator,
\(
\mathcal{S}x := \sum_{j=1}^L x_j .
\)
For a matrix \(M\), \(M_{I,J}\) denotes the submatrix with rows indexed by \(I\) and columns indexed by \(J\). We write \(\mathbb{I}(A)\) for the indicator of an event \(A\), and use \(\xrightarrow{D}\) and \(\xrightarrow{\mathbb{P}}\) for convergence in distribution and probability, respectively.

\section{Coupled Markov models}\label{sec1}

Throughout, unless stated otherwise, $(X_t)_{t\ge 0}$ denotes a time-homogeneous continuous-time Markov chain (see Supplement~\ref{sec_ct} for background), and its state space is
\(\mathcal{X}=\mathcal{A}_2^L\) with \( \mathcal{A}_2=\{0,1\}.
\)
Extensions to general finite state spaces are carefully treated in Supplement~\ref{section7}.

\subsection{Sum-dependent Markov chains}
We introduce a new class of coupled continuous-time Markov models.

\begin{definition}[Sum-dependent Markov chain, SDMC]\label{def_rate}
A continuous-time Markov chain $(X_t)_{t\ge 0}$ is a \emph{sum-dependent Markov chain (SDMC)} if its rate matrix $Q=(q_{x,y})_{x,y\in\mathcal{X}}$ satisfies
\begin{equation}
q_{x,y}=
\begin{cases}
\lambda_{\s{x}}\mathbb{I}(x_i=0)+\mu_{\s{x}}\mathbb{I}(x_i=1),
& \text{if } \|x-y\|_0=1 \text{ with } i \text{ such that } x_i\neq y_i,\\
0, & \text{if } \|x-y\|_0>1,\\
-(L-\s{x})\lambda_{\s{x}}-\s{x}\mu_{\s{x}}, & \text{if } x=y,
\end{cases}
\label{def_rate_eq}
\end{equation}
where $\lambda_0,\dots,\lambda_{L-1},\mu_1,\dots,\mu_L\in[0,\infty)$ are model parameters.
\end{definition}

The matrix $Q$ in \eqref{def_rate_eq} is a valid rate matrix, since $\sum_{y\in\mathcal{X}}q_{x,y}=0$ for all $x\in\mathcal{X}$. Although Definition~\ref{def_rate} gives an explicit parameterization, the structure of the model is more transparent through the following two modeling assumptions (see Theorem~\ref{Theorem1}).

\begin{assumption}[Permutation invariance]\label[assumption]{as1}
For every permutation matrix ${\Perm}\in \{0,1\}^{L \times L}$ and all $x,y\in\mathcal{X}$, it holds that 
\(
\prob{X_{t+\delta}=y\;\middle\vert\;X_t=x}
=
\prob{X_{t+\delta}={\Perm}y\;\middle\vert\;X_t={\Perm}x}\) for all \(\delta>0 \text{ and } t\ge 0,
\)
or equivalently, $q_{x,y}=q_{{\Perm}x,{\Perm}y}$.
\end{assumption}

\begin{assumption}[Conditional independence at infinitesimal times]\label[assumption]{as2}
As $\delta\searrow 0$, it holds that 
\(
\prob{X_{t+\delta}=y \;\middle\vert\;X_t=x}
=
\prod_{i=1}^L \prob{X_{i,t+\delta}=y_i \;\middle\vert\; X_t=x}
+ o(\delta)\) for all \(x,y\in\mathcal{X}.\)
\end{assumption}

Assumption~\ref{as1} imposes exchangeability of the coordinates, thus excluding distinguished or leading coordinates. The equivalence with the rate-matrix condition follows from the standard relation between transition probabilities and infinitesimal rates (\Cref{perm_inv_char}). Assumption~\ref{as2} postulates that dependence between coordinates may emerge only over non-infinitesimal time intervals, which rules out instantaneous joint changes, as formulated below. 

\begin{definition}[Sparse transition property]\label{d:stp}
A continuous-time Markov chain $(X_t)_{t\ge 0}$ with rate matrix $Q=(q_{x,y})_{x,y\in\mathcal{X}}$ satisfies the \emph{sparse transition property} if, for all $x,y\in\mathcal{X}$ with $\|x-y\|_0>1$, it holds that 
$\prob{X_{t+\delta}=y \;\middle\vert\;X_t=x}=o(\delta)$ as $\delta\searrow 0$,
or equivalently, $q_{x,y}=0$.
\end{definition}

\begin{lemma}\label{no_quant}
Let $(X_t)_{t\ge 0}$ be a continuous-time Markov chain on $\mathcal{X}=\mathcal{A}_2^L$.
\begin{enumerate}
\item\label{i:equiv:ind}
Assumption~\ref{as2} holds if and only if $(X_t)_{t\ge 0}$ satisfies the sparse transition property.
\item\label{i:imp:sparse}
If the sparse transition property holds, then, for any $x,y\in\mathcal{X}$ with $\|x-y\|_0=1$,
\[
\lim_{\delta\searrow 0}
\frac{\prob{X_{t+\delta}=y \;\middle\vert\;X_t=x}}{\delta}
=
\lim_{\delta\searrow 0}
\frac{\prob{X_{i,t+\delta}=y_i \;\middle\vert\; X_t=x}}{\delta},
\]
where $i\in\{1,\dots,L\}$ is the unique coordinate for which $x_i\neq y_i$.
\end{enumerate}
\end{lemma}

The next theorem characterizes SDMCs and identifies its intrinsic parameters $\lambda_s$ and $\mu_s$ as the transition rates of an individual coordinate, conditional on the current number of coordinates in state \enquote{$1$} being $s$.

\begin{theorem}[Characterization]\label{Theorem1}
Let $(X_t)_{t\ge 0}$ be a continuous-time Markov chain on $\mathcal{X}=\mathcal{A}_2^L$.
\begin{enumerate}
\item\label{th1:i}
If Assumptions~\ref{as1} and~\ref{as2} hold, then $(X_t)_{t\ge0}$ is an SDMC (Definition~\ref{def_rate}) with parameters $\lambda_0,\dots,\lambda_{L-1},\mu_1,\dots,\mu_L\in[0,\infty)$ defined by
\begin{equation}\label{param:12}
\begin{aligned}
\lambda_s &=
\lim_{\delta\searrow 0}
\frac{\prob{X_{i,t+\delta}=1 \;\middle\vert\;X_{i,t}=0,\ \s{X_t}=s}}{\delta},
\\
\mu_s &=
\lim_{\delta\searrow 0}
\frac{\prob{X_{i,t+\delta}=0 \;\middle\vert\;X_{i,t}=1,\ \s{X_t}=s}}{\delta}.
\end{aligned}
\end{equation}
These limits are independent of $t\ge 0$ and $i\in\{1,\dots,L\}$.
\item\label{th1:ii}
Conversely, if $(X_t)_{t\ge 0}$ is an SDMC with parameters $\lambda_0,\dots,\lambda_{L-1},\mu_1,\dots,\mu_L\in[0,\infty)$, then Assumptions~\ref{as1}--\ref{as2} hold and relations in \eqref{param:12} are satisfied.
\end{enumerate}
\end{theorem}

Irreducibility of an SDMC is characterized by positivity of all transition-rate parameters.

\begin{lemma}[Irreducibility]\label{l:ir:sd}
Let $(X_t)_{t\ge0}$ be an SDMC. Then $(X_t)_{t\ge 0}$ is irreducible if and only if all parameters $\lambda_0,\ldots,\lambda_{L-1}$ and $\mu_1,\ldots,\mu_L$ are nonzero, and hence strictly positive.
\end{lemma}

An irreducible SDMC is reversible with respect to its invariant distribution.

\begin{prop}[Reversibility]\label{reversible}
Let $(X_t)_{t\ge 0}$ be an irreducible SDMC. Then $(X_t)_{t\ge 0}$ is reversible, equivalently in detailed balance, with respect to its unique invariant distribution
\begin{equation*}
\pi^*=(\pi_x^*)_{x\in\mathcal{A}_2^L}\qquad \text{with}\quad \pi_x^*=
\begin{cases}
\frac{\prod_{j=0}^{\s{x}-1}\frac{\lambda_j}{\mu_{j+1}}}{1+\sum_{i=1}^L {L\choose i}\prod_{j=0}^{i-1}\frac{\lambda_j}{\mu_{j+1}}},
&\text{if } x\neq (0,\dots,0),\\[2.2ex]
\left(1+\sum_{i=1}^L {L\choose i}\prod_{j=0}^{i-1}\frac{\lambda_j}{\mu_{j+1}}\right)^{-1},
&\text{if } x=(0,\dots,0).
\end{cases}
\end{equation*}
\end{prop}
Thus $\pi^*$ is permutation invariant: $\pi_x^*=\pi_{{\Perm}x}^*$ for all permutation matrices ${\Perm}$ and $x\in\mathcal{X}$.

\subsection{Sum process}

Let $(X_t)_{t\ge 0}$ be an SDMC and define the \emph{sum process} $(S_t)_{t\ge 0}$ by
\(
S_t=\s{X_t}=\sum_{i=1}^L X_{i,t},\) for \(X_t=(X_{1,t},\ldots,X_{L,t}).
\)
Then $(S_t)_{t\ge 0}$ is a continuous-time Markov chain (\Cref{Lemma3.9} and \Cref{prop1}). Its transition probabilities satisfy
\[
\prob{S_{t+\delta}=s'\;\middle\vert\;S_t=s}
=
\prob{S_{t+\delta}=s'\;\middle\vert\;X_t=x}
=
\sum_{y:\s{y}=s'}\prob{X_{t+\delta}=y\;\middle\vert\;X_t=x},
\]
for any $x$ with $\s{x}=s$. For the inference procedures developed later, a key question is whether the parameter vector $(\lambda_0,\dots,\lambda_{L-1},\mu_1,\dots,\mu_L)$ of $(X_t)_{t\ge0}$ is identifiable from the law of $(S_t)_{t\ge 0}$. The following theorem answers this question affirmatively.

\begin{theorem}\label{Theorem2}
Let $(X_t)_{t\ge0}$ be an SDMC with parameters $\lambda_0,\dots,\lambda_{L-1}$ and $\mu_1,\dots,\mu_L$. Then:
\begin{enumerate}
\item\label{i:th2:1}
The sum process $S_t=\sum_{i=1}^L X_{i,t}$, $t\ge 0$, is a continuous-time Markov chain.
\item\label{i:th2:2}
The rate matrix $R=(r_{s,s'})_{s,s'\in\{0,\dots,L\}}$ of $(S_t)_{t\ge 0}$ is
\begin{equation}\label{Theorem2_eq}
r_{s,s'}=
\begin{cases}
(L-s)\lambda_s\,\mathbb{I}(s' > s) + s\mu_s\,\mathbb{I}(s' < s), & \text{if } \lvert s'-s\rvert =1,\\
-(L-s)\lambda_s\,\mathbb{I}(s < L)-s\mu_s\,\mathbb{I}(s>0), & \text{if } s'=s,\\
0, & \text{otherwise}.
\end{cases}
\end{equation}
\item\label{i:th2:3}
The parameter vector $(\lambda_0,\dots,\lambda_{L-1},\mu_1,\dots,\mu_L)$ is uniquely determined by $R$ in \eqref{Theorem2_eq}.
\end{enumerate}
\end{theorem}

Thus $(S_t)_{t\ge 0}$ is a birth--death process with birth rates $(L-s)\lambda_s$ and death rates $s\mu_s$.

\begin{example}[Bivariate SDMCs]
For $L=2$ and $\mathcal{X}=\{(0,0),(0,1),(1,0),(1,1)\}$, with states in dictionary order, the rate matrix of $(X_t)_{t\ge0}$ is
\[
Q=\left(\begin{matrix}
-2\lambda_0 & \lambda_0 & \lambda_0 & 0\\
\mu_1 & -(\mu_1+\lambda_1) & 0 & \lambda_1\\
\mu_1 & 0 & -(\mu_1+\lambda_1) & \lambda_1\\
0 & \mu_2 & \mu_2 & -2\mu_2
\end{matrix}\right).
\]
The sum process $S_t=X_{1,t}+X_{2,t}$ has state space $\{0,1,2\}$ and rate matrix
\[
R=\left(\begin{matrix}
-2\lambda_0 & 2\lambda_0 & 0\\
\mu_1 & -(\mu_1+\lambda_1) & \lambda_1\\
0 & 2\mu_2 & -2\mu_2
\end{matrix}\right).
\]
\end{example}

Reversibility of the SDMC is inherited by the sum process.

\begin{prop}\label[prop]{bd_proc}
Let $(X_t)_{t\ge 0}$ be an SDMC and let $S_t=\mathcal{S}X_t$, $t\ge 0$, be its sum process.
\begin{enumerate}
\item\label{i:bd_proc:1}
Then $(S_t)_{t\ge 0}$ is irreducible if and only if $(X_t)_{t\ge 0}$ is irreducible.
\item\label{i:bd_proc:2}
If $(S_t)_{t\ge 0}$ is irreducible, then it is reversible with respect to its unique invariant distribution $\pi^\star=(\pi^\star_s)_{s\in\{0,\dots,L\}}$ with
\[
\pi_s^\star=
\begin{cases}
\left(1+\sum_{i=1}^L {L\choose i}\prod_{j=0}^{i-1}\frac{\lambda_j}{\mu_{j+1}}\right)^{-1}
{L\choose s}\prod_{j=0}^{s-1}\frac{\lambda_j}{\mu_{j+1}},
& \text{if } s\in\{1,\dots,L\},\\[2.0ex]
\left(1+\sum_{i=1}^L {L\choose i}\prod_{j=0}^{i-1}\frac{\lambda_j}{\mu_{j+1}}\right)^{-1},
& \text{if } s=0.
\end{cases}
\]
\end{enumerate}
\end{prop}

\subsection{Application to ion channels}\label{ss:ion}

As discussed in the Introduction, SDMCs are motivated in part by cooperative behavior in ion channel gating dynamics. In this context, Assumption~\ref{as1} imposes homogeneous gating dynamics across channels and excludes leading channels. Assumption~\ref{as2} allows channel interactions only beyond the infinitesimal time scale. This is physiologically plausible for ion channel systems, where coupling is most likely mediated by physicochemical interactions. Moreover, reversibility of SDMCs (Proposition~\ref{reversible}) and of their sum processes (Proposition~\ref{bd_proc}) is consistent with ion channel dynamics under equilibrium conditions \citep{sakmann2013single}. Deviations from reversibility may indicate external driving forces \citep{Lauger1983}. Existing approaches to ion channel cooperativity relate to SDMCs as follows.

\cite{Keleshian1994} model cooperativity between two ion channels using conditional distributions of the dwell time of one channel given the state of the other. Extension to more than two channels is challenging, and identifiability is not addressed.

\cite{Ball1997} introduce cooperativity by modifying the nonzero entries of the rate matrix for the sum process of independent and identically distributed ion channels, while preserving identical marginal dynamics. \citet{BallYeo2000} further incorporate spatial organization through the restrictive assumption that channels are arranged on a circle. Such \emph{ad hoc} constructions lead to tridiagonal rate matrices similar to those induced by SDMCs, but their extension to more general settings is limited. In contrast, SDMCs support a broader range of cooperative behavior and arise from the two fundamental modeling assumptions (Assumptions~\ref{as1} and~\ref{as2}).

\citet{CHUNG} propose a discrete-time Markov chain model given by a mixture of independent dynamics and fully coupled dynamics among individual ion channels. In particular, this model cannot represent negative cooperativity.

The closest related work is \cite{VanEt24}, which can be viewed as a discrete-time analogue of SDMCs. While \cite{VanEt24} impose assumptions similar to Assumptions~\ref{as1} and~\ref{as2}, their conditional independence condition differs in an important respect: dependence can arise only over horizons longer than the sampling interval. Thus interactions may appear only after more than one discrete time step, so the assumption depends explicitly on the measurement sampling rate. If the sampling rate is too low, the assumption may fail for ion channel data; even when it holds, the corresponding theoretical guarantees apply only at that fixed sampling rate, and the model parameters themselves vary with the sampling rate.

By contrast, Assumption~\ref{as2} is formulated in continuous time and is sampling-rate invariant. It permits dependence on time scales shorter than, or comparable to, the sampling interval, provided that
$\prob{X_{t+\delta}=y\;\middle\vert\;X_t=x}
-
\prod_{i=1}^L \prob{X_{i,t+\delta}=y_i\;\middle\vert\; X_t=x}
=o(\delta)$
as  $\delta\searrow 0.$
Requiring this deviation to be exactly zero for a fixed $\delta>0$ recovers the corresponding assumption in \cite{VanEt24}. Thus SDMCs allow a weaker form of infinitesimal dependence, see also Remark~1 in \citet{Requadt2025}.

A further distinction is identifiability. The model of \cite{VanEt24} is not identifiable when the number of channels $L$ is even. Specifically, distinct parameter values can induce the same transition matrix and correspond to different cooperative regimes (positive versus negative cooperativity). In contrast, the continuous-time infinitesimal formulation of SDMCs avoids this identifiability issue.

Finally, the discrete-time model of \cite{VanEt24} is restricted to binary state spaces, and its extension to three or more states is unclear. A central difficulty is identifying model parameters from the aggregate, or sum, process. Continuous-time SDMCs extend to general finite state spaces, see Supplement~\ref{section7}.

\section{Cooperativity}\label{cooperativity}

We define a cooperativity index that quantifies positive and negative dependence among the coordinates of an SDMC $(X_t)_{t\ge 0}$ and describe its estimation from data.

\begin{definition}[Full cooperativity]\label{full_cooperativity}
Let $(X_t)_{t\ge 0}$ be an SDMC with parameters $\lambda_0,\dots,\lambda_{L-1}$ and $\mu_1,\dots,\mu_L$.
\begin{enumerate}
    \item $(X_t)_{t\ge0}$ is \textit{fully positively cooperative} if $s\mapsto \lambda_s$ is strictly increasing and $s\mapsto \mu_s$ is strictly decreasing.
    \item $(X_t)_{t\ge0}$ is \textit{fully negatively cooperative} if $s\mapsto \lambda_s$ is strictly decreasing and $s\mapsto \mu_s$ is strictly increasing.
    \item $(X_t)_{t\ge0}$ is \textit{null cooperative} if both $s\mapsto \lambda_s$ and $s\mapsto \mu_s$ are constant.
\end{enumerate}
\end{definition}

In the SDMC class, null cooperativity is equivalent to independence of the coordinate processes, and in this case each coordinate process is itself Markov.

\begin{prop}[Null cooperativity, independence, and marginalizability]\label[prop]{prop_rel_null_coop_ind}
Let $X_t=(X_{1,t},\dots,X_{L,t})$, $t\ge0$, be an SDMC. Then the following statements are equivalent:
\begin{enumerate}
\item\label{i:null}
 $(X_t)_{t\ge 0}$ is null cooperative.
\item\label{i:ind}
$(X_{1,t})_{t\ge 0},\dots,(X_{L,t})_{t\ge 0}$ are mutually independent.
\item\label{i:marg}
$(X_{1,t})_{t\ge 0},\dots,(X_{L,t})_{t\ge 0}$ are continuous-time Markov chains.
\end{enumerate}
\end{prop}

Cooperativity is defined through monotone orderings of conditional opening and closing rates. These orderings are inherited by the corresponding conditional transition probabilities at every time lag $\delta>0$.

\begin{lemma}\label[lemma]{le:mon:prob}
Let $(X_t)_{t\ge 0}$ be an SDMC with parameters
$\lambda_0,\ldots,\lambda_{L-1}$ and $\mu_1,\ldots,\mu_L$. For
$s\in\{1,\ldots,L\}$ and $\delta>0$, define
\[
p^{\mathrm{o}}_{s-1}(\delta)
 =
 \mathbb P\{X_{i,t+\delta}=1\mid X_{i,t}=0,\ S_t=s-1\},
\quad
 p^{\mathrm{c}}_s(\delta)
 =
 \mathbb P\{X_{i,t+\delta}=0\mid X_{i,t}=1,\ S_t=s\}.
\]
These probabilities do not depend on $i\in\{1,\ldots, L\}$ or $t\ge 0$. If $(X_t)_{t\ge 0}$ is fully positively cooperative, then
$
p^{\mathrm{o}}_0(\delta)<p^{\mathrm{o}}_1(\delta)<\cdots<p^{\mathrm{o}}_{L-1}(\delta)$, $
p^{\mathrm{c}}_1(\delta)> p^{\mathrm{c}}_2(\delta)>\cdots> p^{\mathrm{c}}_L(\delta).$
If $(X_t)_{t\ge0}$ is fully negatively cooperative, then
$p^{\mathrm{o}}_0(\delta)>p^{\mathrm{o}}_1(\delta)>\cdots>p^{\mathrm{o}}_{L-1}(\delta)$, $p^{\mathrm{c}}_1(\delta)< p^{\mathrm{c}}_2(\delta)<\cdots< p^{\mathrm{c}}_L(\delta).$
If $(X_t)_{t\ge0}$ is null cooperative, both sequences $(p^{\mathrm{o}}_s(\delta))_{s=0}^{L-1}$ and $(p^{\mathrm{c}}_s(\delta))_{s = 1}^L$ are constant in $s$.
\end{lemma}

\subsection{Cooperativity index}

Cooperative behavior need not fall into the three regimes of Definition~\ref{full_cooperativity}. We thus
introduce a scalar summary of cooperativity that covers all possible regimes.

\begin{definition}[Cooperativity index]\label{def_coop_index}
The \textit{cooperativity index} $\Lambda:[0,\infty)^{2L}\to[-1,1]$ is
\[
\Lambda(\theta)
:=
\frac{1}{L(L-1)}\sum_{s=1}^{L-1}\sum_{r=s+1}^L
\bigl(\sign(\mu_s-\mu_r)+\sign(\lambda_{r-1}-\lambda_{s-1})\bigr),
\]
where $\theta=(\lambda_0,\dots,\lambda_{L-1},\mu_1,\dots,\mu_L)$, and $\sign(x)$ equals $1$, $0$ or $-1$ according as $x>0$, $x=0$ or $x<0$, respectively.
\end{definition}

\begin{prop}\label[prop]{prop3}
Let $(X_t)_{t\ge 0}$ be an SDMC with parameter vector $\theta$. Then:
\begin{enumerate}
    \item\label{prop3.1} $(X_t)_{t\ge 0}$ is fully positively cooperative if and only if $\Lambda(\theta)=1$.
    \item\label{prop3.2} $(X_t)_{t\ge0}$ is fully negatively cooperative if and only if $\Lambda(\theta)=-1$.
    \item\label{prop3.3} If $(X_t)_{t\ge0}$ is null cooperative, then $\Lambda(\theta)=0$.
\end{enumerate}
\end{prop}

\begin{remark}
The cooperativity index can be expressed through Kendall's rank correlation coefficient $\tau$ \citep{kendall1938}. For observations $D=\{(x_1,y_1),\dots,(x_n,y_n)\}$,
\[
\tau(D)=\frac{2}{n(n-1)}\sum_{i=1}^{n-1}\sum_{j=i+1}^n
\sign\bigl((x_i-x_j)(y_i-y_j)\bigr).
\]
Then $\Lambda(\theta)=({\tau(D_1)-\tau(D_2)})/{2}$ with $D_1=\{(s,\lambda_s)\}_{s =0}^{L-1}$ and $D_2=\{(s,\mu_s)\}_{s= 1}^L$.
\end{remark}

\subsection{Empirical cooperativity index and asymptotics}

Given a consistent estimator $\hat{\theta}_n$ of $\theta$, cooperativity index can be estimated by the plug-in statistic $\Lambda(\hat{\theta}_n)$; such estimators are developed in Section~\ref{s:estimation}. Since $\Lambda(\cdot)$ is piecewise constant and discontinuous at parameter vectors with ties among $\{\lambda_0,\dots,\lambda_{L-1}\}$ or $\{\mu_1,\dots,\mu_L\}$, the plug-in estimator  $\Lambda(\hat{\theta}_n)$ need not be consistent at such points. The following result covers both the regular (no ties) and tied cases.

\begin{prop}\label[prop]{coop_index_convergence}
Let $\hat{\theta}_n=(\hat{\lambda}_0,\dots,\hat{\lambda}_{L-1},\hat{\mu}_1,\dots,\hat{\mu}_L)$ be a consistent estimator for
$\theta_{\circ}=(\lambda_0,\dots,\lambda_{L-1},\mu_1,\dots,\mu_L)\in[0,\infty)^{2L}$, that is,
$\hat{\lambda}_{s-1}\xrightarrow{\mathbb{P}}\lambda_{s-1}$ and 
$\hat{\mu}_s\xrightarrow{\mathbb{P}}\mu_s$ for $s\in\{1,\dots,L\}.$
\begin{enumerate}
    \item\label{prop4.3_case1}
    If $\mu_s\neq\mu_r$ and $\lambda_{s-1}\neq\lambda_{r-1}$ for all distinct $s,r\in\{1,\dots,L\}$, then
    \(
    \mathbb{P}\{\Lambda(\hat{\theta}_n)=\Lambda(\theta_{\circ})\}\to1
    \) as $n\to \infty$.

    \item\label{prop4.3_case2}
    Define $A=A_{\theta_{\circ}}:=\{(s,r):\mu_s\neq \mu_r\}$ and $B=B_{\theta_{\circ}}:=\{(s,r):\lambda_{s-1}\neq \lambda_{r-1}\}.$
    Suppose that, for some sequence $(a_n)_{n}\subset(0,\infty)$ and random vector
    $Z=(Z_1,\dots,Z_{2L})\in\mathbb{R}^{2L}$,
    \(
    a_n(\hat{\theta}_n-\theta_{\circ})\xrightarrow{D} Z
    \),
    \(\mathbb{P}(Z_i=Z_j)=0\) for all \((i,j)\notin B,\ i<j,\) and 
\(\mathbb{P}(Z_{L+i}=Z_{L+j})=0\)  for all \( (i,j)\notin A,\ i<j.\)
    Then $\Lambda(\hat{\theta}_n)$ converges weakly to
    \begin{multline}\label{e:limit}
    \frac{1}{L(L-1)}\Biggl[
    \sum_{{(s,r)\notin A,\,s<r}}
        \sign(Z_{L+s}-Z_{L+r})
    +\sum_{{(s,r)\notin B,\, s<r}}
        \sign(Z_{r}-Z_{s}) \\
    +\sum_{{(s,r)\in A,\, s<r}}
        \sign(\mu_s-\mu_r)
    +\sum_{{(s,r)\in B,\, s<r}}
        \sign(\lambda_{r-1}-\lambda_{s-1})
    \Biggr].
    \end{multline}
\end{enumerate}
\end{prop}

In the independent (i.e.\ fully tied) case, the last two sums in \eqref{e:limit} vanish, so the limit in Proposition~\ref{coop_index_convergence} is purely random. In the regular case with no ties, the first two sums in \eqref{e:limit} vanish, and weak convergence reduces to convergence in probability to a deterministic limit.

To handle discontinuity points, let $a_n\searrow 0$ and define the thresholded cooperativity index
\begin{equation}\label{coop_index_adjusted}
\Lambda_n(\theta)
=
\frac{1}{L(L-1)}\sum_{s=1}^{L-1}\sum_{r=s+1}^L
\Bigl(
\soft_{a_n}(\mu_s-\mu_r)+\soft_{a_n}(\lambda_{r-1}-\lambda_{s-1})
\Bigr),
\end{equation}
where $\soft_a(x):=\sign(x)\cdot\mathbb{I}\{|x|>a\}$. We call $\Lambda_n(\hat{\theta}_n)$ the \emph{empirical cooperativity index}. It is consistent when the estimation error is asymptotically negligible relative to the threshold $a_n$.

\begin{theorem}[Pointwise convergence]\label{consistency_ind}
Assume that, for some $a_n\searrow 0$,
$\hat{\mu}_s-\mu_s=o_{\mathbb{P}}(a_n)$ and $\hat{\lambda}_{s-1}-\lambda_{s-1}=o_{\mathbb{P}}(a_n)$ for all $s\in\{1,\dots,L\}$. Then, \(
\Lambda_n(\hat{\theta}_n)\xrightarrow{\mathbb{P}}\Lambda(\theta)\),  as \( n\to\infty\), where $\Lambda_n(\cdot)$ is defined in \eqref{coop_index_adjusted} and $\hat\theta = (\hat\lambda_0,\ldots,\hat\theta_{L-1},\hat\mu_1,\ldots,\hat\mu_L)$. \end{theorem}

If $\hat{\theta}_n$ converges at the parametric rate, $\hat{\theta}_n=\theta+\mathcal{O}_{\mathbb{P}}(n^{-1/2})$, one may choose $a_n\searrow 0$ with $a_n^{-1}/\sqrt{n}\to 0$, so that $a_n$ vanishes more slowly than $n^{-1/2}$, for example, $a_n=n^{-1/2}\log(n)$.

\subsection{Other notions of cooperativity}\label{comparison_cooperativity}

In the discrete-time setting, \cite{VanEt24} define several notions of cooperativity, including cooperative/competitive and $i$-cooperative/$i$-competitive behavior, $i\in\{0,\ldots,L-1\}$, based on selected pairwise comparisons of the conditional transition probabilities $p^{\mathrm{o}}_s(\delta)$ and $p^{\mathrm{c}}_s(\delta)$ from Lemma~\ref{le:mon:prob}, for a fixed sampling interval $\delta>0$. By Lemma~\ref{le:mon:prob}, these orderings agree with those of the rate parameters $\lambda_s$ and $\mu_s$. In particular, when $L=2$, our notions of full positive and full negative cooperativity coincide with their notions of cooperativeness and competitiveness, respectively. The main difference is the scope of interaction captured. Our index summarizes the overall interaction structure among all coordinate processes, whereas the notions in \citet{VanEt24} emphasize selected configurations, typically when the sum process $(S_t)_{t\ge0}$ equals $0$ or a specified level $i$.

One could alternatively try to measure cooperativity from the marginal behavior of the sum process, motivated by the intuition that positively cooperative coordinates tend to evolve similarly. This suggests using the invariant distribution $\pi^\star$ of $(S_t)_{t\ge0}$. However, Proposition~\ref{bd_proc} gives
\[
\frac{\pi^\star_{j}}{\pi^\star_{i}}
= \prod_{k = i}^{j-1}\frac{(L-k)\lambda_k}{(k+1)\mu_{k+1}},
\quad 0 \le i < j \le L.
\]
Thus, $\pi^\star$ depends on the SDMC parameters only through ratios of aggregated birth and death rates. Cooperativity, by contrast, is characterized by pairwise relations among the conditional opening rates $\lambda_i$ and $\lambda_j$ and closing rates $\mu_i$ and $\mu_j$. Hence, $\pi^\star$ alone cannot capture the full range of cooperative structures allowed by SDMCs.

Dwell time distributions are also central in biophysical applications, particularly ion channel dynamics. Let the jump times of $(S_t)_{t\ge0}$ be $\tau_{n}:=\inf\{t\in [\tau_{n-1},\infty):S_t\neq S_{\tau_{n-1}}\},$
 $n\in \mathbb{N},$ and $\tau_0:=0,$
with the convention that the infimum of the empty set is $\infty$, and define the $n$-th holding time by
\(D_n:=\tau_{n+1}-\tau_n,\)
By the Markov property,
\[
D_n\mid S_{\tau_n}=s\;
\sim\;
\operatorname{Exp}\!\bigl(
(L-s)\lambda_s\,\mathbb{I}(s\in\{0,\dots,L-1\})
+s\mu_s\,\mathbb{I}(s\in\{1,\dots,L\})
\bigr).
\]
Thus, dwell times depend on the SDMC parameters only through the total rate of leaving each state, namely the scaled sum of the opening and closing rates. The full cooperativity structure of $(X_t)_{t\ge0}$ thus cannot be recovered from dwell times alone.

Consequently, the information-theoretic metrics of \citet{wawrzkiewicz2026information}, which rely on either marginal state or dwell time distributions, are not adequate for quantifying cooperativity among channels.

\section{Parameter estimation in hidden Markov models}
\label{s:estimation}

In many applications, the SDMC governing the system dynamics is not observable, and only a discrete, noisy signal related to the sum process is available. Thus neither the individual state transitions of the latent process nor the exact trajectory of the sum process can be recovered without uncertainty. This motivates a hidden Markov model in which the continuous-time sum process is latent and discrete-time observations arise through a noisy emission mechanism.

\subsection{Sum-dependent hidden Markov models}

\begin{model}[Sum-dependent hidden Markov model, SD-HMM]\label[model]{model:SDHMM}
Let $(X_t)_{t\ge0}$ be an~SDMC and $S_t=\s{X_t}$, $t\ge0$, its sum process, both evolving in continuous time and unobserved. Assume that the model parameters are strictly positive,
\begin{equation}\label{e:strict:positive}
\theta=(\lambda_0,\ldots,\lambda_{L-1},\mu_1,\ldots,\mu_L)
\in\Theta:=(0,\infty)^{2L},
\end{equation}
and an initial distribution $\pi_\theta$ for $S_0$ with full support,
$\pi_\theta(j)=\mathbb P(S_0=j)>0,$ $j\in \{0,\ldots,L\}.$
Real-valued observations $Y_k$ are collected at $t_k=(k-1)\delta$, $k\in\{1,\ldots,n\}$, for a fixed $\delta>0$. Conditional on $S_{t_k}=j$, $Y_k$ has density $g_\phi(\cdot\mid j)$ with respect to a $\sigma$-finite measure $\nu$ on $\mathbb R$, where $\phi\in\Phi\subseteq\mathbb R^d$. The observations are conditionally independent given the latent states,
\begin{equation}\label{HMM_model}
\mathbb P(Y_1,\ldots,Y_n\mid S_{t_1},\ldots,S_{t_n})
=\prod_{k=1}^n \mathbb P(Y_k\mid S_{t_k}).
\end{equation}
We call this model the \emph{sum-dependent hidden Markov model} (SD-HMM), with full parameter $\eta=(\theta,\phi)\in\Theta\times\Phi\subseteq\mathbb R^{2L+d}$.
\end{model}

Inference from discrete-time observations requires identifiability of the continuous-time rate matrix from the corresponding transition matrix. This is related to the classical embedding problem \citep{elfving1937theorie}, but the present setting is identifiable due to reversibility \citep{Jia2016} and the specific structure of SDMCs.

\begin{lemma}[Unique embedding]\label[lemma]{identifiability}
Let $(X_t)_{t\ge0}$ be an SDMC satisfying \eqref{e:strict:positive}. Then, for every $\delta>0$, the map
$\theta\mapsto e^{\delta R(\theta)}$,  $\theta\in\Theta$, is injective, with $R(\theta)$ the rate matrix of $S_t=\s{X_t}$ in~\eqref{Theorem2_eq}.
\end{lemma}

\subsection{Estimation of model parameters}\label{ss:mle}

We estimate the parameters in SD-HMM (Model~\ref{model:SDHMM}) by maximum likelihood. By \eqref{HMM_model} and the Markov property of $(S_t)_{t\ge0}$, the likelihood of $(Y_1,\ldots,Y_n)$ under $\eta=(\theta,\phi)$ is
\begin{equation}\label{e:likelihood}
p_\eta(y_1,\ldots,y_n)
=\sum_{s_1=0}^L\cdots \sum_{s_n=0}^L
\left\{\pi_\theta(s_1)
\prod_{k=1}^{n-1}p_{s_k,s_{k+1}}(\theta)
\prod_{k=1}^n g_\phi(y_k\mid s_k)\right\}.
\end{equation}
The \emph{maximum likelihood estimator} is any maximizer
\begin{equation}\label{e:mle}
\hat \eta_n(y_1,\ldots,y_n)
\in \mathop{\arg \max}_{\eta\in\Theta\times\Phi}
p_\eta(y_1,\ldots,y_n).
\end{equation}
Let the true parameter of the SD-HMM be
\(
\eta_{\circ}=(\theta_{\circ},\phi_{\circ})\in\Theta^\circ\times\Phi^\circ,
\)
the interior of $\Theta \times \Phi$.
We impose the following conditions on the emission distributions.

\begin{assumption}[Emission distributions]\label{a:em}
\begin{enumerate}
\item  \label{HMM_consist_1}
For any weights $a_0,\ldots,a_L,a_0',\ldots,a_L'\ge 0$ with
$\sum_{i=0}^L a_i=\sum_{i=0}^L a_i'=1$, and any $\phi,\phi'\in \Phi^\circ$, if 
$\sum_{i=0}^{L} a_i\, g_\phi(y\mid i)
=
\sum_{i=0}^{L} a_i'\, g_{\phi'}(y\mid i)
$ for $\nu$-almost every $y\in\mathbb R$, then $(a_0,\ldots,a_L,\phi)=(a_0',\ldots,a_L',\phi')$.

\item \label{HMM_consist_2}
For every $y\in\mathbb R$ and $i\in\{0,\ldots,L\}$, $\phi\mapsto g_\phi(y\mid i)$ is continuous. If $\Phi$ is unbounded, then
\(
\lim_{\|\phi\|\to\infty} g_\phi(y\mid i)=0.
\)

\item \label{HMM_consist_3}
For every $i\in\{0,\ldots,L\}$,
\(
\limsup_{n\to\infty}
\mathbb E_{\eta_{\circ}}\!\left[\lvert\log g_{\phi_{\circ}}(Y_n\mid i)\rvert\right]<\infty.
\)

\item \label{HMM_consist_4}
For every $\phi\in\Phi$, there exists $\delta>0$ such that
\[
\limsup_{n\to\infty}\mathbb E_{\eta_{\circ}}\!\left[
\sup_{\substack{\phi'\in\Phi:\,\|\phi'-\phi\|<\delta\\ i\in\{0,\ldots,L\}}}
\max\!\left\{\log g_{\phi'}(Y_n\mid i),0\right\}
\right]<\infty.
\]
\end{enumerate}
\end{assumption}

Assumption~\ref{a:em}\ref{HMM_consist_1} is a slightly strengthened form of Condition~2 in \cite{LEROUX1992}, yielding full rather than equivalence-class identifiability of the emission parameters (Proposition~\ref{HMM_ident_prop}). Assumptions~\ref{a:em}\ref{HMM_consist_2}--\ref{HMM_consist_4} are standard regularity conditions controlling the likelihood and its stochastic fluctuations. Assumption~\ref{a:em} holds for common emission models, including Gaussian and more general exponential-family emissions.

\begin{theorem}[Consistency] \label{HMM_consistency}
Consider the SD-HMM (Model~\ref{model:SDHMM}) with parameter $\eta_{\circ}=(\theta_{\circ},\phi_{\circ})\in\Theta^\circ\times\Phi^\circ$, and let Assumption~\ref{a:em} hold. Then the maximum likelihood estimator in \eqref{e:mle} satisfies $\hat{\eta}_n\xrightarrow{\mathrm{a.s.}}\eta_{\circ}$ as $n\to\infty.$
\end{theorem}

We next establish asymptotic normality. Following \citet{Bickel98}, define the Fisher information matrix at $\eta$ by
\begin{equation}\label{e:fisher}
\mathcal I_\eta=\mathbb E_\eta[\zeta_1\zeta_1^\intercal],
\end{equation}
where
\(
\zeta_1=\lim_{n\to\infty}\nabla_\eta
\log p_\eta(Y_n\mid Y_1,\ldots,Y_{n-1}),
\)
and $p_\eta(Y_n\mid Y_1,\ldots,Y_{n-1})$ is the conditional density under $\eta\in\Theta\times\Phi$. We assume the following local regularity.

\begin{assumption}[Local regularity of emission distributions]\label{a:reg:em2}
There exists $\varepsilon>0$ such that $\{\phi:\norm{\phi-\phi_{\circ}}<\varepsilon\}\subseteq\Phi$ and the following conditions hold:
\begin{enumerate}
\item \label{HMM_consist_6}
For all $y\in\mathbb R$ and $s\in\{0,\ldots,L\}$, $\phi\mapsto g_\phi(y\mid s)$ is twice continuously differentiable whenever $\norm{\phi-\phi_{\circ}}<\varepsilon$.

\item \label{HMM_consist_7}
For all $s\in\{0,\ldots,L\}$ and $k,k'\in\{1,\ldots,d\}$,
\[
\mathbb E_{\eta_{\circ}}\!\left[
\sup_{\phi:\norm{\phi-\phi_{\circ}}<\varepsilon}
\left|\frac{\partial}{\partial\phi_k}\log g_\phi(Y_1\mid s)\right|^2
\right]<\infty,\quad
\mathbb E_{\eta_{\circ}}\!\left[
\sup_{\phi:\norm{\phi-\phi_{\circ}}<\varepsilon}
\left|\frac{\partial^2}{\partial\phi_k\partial\phi_{k'}}
\log g_\phi(Y_1\mid s)\right|
\right]<\infty,
\]
and, for $l\in\{1,2\}$, any $1\le k_a\le d$, $a=1,\ldots,l$, and all $s\in\{0,\ldots,L\}$,
\[
\int_{\mathbb R}
\sup_{\phi:\norm{\phi-\phi_{\circ}}<\varepsilon}
\left|
\frac{\partial^l}{\partial\phi_{k_1}\cdots\partial\phi_{k_l}}
g_\phi(y\mid s)
\right|\,\mathrm dy<\infty.
\]

\item \label{HMM_consist_8}
For all $s\in\{0,\ldots,L\}$, with $0/0=0$,
\[
\mathbb P_{\eta_{\circ}}\!\left(
\sup_{\phi:\norm{\phi-\phi_{\circ}}<\varepsilon}
\max_{0\le i,j\le L}
\frac{g_\phi(Y_1\mid i)}{g_\phi(Y_1\mid j)}
=\infty
\;\bigg|\; S_{t_1}=s
\right)<1.
\]
\end{enumerate}
\end{assumption}

Assumptions~\ref{a:reg:em2}\ref{HMM_consist_6}--\ref{HMM_consist_7}~are standard regularity conditions typical for the asymptotic analysis of maximum likelihood estimators. These conditions concern only the emission densities, as differentiability of the transition probabilities and the invariant distribution with respect to parameters follows under the SD-HMM. Assumption~\ref{a:reg:em2}\ref{HMM_consist_8}~rules out the case where the supports of $g_\phi(y\mid 0),\dots, g_\phi(y\mid L)$ are disjoint. This technical restriction could be relaxed, as disjoint supports increase information in distinguishing latent states, see \citet{Bickel98}.

\begin{theorem}[Asymptotic normality]
\label[theorem]{HMM_asymptotic_normal}
Under the SD-HMM (\Cref{model:SDHMM}) with parameter $\eta_{\circ}=(\theta_{\circ},\phi_{\circ})\in\Theta^\circ\times\Phi^\circ$, assume that $(X_t)_{t\ge0}$, equivalently $S_t=\s{X_t}$, is stationary. Let Assumption~\ref{a:reg:em2}~hold and let $\hat\eta_n$ be a maximum likelihood estimator satisfying $\hat\eta_n\xrightarrow{\mathrm{a.s.}}\eta_\circ$. If $\mathcal I_{\eta_{\circ}}$ in \eqref{e:fisher} is nonsingular, then $\sqrt n(\hat\eta_n-\eta_{\circ})
\xrightarrow{D}\mathcal N(0,\mathcal I_{\eta_{\circ}}^{-1})$ as $ n\to\infty.$
\end{theorem}

\begin{remark}
The stationarity assumption in Theorem~\ref{HMM_asymptotic_normal} can be relaxed. Let $\Theta\times\Phi$ be compact, which is mild given strong consistency of $\hat\eta_n$, and allow an initial distribution $\pi=\delta_{s_0}$ for arbitrary $s_0\in\{0,\ldots,L\}$. Suppose also that: (i) $\phi\mapsto g_\phi(y\mid s)$ is continuous on $\Phi$ for all $y\in\mathbb R$ and $s\in\{0,\ldots,L\}$; (ii) for every $y\in\mathbb R$ there exists $s$ such that $g_\phi(y\mid s)>0$ for all $\phi\in\Phi$; and (iii), for each $s$,
$\sup_{\phi\in\Phi,\,y\in\mathbb R} g_\phi(y\mid s)<\infty$ and $\mathbb E_{\eta_{\circ}}\!\left[
\abs{\log\left(\inf_{\phi\in\Phi}\sum_{s=0}^L g_\phi(Y_0\mid s)\right)}
\right]<\infty.$
Together with Assumptions~\ref{a:reg:em2}\ref{HMM_consist_6}--\ref{a:reg:em2}\ref{HMM_consist_7}, these conditions imply asymptotic normality of $\hat\eta_n$, see \citet[Theorem~12.5.7]{CappeMoulinesRyden2005} and \cite{DoMa01}.
\end{remark}

We next give an example that will be revisited in the data analysis of Section~\ref{section8}.

\begin{example}[Voltage-clamp measurements of ion channels]
\label{HMM_example2}
Hidden Markov models are standard for ion channel data (\citealp{FredkinRice1992}). The number of open channels $S_t$ is not directly observable, whereas a noisy conductance signal is recorded. For a sampling interval $\delta>0$ and $t_k=(k-1)\delta$, we observe $Y_k=b+aS_{t_k}+\varepsilon_k,$ $k=1,\ldots,n$. Here $b\in\mathbb R$ is a baseline offset, $a>0$ is the single channel conductance, and the errors $(\varepsilon_k)_{k=1}^n$ are conditionally independent given $(S_{t_k})_{k=1}^n$, with the conditional distribution of $\varepsilon_k$ depending on the latent path only through $S_{t_k}$. To model open-channel noise \citep{Sigworth1985}, this distribution may depend on $S_{t_k}$, for example,
\(
\varepsilon_k\mid S_{t_k}\sim\mathcal N(0,\sigma_{S_{t_k}}^2).
\)
Then the emission parameter is $\phi=(b,a,\sigma_0,\ldots,\sigma_L)$, with $a\neq0$ and $\sigma_0,\ldots,\sigma_L>0$. The conditions of Theorems~\ref{HMM_consistency} and~\ref{HMM_asymptotic_normal} hold for this model, see Supplement~\ref{verify}.
\end{example}

The SD-HMM can be extended by relaxing the parametric emission model and using nonparametric alternatives, to improves robustness to diverse noise structures. One approach characterizes the latent sum process through emission medians rather than means, as in \citet{Requadt2025} for a related discrete-time coupled Markov model. A full treatment is beyond the scope of this paper, but presents a promising direction for future research. 

\subsection{Testing of cooperativity}\label{sec:testing}

We further develop a test for cooperativity (\Cref{cooperativity}) within the SD-HMM (\Cref{model:SDHMM}). Let
\(
\theta_{\circ}=(\theta_{\circ,i})_{i=1}^{2L}
=(\lambda_0,\ldots,\lambda_{L-1},\mu_1,\ldots,\mu_L)
\)
be the parameter of the hidden SDMC. We test the pairwise hypotheses $H_{0,(i,j)}:\theta_{\circ,i}=\theta_{\circ,j}$ versus $H_{1,(i,j)}:\theta_{\circ,i}\neq\theta_{\circ,j}$ for 
\begin{equation}\label{eq:index:hyp}
(i,j)\in\mathcal H
:=\{(i,j):1\le j<i\le L\text{ or }L+1\le i<j\le2L\}.
\end{equation}
The intersection of all null hypotheses corresponds to independence of the individual coordinate processes of the hidden SDMC.

Let $\hat\theta_n$ be the first $2L$ components of the maximum likelihood estimator $\hat\eta_n$ in \eqref{e:mle}. By \Cref{HMM_asymptotic_normal},
\(
\sqrt n(\hat\theta_n-\theta_{\circ})
\xrightarrow{D}
\mathcal N\!\left(0,\,
(\mathcal I_{\eta_\circ}^{-1})_{\{1,\ldots,2L\},\{1,\ldots,2L\}}\right).
\)
The limiting covariance matrix is unknown, but can be consistently estimated. Following \citet{Bickel98}, define
\[
\hat{\mathcal I}_n
=-\frac1n\nabla_\eta^2\log p_\eta(Y_1,\ldots,Y_n)\big|_{\eta=\hat\eta_n},
\qquad
\hat\Sigma_n
=\bigl(\hat{\mathcal I}_n^\dagger\bigr)_{\{1,\ldots,2L\},\{1,\ldots,2L\}},
\]
where $\dagger$ denotes the Moore--Penrose pseudoinverse. Under the assumptions of \Cref{HMM_asymptotic_normal}, $\hat{\mathcal I}_n\xrightarrow{\mathbb{P}}\mathcal I_{\eta_\circ}$, which is nonsingular. Thus, with probability tending to one, $\hat{\mathcal I}_n^\dagger=\hat{\mathcal I}_n^{-1}$ and $\hat\Sigma_n$ consistently estimates
\(
\Sigma_\circ:=
(\mathcal I_{\eta_\circ}^{-1})_{\{1,\ldots,2L\},\{1,\ldots,2L\}}.
\)
For $(i,j)\in\mathcal H$, define the test statistic $\hat T_{n,(i,j)}={\sqrt n(\hat\theta_{n,i}-\hat\theta_{n,j})}/{\hat s_{n,(i,j)}}$ whenever $ s_{n,(i,j)}>0$, and set to zero otherwise. Here  $\hat s_{n,(i,j)}
=
\sqrt{(\hat\Sigma_n)_{i,i}-2(\hat\Sigma_n)_{i,j}+(\hat\Sigma_n)_{j,j}}$.  
To tackle multiplicity, we adopt the general stepdown procedure of \citet{RoWo05}. For nonempty $\mathcal K\subseteq\mathcal H$, let $(\xi_k)_{k=1}^{2L}\sim\mathcal N(0,\hat\Sigma_n)$ and define
\begin{equation}\label{eq:thd:joint}
\hat c_{n,\alpha}(\mathcal K)=\inf\left\{t\ge 0:\prob{\max_{(i,j)\in\mathcal K,\, \hat s_{n,(i,j)}>0}\frac{\lvert \xi_i - \xi_j\rvert}{\hat s_{n,(i,j)}} \le t }\ge 1-\alpha\right\}.
\end{equation} 
The map $\mathcal K\mapsto\hat c_{n,\alpha}(\mathcal K)$ is monotone:
$\hat c_{n,\alpha}(\mathcal K_1)\le \hat c_{n,\alpha}(\mathcal K_2)$ whenever
$\emptyset\neq\mathcal K_1\subseteq\mathcal K_2\subseteq\mathcal H$. This ensures asymptotic family-wise error-rate control in the stepdown framework. We call the resulting procedure the Stepdown Cooperativity Test (SCoT), see \Cref{alg:stepdown}.

\begin{algorithm}
\caption{SCoT: Stepdown Cooperativity Test}
\label{alg:stepdown}
\begin{algorithmic}[1]
\Require Test statistics $\hat{T}_{n,(i,j)}$ for $(i,j)\in \mathcal{H}$ in \eqref{eq:index:hyp}, significance level $\alpha \in (0,1)$
\State Initialize $\widehat{\mathcal H}_{1,n} \gets \emptyset$, \; $\mathcal K \gets \mathcal H$
\While{$\mathcal K \neq \emptyset$}
    \State Compute the critical value $\hat c_{n,\alpha}(\mathcal K)$ from \eqref{eq:thd:joint}
    \State Update $\mathcal R \gets \left\{(i,j)\in \mathcal K : \bigl|\hat{T}_{n,(i,j)} \bigr|> \hat c_{n,\alpha}(\mathcal K)\right\}$
    \If{$\mathcal R = \emptyset$}
        \State \textbf{break}
    \Else
        \State Update $\widehat{\mathcal H}_{1,n} \gets \widehat{\mathcal H}_{1,n} \cup \mathcal R$, \;  $\mathcal K \gets \mathcal K \setminus \mathcal R$
    \EndIf
\EndWhile
\State \Return the rejection set  $\widehat{\mathcal H}_{1,n}$
\end{algorithmic}
\end{algorithm}

\begin{theorem}[Strong control of the family-wise error rate]\label{hypothesis_test_main_thm}
Fix $\alpha\in(0,1)$ and define
$\mathcal H_0=\{(i,j)\in\mathcal H:\theta_{\circ,i}=\theta_{\circ,j}\}$, $\mathcal H_1=\{(i,j)\in\mathcal H:\theta_{\circ,i}\neq\theta_{\circ,j}\}.$
Under the assumptions of \Cref{HMM_asymptotic_normal}, the rejection set $\widehat{\mathcal H}_{1,n}$ returned by \Cref{alg:stepdown} satisfies:
\begin{enumerate}
\item \label{test_proof1}
$\displaystyle \lim_{n\to\infty}
\mathbb P\{\mathcal H_1\subseteq\widehat{\mathcal H}_{1,n}\}=1$, and
\item \label{test_proof2}
$\displaystyle \limsup_{n\to\infty}
\mathbb P\{\mathcal H_0\cap\widehat{\mathcal H}_{1,n}\neq\emptyset\}\le\alpha$.
\end{enumerate}
\end{theorem}

The SCoT procedure yields a test-based estimator of the cooperativity index (\Cref{def_coop_index}).

\begin{cor}\label[cor]{test_proof3}
For $\alpha\in(0,1)$, define
\begin{equation}\label{e:estviatest}
\hat\Lambda_{n,\alpha}
=
\frac{1}{L(L-1)}
\sum_{(i,j)\in\mathcal H}
\mathbb I\{(i,j)\in\widehat{\mathcal H}_{1,n}\}
\,\mathrm{sign}(\hat\theta_{n,i}-\hat\theta_{n,j}).
\end{equation}
Let $\Lambda(\theta_\circ)$ be the true cooperativity index. Under the assumptions of \Cref{hypothesis_test_main_thm}, it holds that
$\liminf_{n\to\infty}
\mathbb P\{\hat\Lambda_{n,\alpha}=\Lambda(\theta_\circ)\}
\ge1-\alpha.$
Moreover, if $\mathcal H_0=\emptyset$, then
\(
\lim_{n\to\infty}
\mathbb P\{\hat\Lambda_{n,\alpha}=\Lambda(\theta_\circ)\}=1.
\)
\end{cor}

Recall that $ \Lambda(\theta_\circ) = 0$ corresponds either to independence or to cancellation of positive and negative dependencies (cf.~\Cref{prop_rel_null_coop_ind}). The rejection set resolves this ambiguity. Estimated independence corresponds to
\(
\hat\Lambda_{n,\alpha} = 0\) and 
\(
\widehat{\mathcal H}_{1,n} = \emptyset,
\)
whereas
\(
\hat\Lambda_{n,\alpha} = 0
\) and  \(
\widehat{\mathcal H}_{1,n} \neq \emptyset
\)
indicates cancellation effects. This distinction is unavailable from the point estimator in \eqref{coop_index_adjusted}.

\section{Simulations and data applications}\label{section8}

We assess the finite-sample performance of the proposed methodology through simulations and applications to voltage-clamp recordings of ion channels.

\subsection{Asymptotic normality}\label{section8.3}

We first examine the weak convergence of the maximum likelihood estimator from \Cref{ss:mle} under the SD-HMM (\Cref{model:SDHMM}). Observations are generated as in Example~\ref{HMM_example2}, with emission parameters $b=0$, $a=1$, and $\sigma_0=\cdots=\sigma_L=0.1$. We set $L=2$ and let $(S_t)_{t\ge0}$ be the sum process induced by an SDMC (\Cref{def_rate}) with
\(
\theta_\circ=(3,4,4,3).
\)
The sampling interval is fixed at $\delta=0.05$, with observation times $t_k=(k-1)\delta$, $k=1,\ldots,n$. Thus the true joint parameter is
\(
\eta_\circ
=(\lambda_0,\lambda_1,\mu_1,\mu_2,b,a,\sigma_0,\sigma_1,\sigma_2)
=(3,4,4,3,0,1,0.1,0.1,0.1)
\in(0,\infty)^9 .
\)
By \Cref{HMM_asymptotic_normal},
\(
\sqrt{n}(\hat{\eta}_n-\eta_{\circ})
\xrightarrow{D}
\mathcal{N}(0,\mathcal{I}^{-1}_{\eta_{\circ}}).
\)
We evaluate this asymptotic behavior using $5000$ Monte Carlo repetitions for each $n\in\{50,100,1000\}$. \Cref{fig:MLE_sim} shows QQ-plots of 
\[
\sqrt{
\frac{n}{(\mathcal{I}^{-1}_{\eta_{\circ}})_{j,j}}
}
(\hat{\eta}_{n,j}-\eta_{\circ,j})
\]
against the standard normal distribution. The empirical distributions move toward normality as $n$ increases, consistent with \Cref{HMM_asymptotic_normal}.

\begin{figure}
    \centering
    \includegraphics[width=0.9\linewidth]{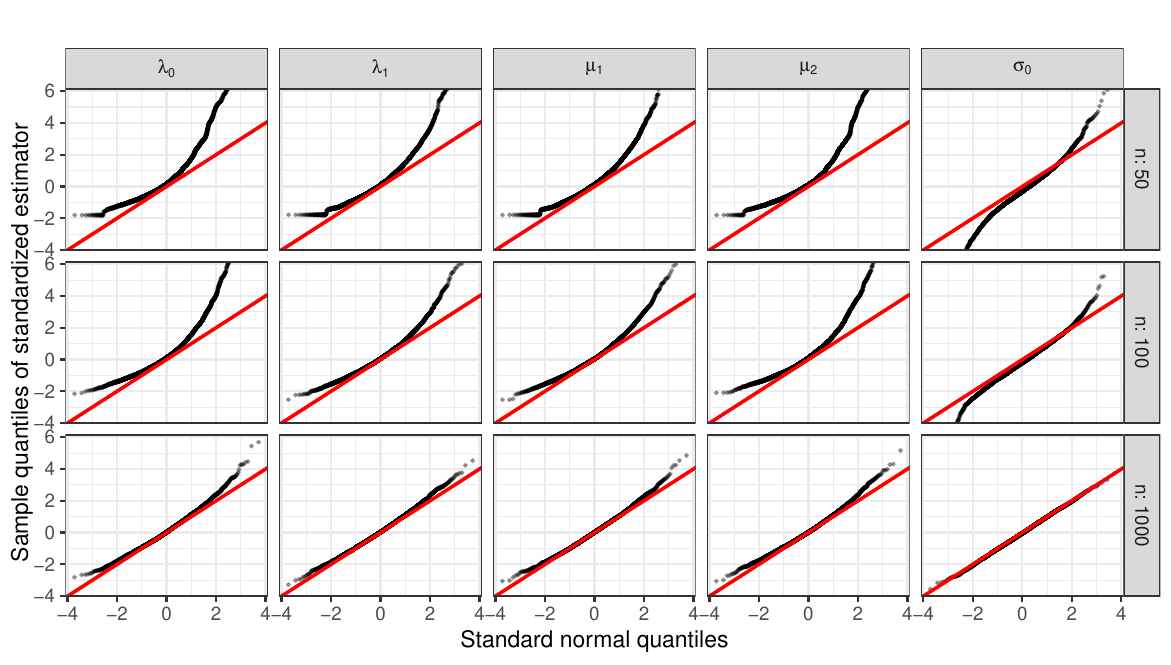}
    \caption{QQ-plots of the standardized estimators
    $n^{1/2}(\hat \eta_{n,j}-\eta_{\circ,j}) / (\mathcal{I}^{-1}_{\eta_\circ})_{j,j}^{1/2}$,
    based on $5000$ Monte Carlo repetitions, compared with the standard normal distribution.
    To save space, we show only $j\in\{1,2,3,4,7\}$ from left to right, corresponding to
    $\lambda_0,\lambda_1,\mu_1,\mu_2,\sigma_0$, respectively.
    Rows correspond to $n\in\{50,100,1000\}$ from top to bottom.
    Values outside $[-4,6]$ are omitted for visual clarity; this truncation does not affect the qualitative assessment of convergence.}
    \label{fig:MLE_sim}
\end{figure}

\subsection{Inference of cooperativity}\label{simulations}

We next evaluate estimation and testing of cooperativity. Observations are again generated as in Example~\ref{HMM_example2}, with Gaussian emissions satisfying $b=0$, $a=1$, and $\sigma_0=\cdots=\sigma_L=0.1$, at times $t_k=(k-1)\delta$, $k\in\{1,\ldots,n\}$. The hidden process is the sum process of an SDMC with $L=3$ coordinates. We consider three scenarios representing different types of cooperativity:
\begin{itemize}
    \item fully positive cooperativity, with $\theta=(1,5,9,9,5,1)$;
    \item null cooperativity, or equivalently independence, with $\theta=(5,5,5,5,5,5)$;
    \item fully negative cooperativity, with $\theta=(9,5,1,1,5,9)$.
\end{itemize}
To assess the influence of the sampling interval, we take
\(
\delta\in\{2^{-5},2^{-6},\ldots,2^{-10}\}.
\)
The observation horizon is fixed at $t_{\max}=15$, so that 
\(
n=n_\delta:=\lfloor 15/\delta\rfloor+1.
\)

For each scenario, we compute the maximum likelihood estimator \(\hat{\eta}_n=(\hat{\theta}_n,\hat{\phi}_n)\) from \Cref{ss:mle}, and estimate the cooperativity index by the empirical cooperativity index $\Lambda_n(\hat{\theta}_n)$ in \eqref{coop_index_adjusted} using the truncation threshold $a_n={\log(n)}/{\sqrt n}$. By \Cref{consistency_ind}, this yields a consistent estimator of the cooperativity index (\Cref{def_coop_index}).

We compare the proposed method with the procedures of \cite{VanEt24} and \cite{CHUNG}, denoted by VND and CK, respectively. Since VND estimates conditional transition probabilities, we compute its empirical cooperativity index using these conditional transition probabilities with truncation threshold $a_n\delta$, where the additional factor $\delta$ accounts for the scale difference between transition probabilities in $[0,1]$ and rate parameters in $[0,\infty)$. Since CK does not provide a natural cooperativity index, we use its estimated coupling factor $\omega\in[0,1]$ as a surrogate. Here $\omega=1$ corresponds to a fully coupled model, interpretable as positive cooperativity, whereas $\omega=0$ corresponds to independence.

\Cref{fig:sim2_2} summarizes the results. The proposed method and VND yield nearly identical estimates and are stable across all considered values of $\delta$. This is consistent with the discussion in \Cref{comparison_cooperativity}, which implies that, for small $\delta$, comparisons based on transition probabilities should agree with those based on the rate parameters. Because $L=3$ here, VND is not subject to identifiability issues. By contrast, CK fails to detect negative cooperativity and the more structured positive cooperativity pattern, despite performing well under independence. This reflects the limited flexibility of CK, which models dependence through a single coupling factor $\omega$ interpolating between a fully coupled model and an independent model.

\begin{figure}[ht]
    \centering
    \includegraphics[width=0.9\linewidth]{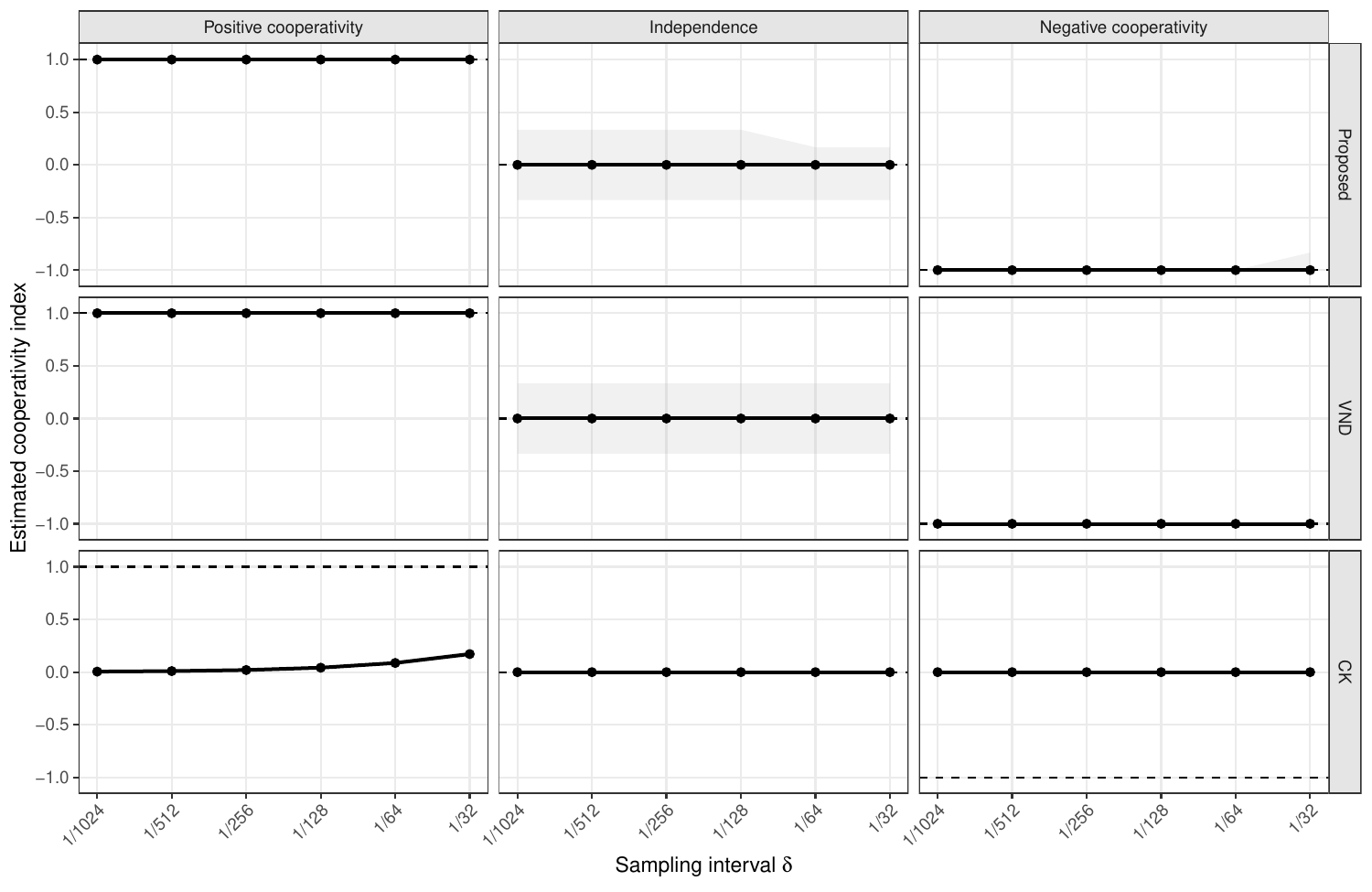}
    \caption{Estimated cooperativity index for the proposed method, VND and CK under fully positive cooperativity, independence and fully negative cooperativity. Solid curves denote medians, and shaded bands denote interquartile ranges, based on $1000$ Monte Carlo repetitions. Dashed horizontal lines mark the true values of $\Lambda(\theta)$.}
    \label{fig:sim2_2}
\end{figure}

We next evaluate the proposed SCoT procedure (\Cref{alg:stepdown}) for testing cooperativity. As benchmarks, we use Holm's method \citep{Holm1979} and the Benjamini--Hochberg procedure \citep{BH95}, both based on marginal $p$-values computed from the asymptotic distribution in \Cref{HMM_asymptotic_normal}. SCoT and Holm control the family-wise error rate (FWER), whereas Benjamini--Hochberg targets the false discovery rate (FDR). Under general dependence between $p$-values, valid FDR control may require additional adjustment, as in \citet{BeYe01}. We nevertheless include the original Benjamini--Hochberg procedure as a high-power benchmark. All nominal error rates are set to $0.05$.

The results are given in \Cref{fig:fwer,fig:power}. Under independence, FDR and FWER coincide, and all procedures provide satisfactory error control, see \Cref{fig:fwer}. Under fully positive and fully negative cooperativity, no null hypothesis is true, and  performance is thus measured by the true positive rate, see \Cref{fig:power}. Benjamini--Hochberg has the highest power, as expected from its less stringent error criterion. SCoT is more powerful than Holm's method, reflecting the benefit of exploiting the dependence structure of the limiting distribution.

\begin{figure}[ht]
    \centering
    \includegraphics[width=0.7\linewidth]{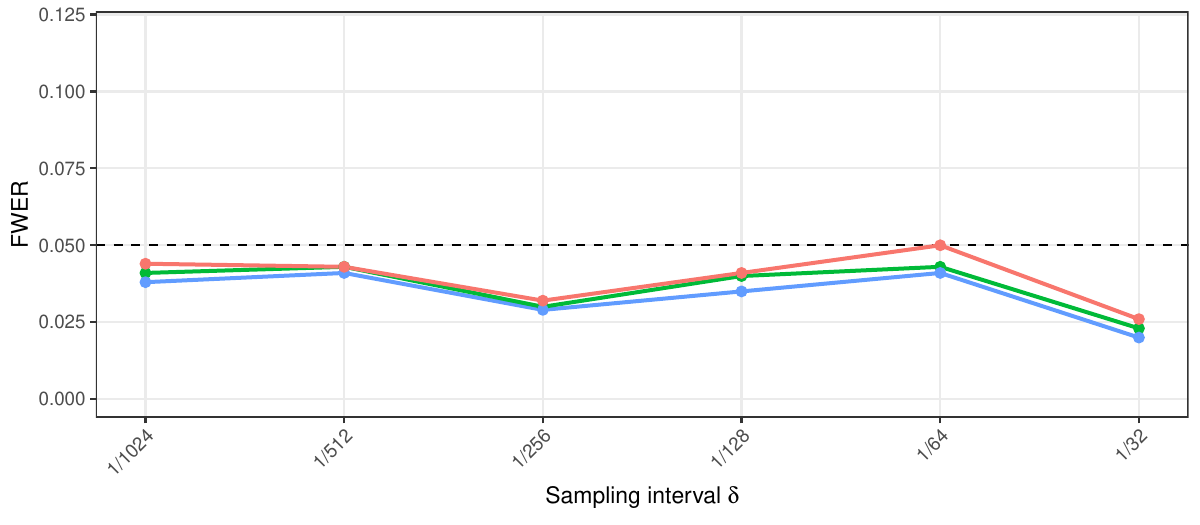}
    \caption{FWER of the proposed SCoT (orange), Holm's method (blue) and the Benjamini--Hochberg procedure (green) under independence, based on $1000$ Monte Carlo repetitions.}
    \label{fig:fwer}
\end{figure}

\begin{figure}[ht]
    \centering
    \includegraphics[width=0.8\linewidth]{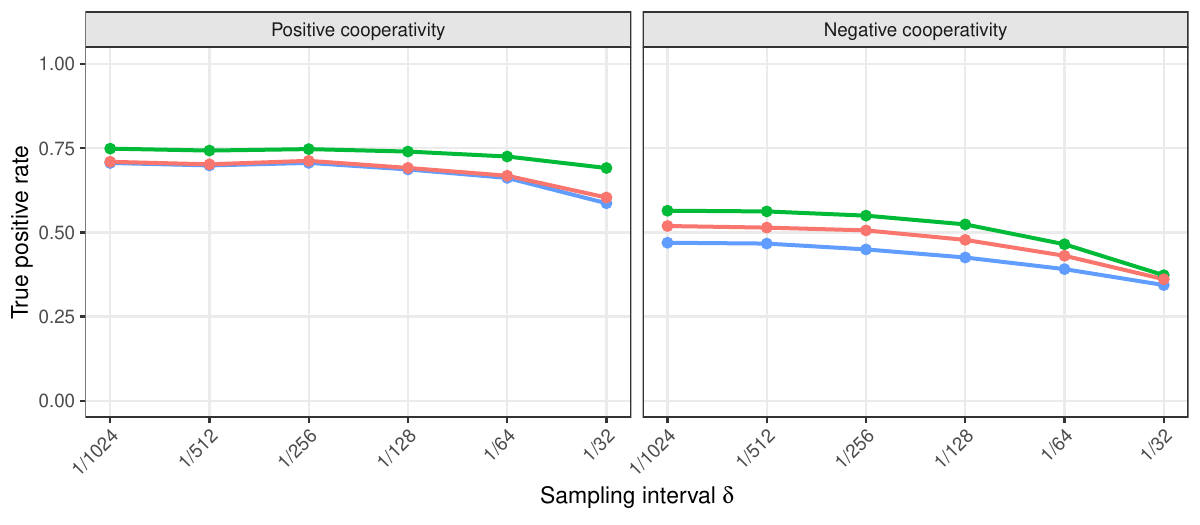}
    \caption{True positive rates of the proposed SCoT (orange), Holm's method (blue) and the Benjamini--Hochberg procedure (green) under fully positive and fully negative cooperativity, based on $1000$ Monte Carlo repetitions.}
    \label{fig:power}
\end{figure}

Finally, \Cref{fig:coop_est} reports the performance of the cooperativity index estimator in \eqref{e:estviatest}. We also include two variants that replace $\widehat{\mathcal H}_{1,n}$ in \eqref{e:estviatest} by the rejection sets from Holm's method and Benjamini--Hochberg. The three estimators perform similarly and are close to the empirical index obtained with a hand-tuned truncation threshold (\Cref{fig:sim2_2}). This supports the theoretical finding in  \Cref{test_proof3}.

\begin{figure}
    \centering
    \includegraphics[width=0.8\linewidth]{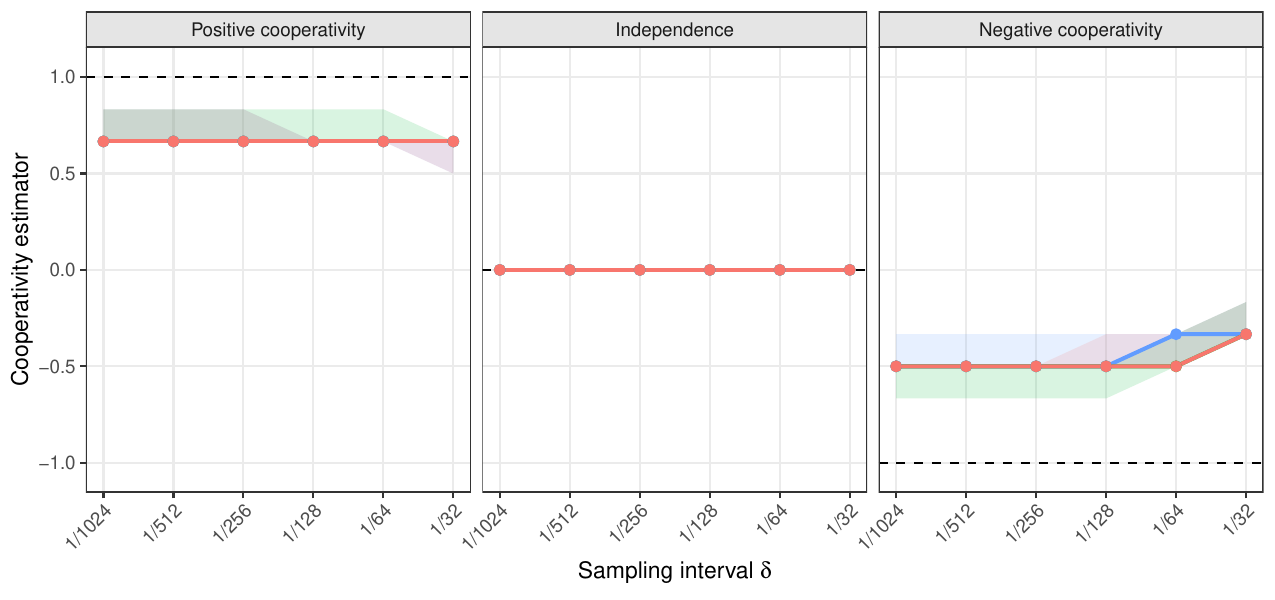}
    \caption{Performance of the cooperativity estimator form \eqref{e:estviatest} using rejection sets from the proposed SCoT (orange), Holm's method (blue) and the Benjamini--Hochberg procedure (green) under fully positive cooperativity, independence and fully negative cooperativity. Solid curves denote medians, and shaded bands denote interquartile ranges, based on $1000$ Monte Carlo repetitions. Dashed horizontal lines mark the true values of $\Lambda(\theta)$.}
    \label{fig:coop_est}
\end{figure}

\subsection{Ion channel data applications}\label{section8.4}

We apply the proposed methodology to voltage-clamp recordings of ion channels in artificial lipid membranes. Such experiments record current traces generated by ion transport through multiple channels, and further background is given in \citet{Requadt2025}. Since these recordings are naturally modeled by Example~\ref{HMM_example2}, they provide a relevant setting for inference on cooperativity among ion channels.

The theory assumes that the number $L$ of coordinates is known. In voltage-clamp applications, $L$ corresponds to the number of ion channels contributing to the measured current trace and is typically unknown. We thus select $L$ by model selection. Following \citet{VanEt24}, we minimize either the Bayesian information criterion (BIC),
\(\text{BIC}(L)=(3L+3)\log(n)-2\log p_{\hat \eta_L}(y_1,\dots,y_n)\)
or the Akaike
information criterion (AIC),
\(
\text{AIC}(L)=(6L+6)-2\log p_{\hat \eta_L}(y_1,\dots,y_n),
\)
where $(y_i)_{i=1}^n$ are the observations, and $p_{\hat \eta_L}$ is the likelihood in \eqref{e:likelihood} evaluated at the maximum likelihood estimator $\hat \eta_L$ in \eqref{e:mle} for a fixed $L$.

\subsubsection{Ryanodine receptor type 2}

We first analyze two datasets of wild-type ryanodine receptor type 2 (RyR2) ion channels from \citet{VanEt24}. The recordings were obtained from voltage-clamp experiments on artificial membranes under asymmetric ion concentrations, with two trans-side Ca$^{2+}$ conditions. No external voltage was applied, and the observed currents were induced by the chemical gradient across the membrane. Details are given in \citet{VanEt24}.

For both datasets, the sampling rate is $4$ kHz, corresponding to  $\delta=1/4000$, and the sample size is $n=600000$. AIC and BIC both suggest $L=3$, see \Cref{fig:bic_both} in Supplement~\ref{s:imp}. The maximum likelihood estimator in \eqref{e:mle} of the rate parameters in SDMC is:
\[
\hat{\theta}_{n}=\begin{cases}
(33.36,  4.50,   2.99,  31.65, 271.55, 243.76), &\text{for Dataset 1},\\
(43.90,   5.65,  8.30, 19.33, 195.32, 259.87), &\text{for Dataset 2}.
\end{cases}
\]
As a visual assessment of model fit, \Cref{fig3} shows the observed trace for Dataset~1 and its Viterbi path \citep{Vit67,MLM26}, the maximum a posteriori estimate of the hidden sum process. The corresponding plot for Dataset~2 is  in \Cref{ryr_dataset2} of Supplement~\ref{s:imp}.

The testing results by the proposed SCoT at significance level $\alpha\in\{0.01,0.05,0.1\}$ are reported in \Cref{tab:test_ryr2}. Most test statistics are negative, indicating negative cooperativity among RyR2 channels. Moreover, for all three significance levels, the cooperativity index estimator in \eqref{e:estviatest} gives $\hat \Lambda_{n,\alpha}=-{1}/{2}$ for Dataset~1, and $\hat \Lambda_{n,\alpha}=-{2}/{3}$ for Dataset~2.
These results provide evidence of negative cooperativity, consistent with \citet{VanEt24}, who used a discrete-time Markov model. The present analysis complements their conclusions with asymptotically justified statistical guarantees..

\begin{table}[h]
\centering
\caption{Testing results by the proposed SCoT for the RyR2 datasets. For each pairwise comparison, we report the test statistic and the smallest critical value used in \Cref{alg:stepdown} at significance level $\alpha\in\{0.01,0.05,0.1\}$.}
\label{tab:test_ryr2}
{\footnotesize
\begin{tabular}{ccSSSSc}
\toprule
& & & \multicolumn{3}{c}{\textbf{Critical value $\hat c_{n,\alpha}(\cdot)$}} & \\ 
\cmidrule{4-6}
& \textbf{Index pair} & \textbf{Test statistic} 
& {$\alpha = 0.01$} & {$\alpha = 0.05$} & {$\alpha = 0.1$} 
& \textbf{Null hypothesis} \\ 
\midrule
\multirow{6}{*}{\textbf{Dataset 1}} 
& (1, 2) & -47.7136 & 3.0632 & 2.5319 & 2.2658 & Rejected \\
& (4, 5) & -26.2961 & 2.9884 & 2.4469 & 2.1739 & Rejected \\
& (1, 3) & -20.6334 & 2.8857 & 2.3240 & 2.0377 & Rejected \\
& (4, 6) & -1.9363  & 2.8224 & 2.2521 & 1.9640 & Not rejected \\
& (2, 3) & -1.1162  & 2.8224 & 2.2521 & 1.9640 & Not rejected \\
& (5, 6) & 0.2528   & 2.8224 & 2.2521 & 1.9640 & Not rejected \\ 
\midrule
\multirow{6}{*}{\textbf{Dataset 2}}
& (1, 2) & -40.2092 & 3.0726 & 2.5392 & 2.2723 & Rejected \\
& (4, 5) & -32.0277 & 3.0028 & 2.4591 & 2.1863 & Rejected \\
& (1, 3) & -18.8108 & 2.9063 & 2.3414 & 2.0556 & Rejected \\
& (4, 6) & -4.7443  & 2.8211 & 2.2560 & 1.9678 & Rejected \\
& (2, 3) & 1.6013   & 2.8055 & 2.2347 & 1.9463 & Not rejected \\
& (5, 6) & -1.2586  & 2.8055 & 2.2347 & 1.9463 & Not rejected \\ 
\bottomrule
\end{tabular}}
\end{table}

\begin{figure}
    \centering
    \includegraphics[width=0.9\linewidth]{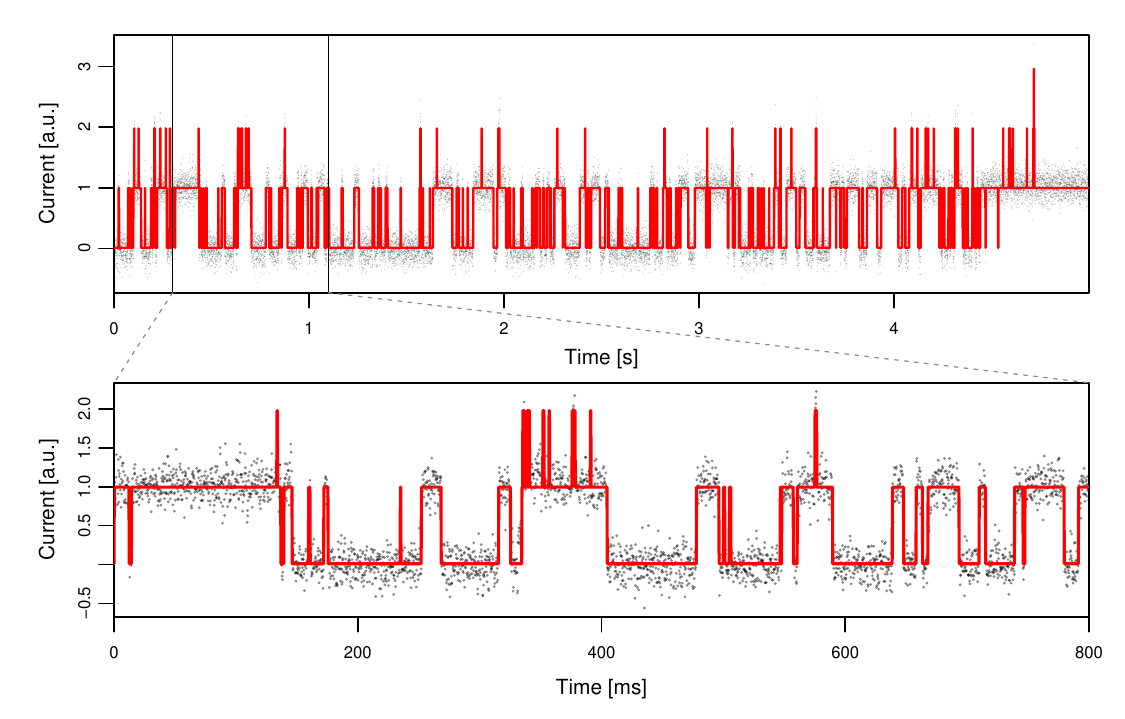}
    \caption{Dataset~1 of RyR2 channels:  current trace, consisting of $20000$ observations, is shown as black dots, and its Viterbi path is plot as a red line.}
    \label{fig3}
\end{figure}

\subsubsection{Gramicidin D}

We next analyze a dataset of gramicidin D channels from \citet{Requadt2025}. Gramicidin channels were incorporated into the membrane, and current traces were recorded under a fixed holding voltage of $150$ mV, see \citet{Requadt2025} for the full protocol. The sampling rate is $19.53$ kHz, corresponding to \(\delta=5.12\cdot 10^{-5}.\) For this dataset, AIC and BIC favor more complex models, see \Cref{fig:BIC_gram} in Supplement~\ref{s:imp}. This may be due to the substantially higher noise level in the measurement. To avoid overfitting, we instead choose $L=5$, as suggested by the histogram of idealized states in Figure~14 of \citet{Requadt2025}.

The maximum likelihood estimator of the SDMC model parameter is 
\[
\hat{\theta}_n
=(0.63, 1.32, 1.32, 6.91, 19.06, 6.15, 5.87, 3.33, 12.59, 20.59).
\]
\Cref{fig:viterbi_gram} displays the measured current trace and the corresponding Viterbi path.
The  testing results are reported in \Cref{tab:gram}. The global null hypothesis is not rejected at any significance level $\alpha\in\{0.01,0.05,0.1\}$. Thus, we find no statistically significant evidence of cooperativity among the gramicidin D channels. This agrees with \citet{Requadt2025}, while providing an asymptotic statistical justification within the proposed SD-HMM framework.

\begin{table}[h]
\centering
\caption{Testing results by the proposed SCoT for the gramicidin D dataset. For each pairwise comparison, we report the test statistic and the smallest critical value used in \Cref{alg:stepdown} at significance level $\alpha\in\{0.01,0.05,0.1\}$.}
\label{tab:gram}
{\footnotesize
\begin{tabular}{ccSSSSc}
\toprule
& &
& \multicolumn{3}{c}{\textbf{Critical value $\hat c_{n,\alpha}(\cdot)$}} 
&  \\ 
\cmidrule{4-6}
&  \textbf{Index pair} &  \textbf{Test statistic} 
& {$\alpha = 0.01$} & {$\alpha = 0.05$} & {$\alpha = 0.1$} 
& \textbf{Null hypothesis} \\ 
\midrule
& (1, 2)  & 1.2740  & 3.3807 & 2.8592 & 2.5940 & Not rejected \\
& (6, 7)  & 0.0873  & 3.3807 & 2.8592 & 2.5940 & Not rejected \\
& (1, 3)  & 0.7188  & 3.3807 & 2.8592 & 2.5940 & Not rejected \\
& (6, 8)  & 0.8818  & 3.3807 & 2.8592 & 2.5940 & Not rejected \\
& (1, 4)  & 1.4531  & 3.3807 & 2.8592 & 2.5940 & Not rejected \\
& (6, 9)  & -0.7885 & 3.3807 & 2.8592 & 2.5940 & Not rejected \\
& (1, 5)  & 0.9659  & 3.3807 & 2.8592 & 2.5940 & Not rejected \\
& (6, 10) & -0.6970 & 3.3807 & 2.8592 & 2.5940 & Not rejected \\
& (2, 3)  & -0.0035 & 3.3807 & 2.8592 & 2.5940 & Not rejected \\
& (7, 8)  & 0.7548  & 3.3807 & 2.8592 & 2.5940 & Not rejected \\
& (2, 4)  & 1.2866  & 3.3807 & 2.8592 & 2.5940 & Not rejected \\
& (7, 9)  & -0.8164 & 3.3807 & 2.8592 & 2.5940 & Not rejected \\
& (2, 5)  & 0.9295  & 3.3807 & 2.8592 & 2.5940 & Not rejected \\
& (7, 10) & -0.7097 & 3.3807 & 2.8592 & 2.5940 & Not rejected \\
& (3, 4)  & 1.2669  & 3.3807 & 2.8592 & 2.5940 & Not rejected \\
& (8, 9)  & -1.1262 & 3.3807 & 2.8592 & 2.5940 & Not rejected \\
& (3, 5)  & 0.9290  & 3.3807 & 2.8592 & 2.5940 & Not rejected \\
& (8, 10) & -0.8322 & 3.3807 & 2.8592 & 2.5940 & Not rejected \\
& (4, 5)  & 0.6214  & 3.3807 & 2.8592 & 2.5940 & Not rejected \\
& (9, 10) & -0.3629 & 3.3807 & 2.8592 & 2.5940 & Not rejected \\
\bottomrule
\end{tabular}}
\end{table}

\begin{figure}[ht]
    \centering
    \includegraphics[width=\linewidth]{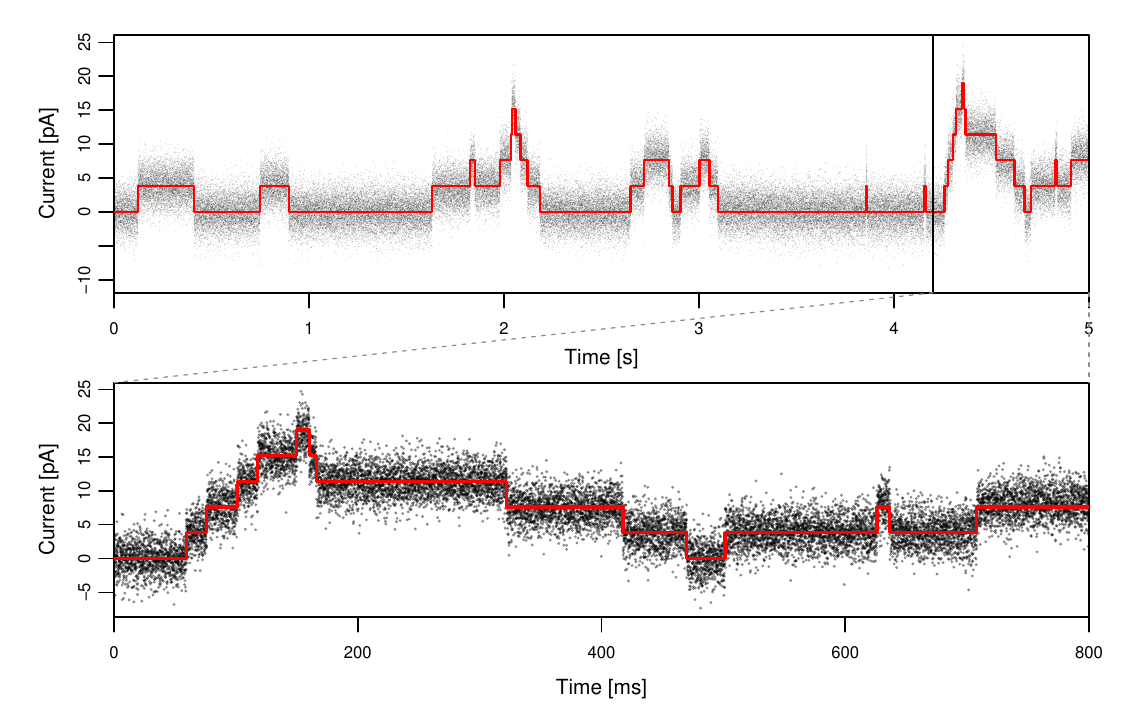}
    \caption{Dataset of Gramicidin D channels:  current trace, consisting of $100001$ observations, is shown as black dots, and its Viterbi path is plot as a red line.}
    \label{fig:viterbi_gram}
\end{figure}

\section*{Acknowledgments}
The authors thank the authors of \cite{VanEt24}, particularly B.~Eltzner, for sharing the two RyR2 datasets. They also thank the C.~Steinem lab at the University of Göttingen, especially M.~Fink, for measuring and sharing the gramicidin D voltage-clamp recordings. This work was supported by by German
Research Foundation (DFG) through the Mathematics of Experiment
under Grant CRC 1456 and in part by DFG under Germany's Excellence
Strategy through the Multiscale Bioimaging: From Molecular Machines
to Networks of Excitable Cells (MBExC) under Project EXC 2067.


\clearpage

\appendix

\section{Extension to multiple states}\label{section7}

We now extend the framework of SDMCs to a multi-state setting, thus broadening its applicability to settings where more than two regimes occur. In ion channels, for example, conductance may not switch directly between open and closed states but instead exhibit intermediate \enquote{sub-gating} levels corresponding to partially open conformations \citep{DANI1991401}. Such behavior cannot be captured by a binary-state model and motivates the extension developed in this section.

\subsection{Multi-state SDMCs}
For an integer $\kappa \ge 2$, we introduce an \emph{alphabet} 
$\mathcal{A}_{\kappa} = \{c_1,\dots,c_{\kappa}\}$ with ordered levels 
$0 = c_1 < c_2 < \dots < c_{\kappa} = 1$. 
When $\kappa = 2$, this reduces to the binary alphabet $\mathcal{A}_2 = \{0,1\}$.
We consider a continuous-time Markov chain 
$X_t = (X_{1,t}, \dots, X_{L,t})$, $t \ge 0$, 
with state space $\mathcal{A}_{\kappa}^L$. 
At a fixed time $t$, let $m_i$ denote the number of coordinates of $X_t$ 
that take the level $c_i \in \mathcal{A}_{\kappa}$. 
The vector $(m_1, \ldots, m_{\kappa})$ is called a \emph{configuration} of $X_t$ and belongs to
\[
\mathcal{M}_{L}^{\kappa} 
:= 
\left\{
(m_1,\ldots, m_\kappa) \in \{0,1, \ldots, L\}^\kappa 
:\; 
\sum_{i = 1}^{\kappa} m_i = L
\right\}.
\]
Let $S_t := \s{X_t}$, $t \ge 0$, denote the sum process, 
whose values correspond to the total level of $X_t$ over coordinates, and lie in
\begin{equation}\label{e:sum:k:state}
\mathcal{T}_L^{\kappa}
:=
\left\{
\sum_{i=1}^{\kappa} c_i m_i 
:\;
(m_1,\ldots, m_{\kappa}) \in \mathcal{M}_{L}^{\kappa}
\right\}.
\end{equation}
In particular, $\mathcal{T}_L^{2} = \{0, 1,\ldots, L\}$. One important property of an alphabet $\mathcal{A}_{\kappa}$ is that it may allow the unique recovery of the configuration of $X_t$ from the value of $S_t$.

\begin{definition}[Informativity]\label[definition]{d:inform}
We say that an alphabet $\mathcal{A}_{\kappa}$ is \emph{$L$-informative} if the map
\[
(m_1,\dots,m_{\kappa}) 
\mapsto 
\sum_{i=1}^{\kappa} c_i m_i
\]
is injective (and hence bijective) from $\mathcal{M}_{L}^{\kappa}$ to $\mathcal{T}_L^{\kappa}$. Further, we call $\mathcal{A}_{\kappa}$ \emph{informative} if it is $L$-informative for all $L \in \mathbb{N}$.
\end{definition}

Clearly, if $\mathcal{A}_{\kappa}$ is $L$-informative, then it is also $L'$-informative for all $L' \le L$. In particular, $\mathcal{A}_2$, as well as $\mathcal{A}_3$ with $c_2$ irrational, are informative.  More generally, informativity is closely related to linear Diophantine equations,  see \Cref{app:diophantine}.

In the sequel, we always assume that $\mathcal{A}_{\kappa}$ is $L$-informative; namely, for every $s\in \mathcal{T}_L^{\kappa}$ there exists a unique configuration $(m_1, \ldots, m_\kappa) \in \mathcal{M}_{L}^{\kappa}$ satisfying $s = \sum_{i = 1}^{\kappa}c_i m_i$. In such a case, we write, with slight abuse of notation,  
\[
m_i = (s)_i^{-1},\qquad i \in \{1,\ldots, \kappa\}.
\]
We now introduce the multi-state extension of SDMCs. 

\begin{definition}[$\kappa$-state SDMC]\label{d:kSDMC}
A continuous-time Markov chain $(X_t)_{t\ge 0}$ with state space $\mathcal{X}=\mathcal{A}_\kappa^L$ is called a \emph{$\kappa$-state sum-dependent Markov chain ($\kappa$-state SDMC)} if its rate matrix $Q=(q_{x,y})_{x,y\in\mathcal{X}}$ satisfies, for $x\neq y$,
\begin{equation*}
q_{x,y}=
\begin{cases}
 \sum_{j=1}^{\kappa}\sum_{k = 1}^{\kappa}\param_{j,k;\s{x}}\mathbb{I}(x_{i}=c_j, y_{i} = c_{k}),
& \text{if } \|x-y\|_0=1,\\
0, & \text{if } \|x-y\|_0>1,
\end{cases}
\end{equation*}
and $q_{x,x} = -\sum_{y: y\neq x} q_{x,y}$. Here 
$$
\left\{\param_{j,k;s}\; :\; j\neq k , j,k \in \{1,\ldots, \kappa\},\, s \in \mathcal{T}_L^\kappa \text{ and } (s)_j^{-1} \ge 1 \right\} \;\subseteq\; [0,\infty)
$$
is the set of model parameters, the cardinality of which is $\kappa\cdot (\kappa-1)\cdot{L+\kappa-2\choose \kappa-1 }$.
\end{definition}

Similar to binary-state SDMCs (cf.~Section~\ref{sec1}), a $\kappa$-state SDMC can equivalently be characterized by permutation invariance (\Cref{as1}) and conditional independence at infinitesimal times (\Cref{as2}) in exactly the same manner. Moreover, the model parameters of the $\kappa$-state SDMC satisfy 
\[
\param_{j,k;s} = \lim_{\delta\searrow 0} \frac{\prob{X_{i,t+\delta}=c_{k}\;\middle\vert\;X_{i,t}=c_{j}, \s{X_t}=s} }{\delta}
\]
which is independent of both $i \in\{1,\ldots, L\}$ and $t\ge0$. Consequently, Theorem~\ref{Theorem1} extends directly to general $\kappa$-state SDMCs.

There is also a one-to-one correspondence between the rate matrix (or equivalently the model parameters $\param_{j,k;s}$) of a $\kappa$-state SDMC and that of its associated sum process.

\begin{theorem}[Identifiability]\label{th:id:kSDMC}
Let $(X_t)_{t\ge 0}$ be a $\kappa$-state SDMC. Then: 
\begin{enumerate}
\item \label{th:id:kSDMC:1} The sum process $S_t= \s X_t$, $t\ge 0$, is a continuous-time Markov chain.
\item \label{th:id:kSDMC:2}The rate matrix $R = (r_{s,s'})_{s,s'\in \mathcal{T}_L^{\kappa}}$ of $(S_t)_{t\ge0}$ takes the form of 
\[
r_{s,s'} = 
\begin{cases}
(s)^{-1}_j\param_{j,k; s}, & \text{ if }s' = s  - c_j +c_k\text{ for some distinct }j,k,\\
-\sum_{\text{distinct }j, k}(s)^{-1}_j \param_{j,k;s},& \text{ if } s' = s,\\
0, & \text{ otherwise.}
\end{cases}
\]
\item \label{th:id:kSDMC:3} The model parameters $\param_{j,k;s}$, for distinct $j,k \in \{1,\ldots, \kappa\}$ and $s \in \mathcal{T}_L^\kappa$ satisfying $(s)_j^{-1} \ge 1$, are uniquely determined by the rate matrix $R$ defined above.
\end{enumerate}
\end{theorem}
\begin{proof}
\textbf{Part~\ref{th:id:kSDMC:1}} By Definition~\ref{d:kSDMC}, the rate matrix $Q = (q_{x,y})_{x,y\in \mathcal A_\kappa^L}$ is permutation invariant, as in Assumption~\ref{as1} with $\mathcal{X} = \mathcal{A}_\kappa^L$. Then, by Proposition~\ref{prop1} and Lemma~\ref{Lemma3.9}, the sum process $(S_t)_{t\ge0}$ is a Markov chain.

\smallskip

\textbf{Part~\ref{th:id:kSDMC:2}} By Proposition \ref{prop2}, it holds for any $x\in \mathcal{A}_\kappa^L$ with $\s{x}=s$,
\[
r_{s,s'}=\sum_{y:\s{y}=s'}q_{x,y}=\sum_{y:\s{y}=s',\norm{x-y}_0=1}q_{x,y}.
\]
Then $r_{s,s'}=0$, whenever 
\[
s'\notin \bigl\{s-c_{j}+c_{k}\;:\; j,k\in\{1,\dots,\kappa\}\bigr \}.
\]
By the $L$-informativity of $\mathcal{A}_\kappa$, the above displayed condition is equivalent to $\norm{x-y}_0>1$ for all $x,y\in \mathcal{A}_\kappa^L$ such that $\s{x}=s$ and $\s{y}=s'$. Fix arbitrarily  $x\in \mathcal{A}_\kappa^L$  such that  $\s{x}=s$. Then,  for any $k\neq j$, we have
\[         
r_{s,s- c_{j}+c_{k}}=\sum_{y:\s{y}=s-c_{j}+c_{k},\norm{x-y}_0=1}q_{x,y}
         =\sum_{y:\s{y}=s-c_{j}+c_{k},\norm{x-y}_0=1} \param_{j,k;s}
         =(s)_{j}^{-1}\param_{j,k;s}.
\]
The formula for the diagonal elements follows from the fact that the row sums are zero.

\smallskip

\textbf{Part~\ref{th:id:kSDMC:3}} For distinct $j,k \in \{1,\ldots, \kappa\}$ and $s \in \mathcal{T}_L^\kappa$ satisfying $(s)_j^{-1} \ge 1$, we have, by part~\ref{th:id:kSDMC:2}, 
\[
\param_{j,k;s}  = \frac{r_{s,s-c_{j}+c_{k}}}{(s)_{j}^{-1}}.
\]
Thus, the one-to-one correspondence between the rate matrix $R$ and the model parameters $\param_{j,k;s}$ is immediate. 
\end{proof}

In contrast to the binary-state setting, the sum process $S_t$ of a $\kappa$-state SDMC is no longer a birth--death process. Instead, via $L$-informativity, it can be identified with a Markov population process taking values in $\mathcal{M}_{L}^{\kappa}$ \citep{Kingman69}. Unlike birth--death processes, reversibility does not generally hold for Markov population processes. Nevertheless, the special structure of $\kappa$-state SDMCs still guarantees reversibility, a property that, as noted earlier, is desirable in many applications.

\begin{theorem}[Irreducibility and reversibility]\label{th:rev:kSDMC}
Let $(X_t)_{t\ge 0}$ be a $\kappa$-state SDMC with $\kappa\ge 3$, and let $S_t=\s X_t$, $t\ge0$, denote the associated sum process. Then:
\begin{enumerate}
\item \label{th:rev:kSDMC:1}
 $(X_t)_{t\ge0}$ is irreducible {if and only if} $(S_t)_{t\ge0}$ is irreducible. Further, if both chains are irreducible, they admit unique invariant distributions
$\pi=(\pi_x)_{x\in \mathcal{A}^L_\kappa}$ and
$\gamma=(\gamma_s)_{s\in\mathcal{T}_L^\kappa}$, respectively, which satisfy
\[
\gamma_s
= \sum_{x':\,\s{x'}=s}\pi_{x'}
= \frac{L!}{(s)^{-1}_1!\cdots (s)^{-1}_\kappa!}\,\pi_x
\]
for any $x$ such that $\s{x}=s$.

\item \label{th:rev:kSDMC:2}
Suppose that $(X_t)_{t\ge0}$ or $(S_t)_{t\ge0}$ is irreducible. Then $(X_t)_{t\ge0}$ is reversible if and only if $(S_t)_{t\ge0}$ is reversible.

\item \label{th:rev:kSDMC:3}
Suppose that all model parameters are  nonzero, and thus strictly positive, i.e.
\[
\{\param_{j,k;s} : j\neq k, j,k\in\{1,\ldots,\kappa\},\ s\in\mathcal{T}_L^\kappa,\ (s)_j^{-1}\ge1\}
\subset (0,\infty).
\]
Then both $(X_t)_{t\ge0}$ and $(S_t)_{t\ge0}$ are irreducible. Moreover, both chains are reversible if and only if, for all pairwise distinct
$a_1,a_2,a_3\in\{1,\ldots,\kappa\}$ and all
$s\in\mathcal{T}_L^\kappa$ with $(s)^{-1}_{a_1}>0$,
\begin{equation}
\param_{a_1,a_2;s}\,
\param_{a_2,a_3;s+c_{a_2}-c_{a_1}}\,
\param_{a_3,a_1;s+c_{a_3}-c_{a_1}}
=
\param_{a_1,a_3;s}\,
\param_{a_3,a_2;s+c_{a_3}-c_{a_1}}\,
\param_{a_2,a_1;s+c_{a_2}-c_{a_1}}.
\label{eq:cycle}
\end{equation}
\end{enumerate}
\end{theorem}

\begin{proof}
\textbf{Part~\ref{th:rev:kSDMC:1}} By Propositions~\ref{prop2} and \ref{prop1}, it holds, for any $s,s' \in \mathcal{T}_L^\kappa$, and any $\delta>0$, 
\begin{equation} \label{sum_id}
\prob{S_{t+\delta}=s'\mid S_t=s}=\prob{S_{t+\delta}=s'\mid X_t=x}=\sum_{y:\s{y}=s'}\prob{X_{t+\delta}=y\mid X_t=x},    
\end{equation}
where $x$ is an arbitrary element in $\mathcal{A}_\kappa^L$ that satisfies $\s{x}=s$. As a consequence, if $(X_t)_{t\ge0}$ is irreducible, then $(S_t)_{t\ge0}$ is irreducible.

Now consider the other direction, and let $x$ and $y$ be arbitrary elements of $\mathcal{A}_\kappa^L$.  The irreducibility of $(S_t)_{t\ge 0}$ implies that there exists a sequence of $s_0,\ldots,s_{m},\ldots,s_{m+n} \in \mathcal{T}_L^{\kappa}$ such that $s_0 = \s x$, $s_{m} = 0$,  $s_{m+n} = \s y$, and 
\[
 \prod_{k=1}^{m+n}r_{s_{k-1},s_{k}} > 0.
\]
By the relation of $(S_t)_{t\ge0}$ and $(X_t)_{t\ge0}$, there exists a sequence of $z_0, \ldots, z_m,\ldots, z_{m+n} \in \mathcal{A}_\kappa^L$ such that $\s z_k = s_k$ for $k\in\{0,\ldots, m+n\}$ and 
\[
 \prod_{k=1}^{m+n}q_{z_{k-1},z_{k}} > 0.
\]
By the $L$-informativity of $\mathcal{A}_\kappa$, there exist (possibly distinct) permutation matrices ${\Perm},{\Perm}'\in  \{0,1\}^{L \times L}$ such that ${\Perm} z_0 = x$ and ${\Perm}'z_{m+n} = y$. We define 
\[
\tilde{z}_k = 
\begin{cases}
{\Perm} z_k&\text{ if } k \in\{0,\ldots,m\},\\
{\Perm}' z_k & \text{ if } k \in\{m,\ldots,m+n\}.
\end{cases}
\]
The definition of $\tilde{z}_m= {\Perm} z_m = {\Perm}' z_m$ is consistent, because $s_m = \s z_m = 0$ and thus $z_m = (0,\ldots,0) \in  \mathcal{A}_\kappa^L$. By the permutation invariance (\Cref{as1}) of $(X_t)_{t\ge0}$, we have 
\[
 \prod_{k=1}^{m+n}q_{\tilde{z}_{k-1},\tilde{z}_{k}} > 0.
\]
Note that $\tilde{z}_0 = x$ and $\tilde{z}_{m+n} = y$. Thus, $(X_t)_{t\ge0}$ is irreducible. 
If both chains $(X_t)_{t\ge0}$ and $(S_t)_{t\ge0}$ are irreducible, they are positive recurrent (as the state space is finite), and thus have unique invariant distributions  (see e.g.\ \citealp[Theorem~3.5.2]{Norris1998}). Let $\pi=(\pi_x)_{x\in \mathcal{A}^L_\kappa}$ be the invariant distribution of $(X_t)_{t\ge0}$. Then, for any $y \in \mathcal{A}^L_\kappa$, 
\[
\sum_{x \in \mathcal{A}^L_\kappa}\pi_x q_{x,y} = 0. 
\]
Define $\gamma=(\gamma_s)_{s\in\mathcal{T}_L^\kappa}$ with $\gamma_s = \sum_{x:\s{x}=s}\pi_x$. By Proposition \ref{prop2}, we have 
\begin{equation}\label{e:r:sum:q}
r_{s,s'}=\sum_{y:\s{y}=s'}q_{x,y}
\end{equation}
for any $x\in \mathcal{A}_\kappa^L$ with $\s{x}=s$. Then, for any $s'\in \mathcal{T}_L^\kappa$, 
\[ 
\sum_{s \in\mathcal{T}_L^\kappa}\gamma_s r_{s,s'}= \sum_{s \in\mathcal{T}_L^\kappa}\sum_{y:\s{y}=s'}\sum_{x:\s{x}=s}\pi_xq_{x,y}=\sum_{y:\s{y}=s'}\sum_{x \in \mathcal{A}^L_\kappa} \pi_xq_{x,y}=0.
\]
Thus, $\gamma=(\gamma_s)_{s\in\mathcal{T}_L^\kappa}$ is the invariant distribution of $(S_t)_{t\ge0}$. Due to the permutation invariance (cf.\ Assumption~\ref{as1}) of $Q = (q_{x,y})_{x,y\in\mathcal{X}}$ and the uniqueness of invariant distribution, we obtain that $\pi_x=\pi_{x'}$ for all $x,x'$ satisfying $\s{x}=\s{x'}$. This together with the $L$-informativity of $\mathcal{A}_{\kappa}$ implies that, for any  $x$ with $\s{x}=s$,
\[ 
\gamma_s=\sum_{x':\s{x'}=s}\pi_{x'}=\frac{L!}{(s)^{-1}_1! \dots (s)^{-1}_\kappa!}\pi_x.
\]
\smallskip

\textbf{Part~\ref{th:rev:kSDMC:2}}  As in part~\ref{th:rev:kSDMC:1}, let $\pi = (\pi_x)_{x\in \mathcal{A}_{\kappa}^L}$ and $\gamma=(\gamma_s)_{s\in\mathcal{T}_L^\kappa}$ denote the invariant distributions of $(X_t)_{t\ge0}$ and  $(S_t)_{t\ge0}$, respectively.

Assume first that $(X_t)_{t\ge0}$ is reversible. Then
\begin{equation*}
\pi_x q_{x,y} = \pi_y q_{y,x}\quad\text{for all }x,y \in \mathcal{A}_{\kappa}^L.
\end{equation*}
Then, by \eqref{e:r:sum:q} and $\gamma_s = \sum_{x:\s{x} = s} \pi_x$, we have, for any $s,s'\in \mathcal{T}_L^\kappa$,
\begin{equation}\label{e:db:r}
\gamma_s r_{s,s'}=\sum_{x,y:\s{x}=s,\s{y}=s'}\pi_x q_{x,y}     =\sum_{x,y:\s{x}=s,\s{y}=s'}\pi_y q_{y,x}=\gamma_{s'} r_{s',s},
\end{equation}
which shows the reversibility of $(S_t)_{t\ge0}$.

Assume next that $(S_t)_{t\ge0}$ is reversible, namely, the detailed balance equations~\eqref{e:db:r} hold. Consider arbitrarily $x,y\in \mathcal{A}_{\kappa}^L$ with  $\norm{x-y}_0 = 1$. Then there exists $i \in \{1,\ldots, L\}$ such that $x_i = c_a$ and $y_i = c_b$ for some $a\neq b \in \{1,\ldots,\kappa\}$, while $x_j = y_j$ for all $j \neq i$. Let $s = \s{x}$ and $s' = \s{y}$. Then $s'=s+c_b-c_a$, and thus 
\begin{equation}\label{e:relation:s:sp}
(s')_b^{-1}=1+(s)^{-1}_b, \quad (s)_a^{-1}=1+(s')^{-1}_a, \quad (s')_k^{-1} = (s)_k^{-1}\text{ for } k\not\in\{a,b\}.
\end{equation} 
By Theorem~\ref{th:id:kSDMC}\ref{th:id:kSDMC:2} and~\eqref{e:db:r}, we have
\[
\gamma_s (s)_a^{-1}\param_{a,b;s}=\gamma_{s'}(s')_b^{-1}\param_{b,a;s'}. 
\]
This, together with \eqref{e:relation:s:sp} and the relation between $\pi$ and $\gamma$ in part~\ref{th:rev:kSDMC:1}, implies 
\[
\pi_x q_{x,y}=\frac{\prod_{k=1}^\kappa (s)^{-1}_k!}{L}\gamma_{s} \param_{a,b;s} =\frac{\prod_{k=1}^\kappa (s')^{-1}_k!}{L}\gamma_{s'}\param_{b,a;s'}=\pi_yq_{y,x},
\]
which further yields the reversibility of $(X_t)_{t\ge0}$, due to the sparse structure of $Q = (q_{x,y})$, see Definition~\ref{d:kSDMC}.

\smallskip

\textbf{Part~\ref{th:rev:kSDMC:3}} Consider arbitrarily two distinct states $x,y\in \mathcal{A}_\kappa^L$. Let $1\le i_1 < \cdots < i_m \le L$ coordinates where $x$ and $y$ differ, i.e.~$x_j \neq y_j$ for $j \in \{i_1,\ldots, i_m\}$ and $m = \norm{x - y}_0$. We define a sequence of states as follows: $z_0 = x$ and $z_k = (z_{k,1},\ldots,z_{k,L})$, $k\in \{1,\ldots,m\}$, with
$$
z_{k,j} = 
\begin{cases}
y_j, & \text{ if } j \le i_k,\\
x_j, & \text{ if } j > i_k.
\end{cases}
$$
Then $z_m = y$, and $\norm{z_k - z_{k-1}}_0 = 1$ for $k \in \{1,\ldots, m\}$. Thus, $q_{z_{k-1}, z_k} > 0$, by positivity of the model parameters. Hence, $(X_t)_{t\ge0}$ is irreducible, and by part~\ref{th:id:kSDMC:1} $(S_t)_{t\ge0}$ is also irreducible.

Due to the $L$-informativity of $\mathcal{A}_{\kappa}$, we can identify $(S_t)_{t\ge 0}$ as a Markov population process taking values in the set $\mathcal{M}_L^\kappa$ of configurations. Then, by \citet[Section~2]{Kingman69}, assumption~\eqref{eq:cycle} implies that $(S_t)_{t\ge 0}$ is reversible, and by part~\ref{th:id:kSDMC:2} $(X_t)_{t\ge0}$ is also reversible. 
\end{proof}

In the more-than-two-state setting ($\kappa \ge 3$), in contrast to the case $\kappa = 2$ (cf.\ Lemma~\ref{l:ir:sd} and Proposition~\ref{bd_proc}\ref{i:bd_proc:1}), positivity of all model parameters is sufficient, but not necessary, for irreducibility of $(S_t)_{t\ge0}$ and $(X_t)_{t\ge0}$. For the characterization of reversibility, the additional condition~\eqref{eq:cycle} is required. This condition constitutes a simplified form of Kolmogorov's criterion \citep[Theorem 1.8]{Kelly1979}, as it imposes constraints only on triangles (i.e.\ closed paths of length three) rather than on cycles of arbitrary length.

\subsection{Informativity and Diophantine equations}\label[appendix]{app:diophantine}
The informativity condition (introduced in \Cref{d:inform}) is closely related to linear Diophantine equations.

\begin{lemma}\label{dioph}
The following two statements are equivalent:
\begin{enumerate}
    \item \label{Lemma_inf_1} $\mathcal{A}_\kappa = \{c_1,\dots,c_{\kappa}\}$ is $L$-informative.
    \item \label{Lemma_inf_2} If a nonzero vector $d=(d_1,\dots,d_\kappa)\in \mathbb{Z}^\kappa$ satisfies the linear Diophantine equations
    \[
    \sum_{i=1}^\kappa c_i d_i=0\quad\text{and}\quad \sum_{i=1}^\kappa d_i=0,
    \]
    then 
    \[
    \sum_{i= 1}^{\kappa}\max\{d_i, 0\} >L.
    \]
\end{enumerate}
\end{lemma}

\begin{proof}
\textbf{\ref{Lemma_inf_1} $\implies$ \ref{Lemma_inf_2} }
Anticipating contradiction, we assume that there exists a nonzero vector  $d\in  \mathbb{Z}^\kappa$ such that 
\[
\sum_{i=1}^\kappa c_i d_i=0,\quad \sum_{i=1}^\kappa d_i=0\quad \text{and}\quad \sum_{i:d_i\ge 0}d_i\le L.
\] 
It follows that $\sum_{i:d_i\ge 0}d_i=\sum_{i:d_i\le 0}(-d_i)$ and $L-\sum_{i:d_i\ge 0}d_i\ge 0$. Define 
\begin{align*}
m_k^+ & = \max(d_k,0) + \mathbb{I}(k=1)\Bigl(L - \sum_{i : d_i \ge 0} d_i\Bigr),\\
\text{and}\quad m_k^- & = \max(-d_k,0) + \mathbb{I}(k=1)\Bigl(L - \sum_{i : d_i \ge 0} d_i\Bigr).
\end{align*} 
Then $m_k^+,m_k^-\in \{0,\dots,L\}$ and $\sum_{k=1}^\kappa m_k^+=\sum_{k=1}^\kappa m_k^-=L$. Thus,  $m^+=(m^+_1,\dots,m^+_\kappa)$ and $m^-=(m^-_1,\dots,m^-_\kappa)$ are valid configurations in $\mathcal{M}^\kappa_L$. Moreover, it holds
\[ 
\sum_{k=1}^\kappa c_k m_k^+ - \sum_{k=1}^\kappa c_k m^-_k=\sum_{k=1}^\kappa c_k d_k=0.
\]
By the $L$-informativity of $\mathcal{A}_\kappa$,  this implies that $m^+_k=m^-_k$ and hence $d_k = m^+_k-m^-_k$, for all $k\in\{1,\dots,\kappa\}$, which contradicts with the fact that $d$ is a nonzero vector. 

\smallskip

\textbf{\ref{Lemma_inf_2}$\implies$ \ref{Lemma_inf_1} } 
Consider arbitrarily $m = (m_1,\ldots, m_{\kappa}), m' = (m'_1,\ldots, m'_{\kappa})\in \mathcal{M}_L^\kappa$ such that $\sum_{i=1}^{\kappa}c_im_i = \sum_{i= 1}^{\kappa}c_i m'_i$. We define $d = (d_1,\ldots, d_{\kappa})$ with $d_i = m_i - m'_i$. Then
\[
\sum_{i=1}^\kappa d_i=0\quad\text{and}\quad\sum_{i=1}^\kappa  c_i d_i = 0.
\]
Note further that 
$$
\sum_{i : d_i\ge 0}d_i\le \sum_{i:d_i\ge 0}m_i\le \sum_{i=1}^\kappa m_i=L.
$$
Thus, by assumption, $d = 0$, i.e.~$m_i = m'_i$ for all $i$, which shows that  $\mathcal{A}_\kappa$ is $L$-informative. 
\end{proof}

The previous result establishes a direct connection to the linear Diophantine system
\begin{equation}\label{e:diosys}
\begin{pmatrix}
1 & 1 & \dots & 1\\
c_1 & c_2 & \dots & c_\kappa
\end{pmatrix} d =
\begin{pmatrix}
0\\
0
\end{pmatrix}, \quad d \in \mathbb{Z}^\kappa.
\end{equation}

We next consider two special cases in which informativity, or equivalently the structure of the Diophantine system, can be characterized more explicitly.

\begin{lemma}
The following statements hold:
\begin{enumerate}
    \item (Irrational case) \label{cor_F_1}
    If real numbers $c_1,\dots,c_\kappa$ are linearly independent over $\mathbb{Q}$, i.e.~linearly independent with respect to rational coefficients, then $\mathcal{A}_\kappa = \{c_1,\dots,c_\kappa\}$ is $L$-informative for all $L\in \mathbb{N}$.
    \item \label{cor_F_4} (Rational case with $\kappa=3$) Let $L\in \mathbb{N}$, and $p,q \in \mathbb{N}$ such that $\gcd(p,q)=1$ and $p<q$, where $\gcd(p,q)$ denotes the greatest common divisor of $p$ and $q$. Then $\mathcal{A}_3=\{0,p/q,1\}$ is $L$-informative if and only if $L<q$.
\end{enumerate}
\end{lemma}

\begin{proof}
\textbf{Part~\ref{cor_F_1}} 
Let $d = (d_1,\ldots, d_{\kappa}) \in \mathbb{Z}^\kappa$ be a solution to \eqref{e:diosys}. By the linear independence of $c_1,\dots,c_\kappa$  over $\mathbb{Q}$, we have $d=0$. Further, Lemma~\ref{dioph} yields the statement.
   
\smallskip

\textbf{Part~\ref{cor_F_4}} 
Note that \eqref{e:diosys} becomes
\[ 
d_2 (p/q)+d_3=0,\quad d_1+d_2+d_3=0.
\]
The general integer solution is given by 
$$
d_1 = n(p-q),\quad d_2=nq,\quad d_3=-np, \quad n\in \mathbb{Z}.
$$
If $d = (d_1,d_2,d_3) \neq 0$, then $n \neq 0$ and thus
\[
\sum_{i:d_i\ge 0}d_i=\mathbb{I}(n\ge 0)nq+\mathbb{I}(n<0)\bigl(n(p-q)-np\bigr)=\abs{n}q \ge q,
\]
with equality attained for $n \in \{-1,1\}$. By Lemma~\ref{dioph}, $\mathcal{A}_3$ is $L$-informative if and only $L<\sum_{i:d_i\ge 0}d_i$ for any nonzero solution $d = (d_1,d_2,d_3)$, which is equivalent to $L < q$.
\end{proof}

\subsection{Informativity under noisy measurements}
It is natural to expect that, for a fixed alphabet $\mathcal{A}_{\kappa}$, the large-sample properties (i.e.~as the number of observations tends to infinity) of maximum likelihood estimator in Section~\ref{s:estimation} would extend to general $\kappa$-state SDMCs. However, when the alphabet $\mathcal{A}_{\kappa}$ is allowed to vary with the sample size, a stronger notion of informativity is required to recover the latent configuration of $X_t$ from noisy observations of the associated sum process $S_t$. We therefore introduce a strengthened notion of informativity that serves as a foundation for asymptotic analysis in settings with varying alphabets. A systematic treatment of the resulting asymptotic statistical theory lies beyond the scope of the present paper but constitutes a promising direction for future research.

\begin{definition}[$\delta$-informativity]\label[definition]{d:delta:inform}
An alphabet $\mathcal{A}_{\kappa}$ is said to be \emph{$L$-$\delta$-informative}, for some $\delta \ge 0$, if for every pair of distinct configurations $(m_1,\dots,m_{\kappa}) \neq (m_1',\dots,m_{\kappa}') \in \mathcal{M}_{L}^{\kappa}$,
\[
\min_{|\varepsilon| \le \delta,\; |\varepsilon'| \le \delta}
\left|
\sum_{i=1}^{\kappa} c_i m_i + \varepsilon
-
\sum_{i=1}^{\kappa} c_i m_i' - \varepsilon'
\right| > 0.
\]
The alphabet $\mathcal{A}_{\kappa}$ is called \emph{$\delta$-informative} if it is $L$-$\delta$-informative for every $L \in \mathbb{N}$.
\end{definition}

When $\delta = 0$, $L$-$\delta$-informativity (\Cref{d:delta:inform}) reduces to $L$-informativity (\Cref{d:inform}). 
We define the \emph{discernibility} of $\mathcal{A}_{\kappa}$ over the configuration space $\mathcal{M}_{L}^{\kappa}$ by
\[
d_{L}(\mathcal{A}_{\kappa})
\;:=\;
\min_{\substack{(m_1,\ldots,m_\kappa) \neq (m_1',\ldots,m_\kappa')\\ \in \mathcal{M}_{L}^{\kappa}}}
\left|
\sum_{i=1}^{\kappa} c_i m_i
-
\sum_{i=1}^{\kappa} c_i m_i'
\right|.
\]
The quantity $d_L(\mathcal{A}_\kappa)$ measures the minimal separation between distinct configurations in the aggregated domain and provides an equivalent characterization of $L$-$\delta$-informativity.

\begin{lemma}\label{e-inf}
An alphabet $\mathcal{A}_{\kappa}$ is $L$-$\delta$-informative if and only if 
\[
d_{L}(\mathcal{A}_{\kappa}) > 2\delta.
\]
\end{lemma}

\begin{proof}
Introduce the notation $w(m):=\sum_{i=1}^\kappa c_i m_i$ for $m=(m_1,\dots,m_\kappa)\in \mathcal{M}_{L}^{\kappa}$. 

\smallskip

\textbf{Direction ``$\impliedby$'':}  Consider arbitrarily $\varepsilon,\varepsilon'$ with $\abs{\varepsilon}\le \delta,\abs{\varepsilon'}\le \delta$,  and arbitrarily $m,m' \in \mathcal{M}_L^\kappa$ with $m\neq m'$. Note that 
\[
\abs{w(m)+\varepsilon-w(m')-\varepsilon'}\ge\abs{w(m)-w(m')}-\abs{\varepsilon-\varepsilon'} \ge d_{L}(\mathcal{A}_{\kappa})-2\delta>0.
\]
This yields the $L$-$\delta$-informativity of $\mathcal{A}_{\kappa}$.

\smallskip

\textbf{Direction ``$\implies$'':} Arguing by contradiction, suppose that $d_{L}(\mathcal{A}_{\kappa})\le 2 \delta$. Then there exist $m,m'\in \mathcal{M}_{L}^{\kappa} $ with $m\neq m'$ such that $\abs{w(m)-w(m')}\le 2 \delta$. We set 
$$
\varepsilon=\frac{w(m')-w(m)}{2} \quad \text{ and }\quad \varepsilon'=-\varepsilon.
$$
Then $\abs{\varepsilon}\le  \delta$, $\abs{\varepsilon'}\le \delta$ and 
\(
\abs{w(m)+\varepsilon-w(m')-\varepsilon'}=0.
\)
This shows that $\mathcal{A}_{\kappa}$ is not $L$-$\delta$-informative, yielding a contradiction. Therefore, it must hold that $d_{L}(\mathcal{A}_{\kappa})> 2 \delta$.
\end{proof}

The $L$-$\delta$-informativity, though formulated via deterministic perturbations, is capable of accommodating stochastic noises. 

\begin{example}
Fix $t \ge 0$ and suppose that the sum process $S_t$, induced by a $\kappa$-state SDMC $X_t$, is observed with additive noise,
\[
Y_t = S_t + \varepsilon_t,
\]
where $\varepsilon_t \sim \mathcal{N}(0,\sigma^2)$ is independent of $S_t$. 
Assume that $\mathcal{A}_{\kappa}$ is $L$-$\delta$-informative and consider the event
\[
A_t := \left\{\text{the configuration of $X_t$ is uniquely recoverable from $Y_t$}\right\}.
\]
By Lemma~\ref{e-inf}, any two distinct configurations are separated by at least $2\delta$ in the noiseless domain. Hence, misclassification can occur only if $|\varepsilon_t| \ge \delta$. 
Applying standard Gaussian tail bounds (i.e.~the Mills ratio) together with Lemma~\ref{e-inf} yields
\[
\prob{A_t}\;\ge\;\max(0,1-2e^{-\delta^2/(2\sigma^2})).
\]
This bound makes precise that $L$-informativity is preserved with high probability under Gaussian noise, with an error exponent determined by the separation margin $\delta$ relative to the noise level $\sigma$.
\end{example}

\clearpage

\section{Continuous-time Markov chains}\label{sec_ct}

This section provides a brief introduction to the theory of continuous-time Markov chains. Let $\{X_t : t \ge 0\}$ be a right-continuous stochastic process with a finite state space $\mathcal{X}$. We begin by recalling the definition of a continuous-time Markov chain.

\begin{definition}[Continuous-time Markov chain]\label{cont_t_mc}
A right-continuous stochastic process $(X_t)_{t \ge 0}$ taking values in a finite state space $\mathcal{X}$ is called a (time-homogeneous) \emph{continuous-time Markov chain} with rate matrix (or infinitesimal generator)
$Q = [q_{x,y}]_{x,y \in \mathcal{X}} \in \mathbb{R}^{|\mathcal{X}| \times |\mathcal{X}|}$
if, for all $x,y \in \mathcal{X}$, uniformly in $t \ge 0$,
\[
\mathbb{P}(X_{t+\delta} = y \;\vert\;X_t = x)
= \mathbb{I}\{x=y\} + \delta q_{x,y} + o(\delta),
\qquad \text{as } \delta \searrow 0.
\]
The vector $\pi = (\pi_x)_{x\in \mathcal{X}}$, with $\pi_x = \mathbb{P}(X_0 = x)$,  is referred to as the initial probability vector.
\end{definition}

The rate matrix $Q = [q_{x,y}]_{x,y \in \mathcal{X}}$ of a continuous-time Markov chain satisfies
\[
q_{x,y} \ge 0 \quad \text{for all } x \neq y,
\qquad \text{and} \qquad
\sum_{y \in \mathcal{X}} q_{x,y} = 0 \quad \text{for all } x \in \mathcal{X}.
\]
Conversely, for any matrix $Q$ satisfying these conditions, there exists a (time-homogeneous) continuous-time Markov chain with $Q$ as its rate matrix.

The infinitesimal characterisation of Markov property in Definition~\ref{cont_t_mc} is equivalent to the formulation via transition probability matrices $P(\delta) =  [p_{x,y}(\delta)]_{x,y \in \mathcal{X}}$, $\delta > 0$, in the sense that 
\[
\prob{X_{t_{n+1}} = x_{n+1} \;\middle\vert\;X_{t_n} = x_n, \dots, X_{t_0} = x_0}
= p_{x_n,x_{n+1}}(t_{n+1} - t_n)
\]
for any $t_0 \le t_1 \le \dots \le t_{n+1}$, see e.g.\ \citet[Theorem 2.8.2]{Norris1998}. For a given continuous-time Markov chain, its rate matrix $Q$ and transition probability matrices $P(\delta)$ are linked via 
\[
P(\delta) = \exp(Q \delta),\quad \delta > 0.
\]
As a consequence, the transition probabilities admit the series representation
\[
p_{x,y}(\delta) = \prob{X_{t+\delta} = y \;\middle\vert\;X_t = x}
= \sum_{k=0}^\infty \frac{\delta^k}{k!} q_{x,y}^{(k)},
\]
for $x,y \in \mathcal{X}$, where $q_{x,y}^{(k)}$ denotes the $(x,y)$-th entry of the matrix power $Q^k$.

The initial probability vector $\pi$ of $(X_t)_{t\ge0}$ is said to be an invariant (or stationary) distribution if $\pi^\intercal Q=0$.

A continuous-time Markov chain $X_t$ is \textit{irreducible}, if for all $x,y\in \mathcal{X}$
\[
\prob{X_{t+\delta}=y\;\middle\vert\;X_t=x}>0
\]
for all $t\ge 0,\delta>0$. By \citet[Theorem 3.2.1]{Norris1998}, a continuous-time Markov chain $X_t$ is irreducible, if and only if for any $x\neq y$ there exist states $s_0,\dots,s_n\in \mathcal{X},n\in \mathbb{N}$ with $s_0=x$ and $s_n=y$ such that
\[ 
\prod_{i=0}^{n-1}q_{s_i,s_{i+1}}(\theta)>0.
\] 

\subsection{Lumpability of Markov chains through sums}
\label{section_lump}
The lumping property (or lumpability)  describes the possibility of aggregating multiple
states of a Markov chain into a smaller state space such that the aggregated process
remains Markovian. A general treatment of lumpability for discrete-time Markov chains can
be found in \citet[Chapter~6.3]{Kemeny1976}, see also \citet{Ball1993} for the continuous-time setting.

Here, we focus on a specific form of lumping: states of a Markov chain
$X_t \in \mathcal{X} \equiv \mathcal{A}_{\kappa}^L$, with integer $\kappa \ge 2$, are aggregated according to their sum over coordinates, i.e.~$\s X_t$.

\begin{definition}[Lumping property via sums]\label{d:lump}
We say that a continuous-time Markov chain $(X_t)_{t \ge 0}$ satisfies the \emph{lumping property}
(with respect to sums) if, with $S_t = \s{X_t}$,
\[
\prob{S_{t+\delta}=s \;\middle\vert\;X_t=x}
=
\prob{S_{t+\delta}=s \;\middle\vert\;X_t=y}
\]
for all $x,y \in \mathcal{X}$ such that $\s{x}=\s{y}$ and
$\prob{X_t=x}\prob{X_t=y}>0$, and for all $s \in \mathcal{T}_L^{\kappa}$ in \eqref{e:sum:k:state}.
\end{definition}

An equivalent characterization of the lumping property is given as follows. 

\begin{prop}
\label{prop2}
Let $(X_t)_{t \ge 0}$ be a continuous-time Markov chain with state space $\mathcal{X}$, and $S_t = \mathcal{S}X_t$, $t \ge 0$, its corresponding sum process. The lumping property (Definition~\ref{d:lump}) of $(X_t)_{t \ge 0}$ is equivalent to
\[
\prob{S_{t+\delta}=s \;\middle\vert\;S_t= \s{x}}
=
\prob{S_{t+\delta}=s \;\middle\vert\;X_t=x}
\]
for any $s\in \mathcal{T}_L^{\kappa}$, and any $x \in \mathcal{X}$ such that $\prob{S_t=\s{x}}\prob{X_t=x}>0. $
\end{prop}

\begin{proof}
Assume first that the lumping property (Definition~\ref{d:lump}) holds. Then, with $s' := \s{x}$,
\begin{align*}
\prob{S_{t+\delta}=s \;\middle\vert\;S_t=s'}
&= \sum_{z:\s{z}=s'}
\prob{S_{t+\delta}=s, X_t=z \;\middle\vert\;S_t=s'} \\
&= \frac{1}{\prob{S_t=s'}}
\sum_{z:\s{z}=s'}
\prob{S_{t+\delta}=s \;\middle\vert\;X_t=z}\prob{X_t=z} \\
&= \frac{\prob{S_{t+\delta}=s \;\middle\vert\;X_t=x}}{\prob{S_t=s'}}
\sum_{z:\s{z}=s'}\prob{X_t=z} \\
&= \prob{S_{t+\delta}=s \;\middle\vert\;X_t=x}
\end{align*}
where the third equality follows from the lumping property.

The converse implication is immediate.
\end{proof}

If the lumping property is satisfied, the aggregated sum process is again a
continuous-time Markov chain.

\begin{lemma}
\label[lemma]{Lemma3.9}
If $(X_t)_{t \ge 0}$ satisfies the lumping property (Definition~\ref{d:lump}), then the sum process
$S_t=\s{X_t}$, $t \ge 0$, is a continuous-time Markov chain.
\end{lemma}

\begin{proof}
Fix arbitrarily $0 \le t_1 < t_2 < \dots < t_n$ and $s_1,\dots,s_n \in \mathcal{T}_L^{\kappa}$. Then, by Proposition~\ref{prop2}, we obtain
\begin{align*}
&\prob{S_{t_n}=s_n, S_{t_{n-1}}=s_{n-1},\dots,S_{t_1}=s_1} \\
 =\; &
\sum_{x^{(1)}:\s{x^{(1)}}=s_1}\!\!\dots\!\!\sum_{x^{(n)}:\s{x^{(n)}}=s_n}
\prob{X_{t_n}=x^{(n)},\dots,X_{t_1}=x^{(1)}} \\
 =\; &
\sum_{x^{(1)}:\s{x^{(1)}}=s_1}\!\!\dots\!\!\sum_{x^{(n)}:\s{x^{(n)}}=s_n}
\prob{X_{t_n}=x^{(n)} \;\middle\vert\;X_{t_{n-1}}=x^{(n-1)}}
\prob{X_{t_{n-1}}=x^{(n-1)},\dots,X_{t_1}=x^{(1)}} \\
= \;&
\prob{S_{t_n}=s_n \;\middle\vert\;S_{t_{n-1}}=s_{n-1}}
\prob{S_{t_{n-1}}=s_{n-1},\dots,S_{t_1}=s_1}.
\end{align*}
It implies that $\prob{S_{t_n}=s_n \;\middle\vert\;S_{t_{n-1}}=s_{n-1},\dots,S_{t_1}=s_1}
=
\prob{S_{t_n}=s_n \;\middle\vert\;S_{t_{n-1}}=s_{n-1}}$, which is the Markov property of $(S_t)_{t\ge 0}$.
\end{proof}

\subsection{Permutation invariance }
The permutation invariance of $(X_t)_{t\ge0}$ in Assumption~\ref{as1} implies the lumping property in Definition~\ref{d:lump}. In the sequel, we impose Assumption~\ref{as1} in the general setting $\mathcal{X} = \mathcal{A}_{\kappa}^L$, where the alphabet size satisfies $\kappa \ge 2$, and assume that  $\mathcal{A}_\kappa$ is $L$-informative. 
\begin{prop}\label[prop]{prop1}
    Assumption~\ref{as1} implies the lumping property (Definition~\ref{d:lump}).
\end{prop}

\begin{proof}
For any $x,y \in \mathcal{X}$ with $\s{x}=\s{y}$, there exists a permutation matrix
${\Perm}$ such that $y={\Perm}x$. Using the permutation invariance (Assumption~\ref{as1}), we obtain
\begin{align*}
\prob{S_{t+\delta}=s \;\middle\vert\;X_t=x}
&= \sum_{z:\s{z}=s}\prob{X_{t+\delta}=z \;\middle\vert\;X_t=x} = \sum_{z:\s{z}=s}\prob{X_{t+\delta}={\Perm}z \;\middle\vert\;X_t={\Perm}x} \\
&= \sum_{z:\s{z}=s}\prob{X_{t+\delta}={\Perm}z \;\middle\vert\;X_t=y} = \prob{S_{t+\delta}=s \;\middle\vert\;X_t=y}.\qedhere
\end{align*}
\end{proof}

The next result establishes the equivalence between two descriptions of permutation invariance in Assumption~\ref{as1}. 

\begin{lemma}\label[lemma]{perm_inv_char} Let $X_t$ be a continuous-time Markov chain with state space $\mathcal{A}_{\kappa}^L$, $\kappa\ge2$. Then, the following two versions of permutation invariance in Assumption \ref{as1} are equivalent:
\begin{enumerate}
\item\label{i:perm:p} For every permutation matrix ${\Perm}\in \{0,1\}^{L\times L}$, $x,y\in \mathcal{A}_{\kappa}^L$, $\delta >0$ and $t \ge 0$, $$\prob{X_{t+\delta}=y\;\middle\vert\;X_t =x}=\prob{X_{t+\delta}={\Perm}y\;\middle\vert\;X_t ={\Perm}x}.$$
\item\label{i:perm:q} For every permutation matrix ${\Perm}\in \{0,1\}^{L\times L}$ and $x,y\in \mathcal{A}_{\kappa}^L$, $q_{x,y}=q_{{\Perm}x,{\Perm}y}.$
\end{enumerate}
\end{lemma}

\begin{proof}If statement~\ref{i:perm:p} holds, then statement~\ref{i:perm:q} holds, because
\begin{align*}
    q_{x,y}& =\lim_{\delta \searrow 0}\frac{\prob{X_{t+\delta}=y\;\middle\vert\;X_t=x}-\mathbb{I}(x=y)}{\delta}\\
    &=\lim_{\delta \searrow 0}\frac{\prob{X_{t+\delta}={\Perm}y\;\middle\vert\;X_t={\Perm}x}-\mathbb{I}({\Perm}x={\Perm}y)}{\delta}=q_{{\Perm}x,{\Perm}y}.
\end{align*}
Assume now statement~\ref{i:perm:q}, and let $q_{x,y}^{(k)}$ denote the entries of $Q^k$, i.e.~$Q^k =(q_{x,y}^{(k)})_{x,y\in \mathcal{A}_{\kappa}^L }$. 
By induction, we can show that, for all $k\in \mathbb{N}$,
\begin{equation}\label{e:perm:q:k}
q^{(k)}_{x,y}=q^{(k)}_{{\Perm}x,{\Perm}y}\qquad\text{for all } x, y \in \mathcal{A}_{\kappa}^L.
\end{equation}
In fact, \eqref{e:perm:q:k} holds for $k =1$; and further if \eqref{e:perm:q:k} is valid for some $k\in \mathbb{N}$, then
\begin{equation*}
    q^{(k+1)}_{x,y}=\sum_{z\in \mathcal{A}_{\kappa}^L }q^{(k)}_{x,z}q_{z,y}=\sum_{z\in \mathcal{A}_{\kappa}^L }q^{(k)}_{{\Perm}x,{\Perm}z}q_{{\Perm}z,{\Perm}y}=q^{(k+1)}_{{\Perm}x,{\Perm}y}.
\end{equation*}
Thus, statement~\ref{i:perm:p} follows, because 
\begin{align*}
\prob{X_{t+\delta}=y\;\middle\vert\;X_t=x}& =\left[\exp(\delta Q)\right]_{x,y}=\sum_{k=0}^\infty \frac{\delta^k}{k!}q_{x,y}^{(k)}\\
& =\sum_{k=0}^\infty \frac{\delta^k}{k!}q_{{\Perm}x,{\Perm}y}^{(k)} 
    =\prob{X_{t+\delta}={\Perm}y\;\middle\vert\;X_t={\Perm}x}. \quad \qedhere
\end{align*}
\end{proof}

\clearpage

\section{Proofs of results in Section~\ref{sec1}}

\subsection{Proof of Lemma \ref{no_quant}}
\begin{proof}
    We focus primarily on Part~\ref{i:equiv:ind}, since Part~\ref{i:imp:sparse} then follows as a by-product of the argument, specifically in Case~\ref{i:case2} below.

   Consider $x,y\in \mathcal{X}$ with $\norm{x-y}_0>1$. There exist two distinct $i,j\in\{1,\dots,L\}$ such that $x_{i}\neq y_{i}$ and $x_{j}\neq y_{j}$. We have, as  $\delta \searrow 0$, 
    \[
    \prob{X_{i, t+\delta}=y_{i}\;\middle\vert\; X_t =x} = O(\delta) \quad\text{and}\quad\prob{X_{j, t+\delta}=y_j\;\middle\vert\; X_t =x} = O(\delta), 
    \]
    since $X_t$ is a continuous-time Markov chain. Then 
    \begin{multline}\label{e:rateprop}
    \frac{1}{\delta}\prod_{\ell =1}^L\prob{X_{\ell, t+\delta}=y_{\ell}\;\middle\vert\; X_t =x} \le \\
    \frac{1}{\delta}\prob{X_{i,t+\delta}=y_{i}\;\middle\vert\; X_t =x}  \prob{X_{j, t+\delta}=y_{j}\;\middle\vert\; X_t =x} = O(1)\cdot O(\delta) =o(1).
    \end{multline}
    
    \smallskip
    
    \textbf{Direction ``$\implies$'':} Assume that \Cref{as2} (conditional independence at infinitesimal times) holds. 
    Then, by \eqref{e:rateprop}, for $x,y\in \mathcal{X}$ with $\norm{x-y}_0>1$, 
\[        \frac{1}{\delta}\prob{X_{t+\delta}=y\;\middle\vert\;X_t =x} 
        =\frac{1}{\delta}\prod_{\ell =1}^L\prob{X_{\ell, t+\delta}=y_{\ell}\;\middle\vert\; X_t =x}+o(1) =o(1),
\]
    which yields  the sparse transition property (\Cref{d:stp}).

    \smallskip

\textbf{Direction ``$\impliedby$'':} Assume now that the sparse transition property holds. Let $Q = (q_{x,y})_{x,y\in \mathcal{A}_2^L}$ be the rate matrix of $(X_t)_{t\ge0}$. For any $x,y\in \mathcal{A}_2^L$, we consider three cases: 
\begin{enumerate}[label=(\alph*)]
\item \label{i:c1} Case of $\norm{x-y}_0>1$. 
	By the sparse transition property, we have  $q_{x,y}=0$, namely, 
    \[
    \frac{1}{\delta}\prob{X_{t+\delta}=y\;\middle\vert\;X_t =x}=o(1).
    \]
  On the other hand, by~\eqref{e:rateprop}, we have
    \[
    \frac{1}{\delta}\prod_{\ell =1}^L\prob{X_{\ell,t+\delta}=y_{\ell}\;\middle\vert\; X_t =x} = o(1).
    \]
\item \label{i:case2} Case of $\norm{x-y}_0=1$. 
There exists $j \in\{1,\dots,L\}$ such that $x_j\neq y_j$ while $x_{\ell}= y_{\ell}$ for all $\ell \neq j$. It holds
    \[\frac{1}{\delta}\prod_{\ell =1}^L\prob{X_{\ell,t+\delta}=y_{\ell} \;\middle\vert\; X_t =x}=\frac{1}{\delta}\prob{X_{j,t+\delta}=y_{j} \;\middle\vert\; X_t =x}(1+o(1)).\]
    Meanwhile, we have
    \begin{multline}\label{e:mcp}
        \prob{X_{t+\delta}=y\;\middle\vert\; X_t=x}\\
        =\prob{X_{j,t+\delta}=y_{j} \;\middle\vert\; X_t =x}-\sum_{z\in \mathcal{A}_2^L:z_{j}=y_{j},z\neq y }\prob{X_{t+\delta}=z \;\middle\vert\;X_t=x}.
    \end{multline}
    Note that for $z\in \mathcal{A}_2^L$ with $z_{j}=y_{j}$ and $z\neq y$, it holds $\norm{x-z}_0\ge 2$. Then, by the sparse transition property and Case~\ref{i:c1}, we have $\prob{X_{t+\delta}=z\;\middle\vert\;X_t=x} = o(\delta)$. Thus, letting $\delta\searrow 0$ in \eqref{e:mcp}, we have
    \[
    \lim_{\delta \searrow 0}\frac{1}{\delta}\prob{X_{t+\delta}=y\;\middle\vert\;X_t=x}=\lim_{\delta \searrow 0}\frac{1}{\delta}\prob{X_{j, t+\delta}=y_{j} \;\middle\vert\;X_t =x},
     \] 
     which is the assertion of Part~\ref{i:imp:sparse}. 
\item Case of $x = y$.  By the sparse transition property, it holds
    \begin{align*}
        \frac{1}{\delta}\bigl(\prob{X_{t+\delta}=x\;\middle\vert\;X_t=x}-1\big)&=-\sum_{z\neq x}\frac{1}{\delta}\prob{X_{t+\delta}=z \;\middle\vert\;X_t=x}\\
        &=-\sum_{z\in\mathcal{A}_2^L:\norm{x-z}_0 =1 }\frac{1}{\delta}\prob{X_{t+\delta}=z \;\middle\vert\;X_t=x}+o(1)\\
        &=-\sum_{m=1}^L\frac{1}{\delta}\prob{X_{t+\delta}=z^{(m)}\;\middle\vert\;X_t=x}+o(1),
    \end{align*}where  $z^{(m)}\in \mathcal{A}_2^L$ is defined as $z^{(m)}_{i}=x_i$ for all $i\neq m$ but $z^{(m)}_{m}=1-x_{m}$. 
    By Case~\ref{i:case2},  we have 
    \[q_{x,z^{(m)}}=\lim_{\delta \searrow 0}\frac{\prob{X_{t+\delta}=z^{(m)}\;\middle\vert\;X_t=x}}{\delta}=\lim_{\delta \searrow 0}\frac{\prob{X_{m,t+\delta}=1-x_{m}\;\middle\vert\;X_t=x}}{\delta} =: \eta_m.\] 
That is, we have shown that 
    \(q_{x,x}=-\sum_{m=1}^L\eta_m.\) 
    Now consider the product
    \begin{align*}
        \prod_{i=1}^L\prob{X_{i,t+\delta}=x_{i}\;\middle\vert\; X_t =x}&=\prod_{i=1}^L\left[1-\prob{X_{i,t+\delta}=1-x_{i}\;\middle\vert\;  X_t =x}\right]\\
      &=\prod_{m=1}^L\left[1-\delta \left(\eta_m +o(1)\right)\right]  =1-\sum_{m=1}^L\delta \eta_m+o(\delta),
    \end{align*}
    as $\delta \searrow 0$. 
    Thus, we have
    \begin{align*}
        &\frac{\prob{X_{t+\delta}=x\;\middle\vert\;X_t=x}-\prod_{i=1}^L\prob{X_{i,t+\delta}=x_{i}\,\middle\vert\, X_t =x}}{\delta}\\=\, &\frac{\prob{X_{t+\delta}=x\,\middle\vert\, X_t=x}-1}{\delta}-\frac{\prod_{i=1}^L\prob{X_{i,t+\delta}=x_{i}\,\middle\vert\, X_t =x}-1}{\delta} = o(1).
    \end{align*}
\end{enumerate}

Therefore, by combining all three cases above, we obtain, for any $x,y\in \mathcal{A}_2^L$,
\[
\prob{X_{t+\delta}=y\;\middle\vert\;X_t=x}
    =\prod_{i=1}^L\prob{X_{i,t+\delta}=y_{i}\,\middle\vert\, X_t =x}+o(\delta),
\]
which leads to Assumption~\ref{as2}.
\end{proof}

\subsection{Proof of Theorem~\ref{Theorem1}}
\begin{proof}
\textbf{Part~\ref{th1:i}}
Let Assumptions \ref{as1} and \ref{as2} hold. The rate matrix $Q=(q_{x,y})_{x,y\in \mathcal{X}}$ satisfies
\begin{align*}
q_{x,y}& =\lim_{\delta\searrow 0}\frac{\prob{X_{t+\delta}=y\;\middle\vert\;X_t =x}}{\delta},\qquad x\neq y\\ \text{and}\quad
q_{x,x}&=\lim_{\delta\searrow 0}\frac{\prob{X_{t+\delta}=x\;\middle\vert\;X_t =x}-1}{\delta}.
\end{align*} 
We fix arbitrarily $x,y \in \mathcal{X}$ with $x\neq y$.  By Assumption~\ref{as2} it holds 
\[
\frac{\prob{X_{t+\delta}=y\;\middle\vert\;X_t =x}}{\delta}=\frac{\prod_{i=1}^L\prob{X_{i,t+\delta}=y_{i}\;\middle\vert\;X_t =x}}{\delta}+o(1),
\]
as $\delta \searrow 0$. Because of Assumption \ref{as1}, we can apply Lemma \ref{Lemma_a1}, which yields
\[ 
\frac{\prod_{i=1}^L\prob{X_{i,t+\delta}=y_{i}\;\middle\vert\;X_t =x}}{\delta}=\frac{\prod_{i=1}^L\prob{X_{i,t+\delta}=y_{i}\;\middle\vert\;X_{i,t}=x_{i},\s{X_t} =\s{x}} }{\delta}.
\]
Since $x\neq y$, there exists at least one $j\in \{1,\dots,L\}$ such that $x_j\neq y_j$. 
\begin{itemize}
\item If exactly one such $j$ exists, we have
\begin{multline*}
\lim_{\delta\searrow 0}\frac{\prod_{i=1}^L\prob{X_{i,t+\delta}=y_{i}\;\middle\vert\;X_{i,t}=x_{i},\s{X_t} =\s{x}} }{\delta} \\
=\lim_{\delta\searrow 0}\frac{\prob{X_{j,t+\delta}=y_j\;\middle\vert\;X_{j,t}=x_j,\s{X_t} =\s{x}} }{\delta}.
\end{multline*}
\item
If there is more than one such $j$, then, based on the above calculation and Lemma~\ref{no_quant}\ref{i:equiv:ind}, 
\[
\lim_{\delta\searrow 0}\frac{\prod_{i=1}^L\prob{X_{i,t+\delta}=y_{i}\;\middle\vert\;X_{i,t}=x_{i},\s{X_t} =\s{x}} }{\delta} = \lim_{\delta\searrow 0} \frac{\prod_{i=1}^L\prob{X_{i,t+\delta}=y_{i}\;\middle\vert\;X_t =x}}{\delta} = 0.
\]
\end{itemize}
Further, define the rates of instantaneous transitions $\lambda_{s}$ and $\mu_{s}$ as in \eqref{param:12}. We \emph{claim} that the limits in \eqref{param:12} are independent of the choices of $t$ and $j$. Under this claim, we see directly that for $x\neq y$
\[
q_{x,y}=\begin{cases}
        \lambda_{\s{x}}\mathbb{I}(x_{i}=0)+\mu_{\s{x}}\mathbb{I}(x_{i}=1),& \norm{x-y}_0=1,y_{i}\neq x_{i}\\
        0,&\norm{x-y}_0>1,
    \end{cases}
 \]
and thus $X_t$ is an SDMC with parameters $\lambda_0,\dots,\lambda_{L-1},\mu_1,\dots,\mu_L$ in $[0,\infty)$. 

We now show the afore-mentioned claim. By Lemma~\ref{Lemma_a1}, it holds, for any $x\in \mathcal{A}_2^L$ with $\s{x}=s$ and $x_j=a$, 
\begin{align*}
    \prob{X_{j,t+\delta}=b\;\middle\vert\;X_{j,t}=a,\s{X_t} =s} &=\prob{X_{j,t+\delta}=b\;\middle\vert\;X_t=x}\\
    &=\sum_{y\in \mathcal{A}_2^L:y_j=b}\prob{X_{t+\delta}=y\;\middle\vert\;X_t=x}.
\end{align*}
Thus, the limits  in \eqref{param:12} are independent of $t$, because $(X_t)_{t\ge 0}$ is time homogeneous. 

Next consider arbitrarily $a,b\in \mathcal{A}_2 = \{0,1\}$ with $a\neq b$, which implies $ b  =1 -a$. Again by Lemma~\ref{Lemma_a1}, it holds,  for any $x\in \mathcal{A}_2^L$ with $x_j=a$ and $\s{x}=s$,
\begin{equation}\label{e:pre1}
\prob{X_{j,t+\delta}=b\;\middle\vert\;X_{j,t}=a,\s{X_t} =s}=\prob{X_{j,t+\delta}=b\;\middle\vert\;X_t=x}.
\end{equation}
Define $y\in \mathcal{A}_2^L$ by $y^{(k)}=x^{(k)}$ for all $k\neq j$ and $y_j=1-x_j=b$. That is, $\norm{x-y}_0=1$ and $y_j\neq x_j$. By Lemma~\ref{no_quant}\ref{i:imp:sparse}, we have 
\begin{equation}\label{e:pre2}
\lim_{\delta\searrow 0} \frac{\prob{X_{j,t+\delta}=b\;\middle\vert\;X_t=x}}{\delta}=\lim_{\delta\searrow 0} \frac{\prob{X_{t+\delta}=y\;\middle\vert\;X_t=x}}{\delta}.\end{equation}
Let ${\Perm}$ be a permutation matrix that permutes the $j$-th and $i$-th entry.  Note that $({\Perm}y)_i=b=1-a=1-({\Perm}x)_i$ and $\norm{{\Perm}y-{\Perm}x}_0=1$. By Assumption~\ref{as1}, it holds
\begin{align*}
    \lim_{\delta\searrow 0} \frac{\prob{X_{t+\delta}=y\;\middle\vert\;X_t=x}}{\delta}&=\lim_{\delta\searrow 0} \frac{\prob{X_{t+\delta}={\Perm}y\;\middle\vert\;X_t={\Perm}x}}{\delta}=\lim_{\delta\searrow 0} \frac{\prob{X_{i,t+\delta}=b\;\middle\vert\;X_t={\Perm}x}}{\delta}\\
    &=\lim_{\delta\searrow 0} \frac{\prob{X_{i,t+\delta}=b\;\middle\vert\;X_{i,t}=a,\s{X_t} =s}}{\delta},
\end{align*}
which in combination of \eqref{e:pre1} and \eqref{e:pre2} implies 
\[
\lim_{\delta\searrow 0}\frac{\prob{X_{j,t+\delta}=b\;\middle\vert\;X_{j,t}=a,\s{X_t} =s} }{\delta} =\lim_{\delta\searrow 0}\frac{\prob{X_{i,t+\delta}=b\;\middle\vert\;X_{i,t}=a,\s{X_t} =s} }{\delta}.  
\]
Thus, we have proven the claim that the limits in \eqref{param:12} are independent of $t$ and $j$.

\smallskip

\textbf{Part~\ref{th1:ii}}
Now assume that $X_t$ is a SDMC with parameters $\lambda_0,\dots,\lambda_{L-1}$ and $\mu_1,\dots,\mu_L$ in $[0,\infty)$. We will show that Assumptions \ref{as1}--\ref{as2} hold. By definition of an SDMC, it satisfies the sparse transition property and moreover, for any permutation matrix ${\Perm}$
\[
q_{x,y}=q_{{\Perm}x,{\Perm}y}.
\]
By Lemmas~\ref{no_quant} and~\ref{perm_inv_char}, we have Assumptions~\ref{as1} and~\ref{as2}. It then follows that the parameters $\lambda_{s}$, $\mu_{s}$ satisfy  \eqref{param:12} and the limits therein are are independent of $t\ge 0$ and $j\in\{1,\dots,L\}$. 
\end{proof}

\subsection{Proof of Lemma~\ref{l:ir:sd}}

\begin{proof}
By \citet[Theorem~3.2.1]{Norris1998}, the irreducibility of $(X_t)_{t\ge0}$ holds if and only if for any distinct $x,y\in \mathcal{X}$ there exist a sequence of states $x^{(0)},x^{(1)},\ldots, x^{(n)}$ with $x^{(0)} = x$ and $x^{(n)} = y$ such that $q_{x^{(0)},x^{(1)}}q_{x^{(1)},x^{(2)}}\cdots q_{x^{(n-1)},x^{(n)}}>0$. This is clearly the case (due to Definition~\ref{def_rate}) if and only if all parameters $\lambda_0,\ldots,\lambda_{L-1},\mu_1,\ldots,\mu_L$ are strictly positive. 
\end{proof}

\subsection{Proof of Proposition~\ref{reversible}}

\begin{proof}
By~\citet[Theorem~3.7.3]{Norris1998}, $(X_t)_{t\ge 0}$ is {reversible} if and only if its rate matrix $Q$ and some initial distribution $\pi$ satisfy {detailed balance} equations, i.e.
\begin{equation}\label{e:db}
\pi_x q_{x,y}=\pi_y q_{y,x}\qquad\text{for all }x,y\in\mathcal{A}_2^L.
\end{equation}
To find a possible $\pi$ of such, we only need to consider the case of $\norm{x-y}_0=1$, since $q_{x,y}=q_{y,x}=0$ for $\norm{x-y}_0>1$ by Lemma~\ref{no_quant}\ref{i:equiv:ind}. For $x\in\mathcal{A}_2^L$, we introduce
\begin{subequations}
\begin{align}
 A^+_{x} &:= \{y \in \mathcal{A}_2\, :\,  \s{y} = \s{x}+1 \text{ and }\norm{y-x}_0=1\}, \label{e:neighbor:1}\\
 \text{and}\quad
 A^-_{x} &:= \{y \in \mathcal{A}_2\, :\,  \s{y} = \s{x}-1 \text{ and }\norm{y-x}_0=1\}. \label{e:neighbor:2}
\end{align}
\end{subequations}
Note that $A^+_{x} \cup A^-_{x} = \{z \in \mathcal{A}_2\, :\, \norm{z-x}_0=1 \}$. The detailed balance equations in \eqref{e:db} are equivalent to
\[
\begin{cases}
\pi_x \lambda_{\s{x}} = \pi_y \mu_{\s{x}+1} &\text{ if } y \in A_x^+,\\
\pi_x \mu_{\s{x}} = \pi_y\lambda_{\s{x}-1} & \text{ if } y \in A_x^-.
\end{cases}
\]
It is easy to check $\pi = \pi^*$ is a solution to the above equations. Further, by the irreducibility of $(X_t)_{t\ge0}$ and \citet[Lemma~3.7.2]{Norris1998},  this $\pi^*$ is the unique solution to the detailed balance equations in \eqref{e:db}, and it is the unique invariant distribution of $(X_t)_{t\ge0}$.
\end{proof}

\subsection{Proof of Theorem \ref{Theorem2}}
\begin{proof}
\textbf{Part~\ref{i:th2:1}} By Theorem~\ref{Theorem1}\ref{th1:ii}, Assumption \ref{as1} is satisfied. Then, by Lemma~\ref{Lemma3.9} and Proposition~\ref{prop1}, $(S_t)_{t\ge0}$ is a continuous-time Markov chain.
 
\smallskip

\textbf{Part~\ref{i:th2:2}} Fix any $x\in \mathcal{A}_2^L$ and let $i = \s{x}$. By Proposition~\ref{prop2}, we have 
\[
r_{i,j}:=\lim_{\delta \searrow 0}\frac{\prob{S_{t+\delta}=j\;\middle\vert\;S_t=i}}{\delta}=\lim_{\delta \searrow 0}\frac{\prob{S_{t+\delta}=j\;\middle\vert\;X_t=x}}{\delta}=\sum_{y:\s{y}=j}q_{x,y}.
\]
By Lemma~\ref{no_quant}\ref{i:equiv:ind}, $(X_t)_{t\ge 0}$ satisfies the sparse transition property (Definition~\ref{d:stp}). Thus,
\[
\sum_{y:\s{y}=j}q_{x,y}=\sum_{y:\s{y}=j,\norm{x-y}_0=1}q_{x,y} \qquad\text{for}\quad j \neq i,
\]
which implies $r_{i,j}=0$ when $\abs{i-j}>1$ and 
\[
r_{i,i+1}=\sum_{y:\s{y}=i+1,\norm{x-y}_0=1}q_{x,y}\quad\text{and}\quad r_{i,i-1}=\sum_{y:\s{y}=i-1,\norm{x-y}_0=1}q_{x,y}.
\]
For $y \in A^+_{x}$ in \eqref{e:neighbor:1}, it holds $q_{x,y} = \lambda_i$. Note that the cardinality of $A^+_{x}$ is $L-i$. Thus, $r_{i,i+1} = (L-i)\lambda_i$.
Similarly, for $y \in A^-_{x}$ in \eqref{e:neighbor:2}, it holds $q_{x,y} = \mu_i$. As the cardinality of $A^-_{x}$ is $i$, we have $r_{i,i-1} = i\mu_i$.
The formula of $r_{i,i}$ follows from the fact that the row sums of $R$ are zero. 

\smallskip

\textbf{Part~\ref{i:th2:3}} It follows since by~\eqref{Theorem2_eq}, 
\[
\lambda_{i-1} = \frac{r_{i-1,i}}{L-i+1}\quad\text{and}\quad \mu_i = \frac{r_{i,i-1}}{i}, \qquad i \in\{1,\ldots, L\}. \qedhere
\]
\end{proof}

\subsection{Proof of Proposition \ref{bd_proc}}\label{ss:proof:bd_proc}

\begin{proof}
\textbf{Part~\ref{i:bd_proc:1}} We adopt a similar proof as Lemma~\ref{l:ir:sd}. By \citet[Theorem~3.2.1]{Norris1998}, the irreducibility of $(S_t)_{t\ge0}$ holds if and only if for any distinct $i,j\in \{0,1,\ldots,L\}$ there exist a sequence of states $s_0,s_1,\ldots, s_n$ with $s_0 = i$ and $s_n = j$ such that $r_{s_0,s_1}r_{s_1,s_2}\cdots r_{s_{n-1},s_n}>0$. By the structure of $R$ in~\eqref{Theorem2_eq} of Theorem~\ref{Theorem2}\ref{i:th2:2}, this is clearly the case if and only if all parameters $\lambda_0,\ldots,\lambda_{L-1},\mu_1,\ldots,\mu_L$ are strictly positive, which by Lemma~\ref{l:ir:sd} is equivalent to the reversibility of $(X_t)_{t\ge0}$.

\smallskip

\textbf{Part~\ref{i:bd_proc:2}} Similar to the proof of Proposition~\ref{reversible}, we employ the sparsity (i.e.\ $r_{i,j}=0$ for $\abs{i-j}>1$) of the rate matrix $R$ in \eqref{Theorem2_eq}, and reduce the detailed balance equations to 
\[
\pi_j r_{j,j-1}=\pi_{j-1}r_{j-1,j},\qquad j \in \{1,\ldots, L\}.
\]
Clearly, $\pi = \pi^{\star}$ is a solution to the above equations. By the irreducibility, and \citet[Lemma~3.7.2]{Norris1998}, $ \pi^{\star}$ is the unique solution to the detailed balance equations and is the unique invariant distribution of $(S_t)_{t\ge0}$.
\end{proof}

\subsection{Supporting lemmas}
The next result is a generalized version of Lemma~A.4 from \citet{VanEt24}.

\begin{lemma}\label{LemmaE4} 
Let $(\Omega,\mathcal{A},\mathbb{P})$ be a probability space and let $B_1,\dots,B_k\in \mathcal{A}$, $k \in \mathbb{N}$, be events such that  $B_1,\dots ,B_k$ are disjoint and $\prob{B_i}>0$ for all $i\in\{1,\dots,k\}$. Assume that $A\in \mathcal{A}$ and $\prob{A\;\middle\vert\;B_i}=\prob{A \;\middle\vert\;B_j}$ for all $i,j\in \{1,\dots,k\}$. Then,
    \[
    \prob{A\;\middle\vert\;\bigcup_{i=1}^k B_i}=\prob{A\;\middle\vert\;B_j}\quad\text{for every }j \in \{1,\ldots,k\}.
    \]
\end{lemma}

\begin{proof}
    The general case follows from the case $k=2$ by induction. For $k=2$, it holds
    \begin{align*}
        \prob{A\;\middle\vert\;B_1\cup B_2}&=\frac{\prob{(A\cap B_1 )\cup (A\cap B_2 )}}{\prob{B_1\cup B_2}}=\frac{\prob{A\cap B_1 }+\prob{A\cap B_2 }}{\prob{B_1\cup B_2}}\\
        &=\frac{\prob{A\cap B_1 }\prob{B_1}}{\prob{B_1}\prob{B_1\cup B_2}}+\frac{\prob{A\cap B_2 }\prob{B_2}}{\prob{B_2}\prob{B_1\cup B_2}}\\
        &=\frac{\prob{A\;\middle\vert\;B_1 }\prob{B_1}}{\prob{B_1\cup B_2}}+\frac{\prob{A\;\middle\vert\;B_2 }\prob{B_2}}{\prob{B_1\cup B_2}}\\
        &=\prob{A\;\middle\vert\;B_1 }\frac{\prob{B_1}+\prob{B_2}}{\prob{B_1\cup B_2}}=\prob{A\;\middle\vert\;B_1 }=\prob{A\;\middle\vert\;B_2 }. \quad\qedhere
    \end{align*}
\end{proof}

We extend Lemma A.6 from \cite{VanEt24} to the continuous-time setting.

\begin{lemma}\label[lemma]{Lemma_a1}
   Let $(X_t)_{t\ge 0}$ be a continuous-time Markov chain that satisfies Assumption~\ref{as1}.
Then, for any $i\in\{1,\dots,L\}$, $x\in \mathcal{A}_2^L$, $a\in \mathcal{A}_2$, $t\ge0$ and $\delta>0$, it holds
   \[\prob{X_{i,t+\delta}=a\;\middle\vert\;X_{t}=x }=\prob{X_{i,t+\delta}=a\;\middle\vert\;X_{i,t}=x_{i},\s{X_t}=\s{x} }.\]
\end{lemma}

\begin{proof}
    Let 
    \(
    F(x,i):=\{z\in\mathcal{A}_2^L:\s{z}=\s{x},z_i=x_{i} \}.
    \) 
    For any $z\in F(x,i)$ there exists a permutation matrix ${\Perm}\in \mathcal{A}_2^{L\times L}$ such that $z={\Perm}x$ and $z_i=x_{i}$. By Assumption~\ref{as1} (permutation invariance) with the afore-introduced ${\Perm}$, we have
\begin{align*}
    \prob{A\;\middle\vert\;B_x}&=\prob{X_{i,t+\delta}=a\;\middle\vert\;X_{t}=x } = \sum_{y\in\mathcal{A}_2^L:y_{i}=a }\prob{X_{t+\delta}=y\;\middle\vert\;X_{t}=x }\\
    &=\sum_{y\in\mathcal{A}_2^L:y_{i}=a }\prob{X_{t+\delta}={\Perm}y\;\middle\vert\;X_{t}={\Perm}x }=\sum_{y\in\mathcal{A}_2^L:y_{i}=a }\prob{X_{t+\delta}=y\;\middle\vert\;X_{t}=z }\\
     &= \prob{X_{i,t+\delta}=a\;\middle\vert\;X_{t}=z }
    =\prob{A\;\middle\vert\;B_z}.
\end{align*}
As
\(
\bigcup_{z\in F(x,i)}B_z= \{\s{X_t}=\s{x},X_{i,t}=x_{i} \}, 
\)
the assertion follows from Lemma~\ref{LemmaE4} with $A=\{X_{i,t+\delta}=a\}$ and $B_z=\{X_t=z\}$ for $z \in F(x,i)$.
\end{proof}

\clearpage

\section{Proofs of results in Section~\ref{cooperativity}}

\subsection{Proof of Proposition~\ref{prop_rel_null_coop_ind}}

\begin{proof}
We will show that the null cooperativity (i.e.\ statement~\ref{i:null}) of $(X_t)_{t\ge0}$ is equivalent to the following property:
\begin{equation}\label{Prop4.1:eq1}
 \text{For all $a,b\in \mathcal{A}_2$,}\quad \sum_{y \in\mathcal{A}_2^{L-1}} q_{(a,x),(b,y)} \quad\text{is independent of } x \in \mathcal{A}_2^{L-1}. 
\end{equation}
Due to the permutation invariance of $(X_t)_{t\ge0}$ (Assumption~\ref{as1} and Theorem~\ref{Theorem1}\ref{th1:ii}), analogous properties to~\eqref{Prop4.1:eq1} hold for any coordinate, not only the first one, whenever~\eqref{Prop4.1:eq1} is satisfied. Consequently, by \citet[Theorem~3.2]{Ball1993} and the sparse transition property of $(X_t)_{t\ge0}$ (Lemma~\ref{no_quant} and Theorem~\ref{Theorem1}\ref{th1:ii}), property~\eqref{Prop4.1:eq1} is equivalent to statement~\ref{i:ind}, and by \citet[Theorem~3.1]{Ball1993}, it is also equivalent to statement~\ref{i:marg}.

We employ the structure of rate matrix $Q$ in Definition~\ref{def_rate}, and obtain, for $x\in \mathcal{A}_2^{L-1}$,
\begin{align}\label{Prop4.1:eq2}
     \sum_{y \in\mathcal{A}_2^{L-1}} q_{(a,x),(b,y)}=
     \begin{cases}
       \mu_{\s{(a,x)}},&\text{ if } a=1,b=0,\\
         \lambda_{\s{(a,x)}},&\text{ if } a=0,b=1,\\
        -\mu_{\s{(a,x)}},&\text{ if }  a=b=0, \\
       -\lambda_{\s{(a,x)}},&\text{ if }  a=b=1.
     \end{cases}
\end{align}

Thus, if $(X_t)_{t\ge0}$ is null cooperative, then the right hand side of~\eqref{Prop4.1:eq2} is equal to 
\[
\lambda_0\bigl[\mathbb{I}(a=0)-\mathbb{I}(a=1)\bigr]\mathbb{I}(b=1)+\mu_1\bigl[\mathbb{I}(a=1)-\mathbb{I}(a=0)\bigr]\mathbb{I}(b=0),
\]
and is thus independent of $x \in \mathcal{A}_2^{L-1}$, i.e., \eqref{Prop4.1:eq1} is satisfied. 

Conversely, if \eqref{Prop4.1:eq1} holds, then, by~\eqref{Prop4.1:eq2}, 
$\lambda_{\s(0,x)}, \lambda_{\s{(1,x)}}, \mu_{\s{(0,x)}}, \mu_{\s{(1,x)}}$ remain constant for $x \in \mathcal{A}_2^{L-1}$. Then, by the permutation invariance of $(X_t)_{t\ge0}$, we have
\[
\lambda_ 0 =\cdots = \lambda_{L-1}\quad\text{and}\quad\mu_1 = \cdots = \mu_{L},
\]
that is, $(X_t)_{t\ge0}$ is null cooperative. 
\end{proof}

\subsection{Proof of \Cref{le:mon:prob}}

\begin{proof}
By time-homogeneity, permutation invariance and \Cref{Lemma_a1}, for any $i\in\{1,\ldots, L\}$ and any $x\in\{0,1\}^L$
with $x_i=0$ and $\s x=s$, it holds
\[
p^{\mathrm{o}}_s(\delta)
 =
 \prob{X_{i,\delta}=1\mid X_{0} = x},
\]
which is independent of $i$ and $t$. Similarly, if $x_i=1$ and $\s x=s$, then
\[
p^{\mathrm{c}}_s(\delta)
 =
 \prob{X_{i,\delta}=0\mid X_0 = x},
\]
thus independent of $i$ and $t$. For vectors $x$ and $y$ of the same dimension, we write $x\le y$ to denote the coordinatewise partial order, namely, $x \le y$ if and only if  $x_j \le y_j$ for every coordinate $j$.

\begin{enumerate}
\item
We first consider the fully positively cooperative case
\[
\lambda_0<\lambda_1<\cdots<\lambda_{L-1},
\qquad
\mu_1>\mu_2>\cdots>\mu_L.
\]
We construct a  joint continuous-time Markov chain \((X_t, X_t')_{t\ge0}\) on the state space
\[
E:=\left\{(x,y):x\leq y\text{ and } x,y\in\{0,1\}^L\right\},
\]
by defining its rate matrix $\Gamma$ as follows. For each coordinate
$j\in\{1,\ldots, L\}$, write $e_j \in \mathbb{R}^L$ for the $j$-th unit vector. If $x_j=y_j=0$, set
\[
\Gamma((x,y),(x+e_j,y+e_j))=\lambda_{\s x},
\qquad
\Gamma((x,y),(x,y+e_j))=\lambda_{\s y}-\lambda_{\s x}.
\]
If $x_j=y_j=1$, set
\[
\Gamma((x,y),(x-e_j,y-e_j))=\mu_{\s y},
\qquad
\Gamma((x,y),(x-e_j,y))=\mu_{\s x}-\mu_{\s y}.
\]
If $x_j=0$ and $y_j=1$, set
\[
\Gamma((x,y),(x+e_j,y))=\lambda_{\s x},
\qquad
\Gamma((x,y),(x,y-e_j))=\mu_{\s y}.
\]
All other off-diagonal entries are zero, and the diagonal entries are chosen
so that the rows sum to zero. Clearly, $\Gamma$ is a rate matrix (\Cref{sec_ct}).  The rate matrix $\Gamma$ satisfies the lumping property required by \citet[Theorem~2.4]{Ball1993}, with respect to the sets 
\[
\mathcal{J}_{x}=\{(x,y)\in E:y\in \{0,1\}^L\}\quad\text{and}\quad \mathcal{G}_y=\{(x,y)\in E:x\in \{0,1\}^L\},
\] 
and hence both marginal processes $(X_t)_{t\ge 0}$ and $(X_t')_{t\ge0}$ are continuous Markov chains, and have the same rate matrix as the original SDMC \citep[Theorem 2.11]{Tian2006}, defined in~\eqref{def_rate_eq}. By construction, 
\begin{multline}\label{e:order}
\prob{X_t\le  X'_t\ \text{for all }t\ge0 \mid X_0 = x, X'_0 = y}\\
 = \prob{(X_t,X_t')\in E\ \text{for all }t\ge0 \mid (X_0, X_0') = (x, y)}=1,
\end{multline}
for all $(x,y) \in E$.

As $\Gamma$ satisfies the lumping property with respect to $\{\mathcal J_x:x\in \{0,1\}^L\}$ and $\{\mathcal G_y:y\in \{0,1\}^L\}$, Theorem~2.11 in \citet{Tian2006} implies
\begin{align*} 
\prob{X_{\delta} = z \mid X_0 = x} &=\sum_{w\in \{0,1\}^L} \prob{X_{\delta} = z,X_\delta'=w \mid X_0 = x, X'_0 = y}, \\
\text{and } \prob{X'_{\delta} = w \mid X'_0 = y}& =\sum_{z\in \{0,1\}^L} \prob{X_{\delta} = z,X_\delta'=w\mid X_0 = x, X'_0 = y},
\end{align*}
for $x,y$ with $(x,y)\in E$  and $z, w\in \{0,1\}^L$. Hence,
\begin{align*}
    \prob{X_{i,\delta} = 1 \mid X_0 = x} &= \prob{X_{i,\delta} = 1 \mid X_0 = x, X'_0 = y}, \\
    \text{and } \prob{X'_{i,\delta} = 1 \mid X'_0 = y}&=\prob{X'_{i,\delta} = 1 \mid X_0 = x, X'_0 = y}.
\end{align*}

Fix arbitrarily $i \in \{1,\ldots, L\}$ and $s\in \{0, \ldots, L-2\}$, and choose $x,y \in \mathcal A_2^L$ such that 
\[
x\le y,\qquad x_i=y_i=0,\qquad \s{x}=s,\qquad \s{y}=s+1.
\]
Then by \eqref{e:order}
\begin{multline*}
p^{\mathrm{o}}_s(\delta) = \prob{X_{i,\delta} = 1 \mid X_0 = x} = \prob{X_{i,\delta} = 1 \mid X_0 = x, X'_0 = y}\\
 \le \prob{X'_{i,\delta} = 1 \mid X_0 = x, X'_0 = y} = \prob{X'_{i,\delta} = 1 \mid X'_0 = y}  = p^{\mathrm{o}}_{s+1}(\delta).
\end{multline*}
The above inequality is in fact strict, because the corresponding transition rate 
\[
\Gamma((x,y),(x,y+e_i))
 =
 \lambda_{s+1}-\lambda_s
 >
0,
\]
see \citet[Theorem~3.2.1]{Norris1998}. Thus, 
\[
p^{\mathrm{o}}_0(\delta)<p^{\mathrm{o}}_1(\delta)<\cdots<p^{\mathrm{o}}_{L-1}(\delta).
\]

Similarly, for $i \in \{1,\ldots,L\}$ and $s \in \{1,\ldots, L-1\}$, choose $x, y \in  \mathcal A_2^L$  such that 
\[
x\le y,\qquad x_i=y_i=1,\qquad \s{x}=s,\qquad \s{y}=s+1.
\]
Then, because of \eqref{e:order} and $\Gamma((x,y),(x-e_i,y))
=
\mu_s-\mu_{s+1}
>0$, we have
\begin{multline*}
p^{\mathrm{c}}_s(\delta) = \prob{X_{i,\delta} = 0 \mid X_0 = x} = \prob{X_{i,\delta} = 0 \mid X_0 = x, X'_0 = y}\\
 > \prob{X'_{i,\delta} = 0 \mid X_0 = x, X'_0 = y} = \prob{X'_{i,\delta} = 0 \mid X'_0 = y}  = p^{\mathrm{c}}_{s+1}(\delta),
\end{multline*}
which concludes the proof for the fully positively cooperative case. 

\smallskip

\item
We next consider the fully negatively cooperative case
\[
\lambda_0>\lambda_1>\cdots>\lambda_{L-1},
\qquad
\mu_1<\mu_2<\cdots<\mu_L.
\]
Fix arbitrarily $i \in \{1,\ldots, L\}$, and define 
\[
Z_t^{(i)}
 :=
 \bigl(X_{i,t},\, L-1-S_t+X_{i,t}\bigr),\qquad t\ge 0,
\]
where the second component counts the number of zero coordinates among the
coordinates different from $i$. By lumpability \citep[Theorems~2.3 and~2.4]{Ball1993}, the process $\bigl(Z_t^{(i)}\bigr)_{t\ge0}$ is Markov, and its rate matrix $\widetilde Q$ is given by
\begin{align*}
\widetilde q((a,b),(1,b)) &
 =
 \mathbb{I}(a=0)\lambda_{L-1-b}, \\
\widetilde q((a,b),(0,b))
 & =
 \mathbb{I}(a=1)\mu_{L-b}, \\
 \widetilde q((a,b),(a,b+1)) &
 =
 (L-1-b)\mu_{L-1-b+a}, \\
 \widetilde q((a,b),(a,b-1)) &
 =
 b\lambda_{L-1-b+a},
\end{align*}
for $(a,b) \in \{0,1\}\times\{0,\ldots,L-1\}$, with all other off-diagonal entries equal to zero. 

We next construct a joint Markov chain $(Z_t, Z'_t)_{t\ge0}$ on the state space
\[ 
E=\bigl\{(z, z'):z\le  z'\text{ and } z,z'\in \{0,1\}\times \{0,\dots,L-1\}\bigr\}
\]
such that both marginal process $(Z_t)_{t\ge0}$ and $(Z'_t)_{t\ge0}$ are Markov and share the same rate matrix $\widetilde Q$, as defined above. More precisely, we define the rate matrix $\Gamma$ for the joint chain $(Z_t, Z'_t)_{t\ge0}$  on $E$ as follows
\begin{align*}
\Gamma\bigl((z,z'),(z+h,z' +h)\bigr) & = \min\{\widetilde{q}(z, z+h), \, \widetilde{q}(z', z'+h)\},\\
\Gamma\bigl((z,z'),(z+h,z')\bigr) & = \mathbb{I}(z+h \le z') \max\{\widetilde q(z,z+h)-\widetilde q(z', z'+h),\, 0\},\\
\Gamma\bigl((z,z'),(z,z'+h)\bigr) & = \mathbb{I}(z \le z'+h) \max\{\widetilde q(z',z'+h)-\widetilde q(z, z+h),\, 0\},
\end{align*}
for $(z,z')\in E$ and $h \in \{(1,0),\, (-1, 0),\, (0, 1),\, (0,-1)\}$. All other off-diagonal entries of $\Gamma$ are zero, and the diagonal entries are chosen
so that the rows sum to zero. By \citet[Theorems~3.1]{Ball1993}, we can verify that indeed both marginals are Markov and have the rate matrix $\widetilde Q$. By construction, for $(z, z') \in E$,
\begin{equation}\label{e:order:2}
\prob{Z_t \le Z_t' \text{ for all }t\ge0 \mid Z_0 = z, Z'_0 = z'} = 1.
\end{equation}

As $\Gamma$ satisfies the lumping property, it follows from \citet[Theorem 2.11]{Tian2006} that 
\begin{align*}
\prob{Z_{\delta} = w \mid Z_0 = x}=\sum_{z\in \{0,1\}\times \{0,\dots,L-1\}} \prob{Z_{\delta} = w,Z_\delta'=z \mid Z_0 = x, Z'_0 = y},\\
\text{and } \prob{Z'_{\delta} = z \mid Z'_0 = y}=\sum_{w\in \{0,1\}\times \{0,\dots,L-1\}} \prob{Z_{\delta} = w,Z_\delta'=z \mid Z_0 = x, Z'_0 = y}.
\end{align*}
for $x,y$ with $(x,y)\in E$ and $w,z\in \{0,1\}\times \{0,\dots,L-1\}$. Hence,
\begin{align*}
    \prob{Z_{1,\delta} = 1 \mid Z_0 = x} &= \prob{Z_{1,\delta} = 1 \mid Z_0 = x, Z'_0 = y}, \\ \text{and } \prob{Z'_{1,\delta} = 1 \mid Z'_0 = y}&=\prob{Z'_{1,\delta} = 1 \mid Z_0 = x, Z'_0 = y} .
\end{align*}

To compare the opening probabilities, we fix arbitrarily $s\in \{0, \ldots, L-2\}$ and set 
\[
z = (0, L-2-s), \qquad z' = (0, L-1-s).
\]
Then $z \le z'$ and by \Cref{Lemma_a1}
\begin{multline*}
p^{\mathrm{o}}_{s+1}(\delta) = \prob{Z_{1,\delta}=1 \mid Z_{0}  = z } = \prob{Z_{1,\delta}=1 \mid Z_{0}  = z, Z'_0 = z' }  \\
<  \prob{Z'_{1,\delta}=1 \mid Z_{0}  = z, Z'_0 = z' } = \prob{Z'_{1,\delta}=1 \mid Z'_0 = z' } = p^{\mathrm{o}}_{s}(\delta), 
\end{multline*}
where the inequality follows by \eqref{e:order:2}, and its strictness comes from  
$$
\Gamma\bigl((z, z'),(z, z'+ (1,0))\bigr) = \lambda_s - \lambda_{s+1} >0
$$
and \citet[Theorem~3.2.1]{Norris1998}. Thus, $p^{\mathrm{o}}_0(\delta)>p^{\mathrm{o}}_1(\delta)>\cdots>p^{\mathrm{o}}_{L-1}(\delta)$.

Analogously, for the closing probabilities, we consider $s\in \{1,\ldots, L-1\}$ and set 
\[
z=(1,\,L-1-s),
\qquad
z'=(1,\, L-s).
\]
Then $z \le z'$ and by \Cref{Lemma_a1}
\begin{multline*}
p^{\mathrm{c}}_{s+1}(\delta) = \prob{Z_{1,\delta}=0 \mid Z_{0}  = z } = \prob{Z_{1,\delta}=0 \mid Z_{0}  = z, Z'_0 = z' }  \\
> \prob{Z'_{1,\delta}=0 \mid Z_{0}  = z, Z'_0 = z' } = \prob{Z'_{1,\delta}=0 \mid Z'_0 = z' } = p^{\mathrm{c}}_{s}(\delta), 
\end{multline*}
where the inequality follows by \eqref{e:order:2}, and its strictness comes from  
$$
\Gamma\bigl((z, z'),(z - (1,0), z')\bigr) = \mu_{s+1} - \mu_{s} >0
$$
and \citet[Theorem~3.2.1]{Norris1998}. Thus, $p^{\mathrm{c}}_1(\delta)<p^{\mathrm{c}}_2(\delta)<\cdots<p^{\mathrm{c}}_{L}(\delta)$.

\smallskip

\item
It remains to consider the null cooperative case. By \Cref{prop_rel_null_coop_ind}, the coordinate processes $(X_{1,t})_{t\ge0},\dots,(X_{L,t})_{t\ge0}$ are independent continuous-time Markov chains. Hence, by \Cref{Lemma_a1},
\begin{align*}
p^{\mathrm{o}}_s(\delta) & 
 =
 \mathbb P\{X_{i,t+\delta}=1\mid X_{i,t}=0,\ \s X_t=s\}
 =
 \mathbb P\{X_{i,t+\delta}=1\mid X_{i,t}=0\},\\
 p^{\mathrm{c}}_s(\delta) &
 =
 \mathbb P\{X_{i,t+\delta}=0\mid X_{i,t}=1,\ \s X_t=s\}
 =
 \mathbb P\{X_{i,t+\delta}=0\mid X_{i,t}=1\}.
\end{align*}
As a consequence, both $p^{\mathrm{o}}_s(\delta)$ and $p^{\mathrm{c}}_s(\delta)$ remain constant for different $s$.\qedhere
\end{enumerate}
\end{proof}

\subsection{Proof of Proposition \ref{prop3}}
\begin{proof}
Note that 
\begin{multline*}
         \Lambda(\theta)=\frac{\sum_{s=1}^{L-1}\sum_{r=s+1}^L \mathbb{I}(\mu_r < \mu_s)+\sum_{s=0}^{L-2}\sum_{r=s+1}^{L-1}\mathbb{I}(\lambda_{r}>\lambda_{s}) }{L(L-1)}
        \\ -\frac{\sum_{s=1}^{L-1}\sum_{r=s+1}^L \mathbb{I}(\mu_r >\mu_s)+\sum_{s=0}^{L-2}\sum_{r=s+1}^{L-1}\mathbb{I}(\lambda_{r}<\lambda_{s})} {L(L-1)}.
\end{multline*}
Thus, for part~\ref{prop3.1}, we have $\Lambda(\theta)=1$ if and only if 
\[
\frac{\left(\sum_{s=1}^{L-1}\sum_{r=s+1}^L \mathbb{I}(\mu_r < \mu_s)+\sum_{s=0}^{L-2}\sum_{r=s+1}^{L-1}\mathbb{I}(\lambda_{r}>\lambda_{s})\right) }{L(L-1)}=1,
\]
which holds if and only if  $s\mapsto \lambda_s$ is strictly increasing and $s\mapsto \mu_s$ is strictly decreasing. 

Part~\ref{prop3.2} follows similarly, due to symmetry.  

For part~\ref{prop3.3}, if  $s\mapsto \lambda_s$ and $s\mapsto \mu_s$ are constant, it clearly holds $\Lambda(\theta)=0$.
\end{proof}

\subsection{Proof of Proposition~\ref{coop_index_convergence}}
\begin{proof}
\textbf{Part~\ref{prop4.3_case1}} For distinct $s,r\in\{1,\dots,L\}$, since $\mu_{s} \neq \mu_r$ and $\lambda_{s-1}\neq \lambda_{r-1}$, the map
\[
(a,b) \mapsto \sign(a - b) 
\]
is continuous at $(\mu_{s}, \mu_{r})$ and also at $(\lambda_{r-1}, \lambda_{s-1})$. Then, by the continuous mapping theorem, 
\begin{multline*}
\Lambda(\hat{\theta}_n)=\frac{1}{L(L-1)}\sum_{s=1}^{L-1}\sum_{r=s+1}^L
\bigl(\sign(\hat{\mu}_s-\hat{\mu}_r)+\sign(\hat{\lambda}_{r-1}-\hat{\lambda}_{s-1})\bigr)\\
\xrightarrow{\mathbb{P}}\; \frac{1}{L(L-1)}\sum_{s=1}^{L-1}\sum_{r=s+1}^L
\bigl(\sign(\mu_s-\mu_r)+\sign(\lambda_{r-1}-\lambda_{s-1})\bigr) = \Lambda(\theta_\circ).
\end{multline*}
As $\Lambda(\cdot)$ takes values in a finite set, the convergence in probability is equivalent to  \[
	\mathbb{P}\bigl(\Lambda(\hat{\theta}_n)=\Lambda(\theta_{\circ})\bigr)\to 1,
	\qquad \text{as } n\to\infty.
	\] 
 
\smallskip
 
\textbf{Part~\ref{prop4.3_case2}} Introduce the maps $W_1,W_2 : \Theta = [0,\infty)^{2L}\to \mathbb{R}$ as 
\begin{align*}
W_{1}(\theta)&=\frac{1}{L(L-1)}\left[\sum_{(s,r)\notin A,s<r} \sign({\theta}_{L+s} - {\theta}_{L+r}) +\sum_{(s,r)\notin B,s<r}\sign({\theta}_{r} -{\theta}_{s})\right]\\
\text{and}\quad
W_{2}(\theta)&=\frac{1}{L(L-1)}\left[\sum_{(s,r)\in A,s<r} \sign({\theta}_{L+s} - {\theta}_{L+r})  +\sum_{(s,r)\in B,s<r}\sign({\theta}_{r} -{\theta}_{s})\right]
\end{align*}
for $\theta = (\theta_1,\ldots,\theta_{2L}) \in \Theta$. Then \(\Lambda(\hat{\theta}_n)=W_{1}(\hat\theta_n)+W_{2}(\hat\theta_n)\), and similar to part~\ref{prop4.3_case1}, by the continuous mapping theorem, 
\[ 
W_{2}(\hat{\theta}_n)\xrightarrow{\mathbb{P}}W_2(\theta_{\circ}).
\]
Note that $W_{1}(\hat{\theta}_n)=W_1\bigl(a_n(\hat{\theta}_n-\theta_{\circ})\bigr)$ and by assumption, the probability that $Z$ lies in the set of continuity points of $W_1(\cdot)$ is one. The assumption that $a_n(\hat{\theta}_n-\theta_{\circ}) \xrightarrow{D}Z$, together with  the continuous mapping theorem, implies
\[
W_{1}(\hat{\theta}_n)\xrightarrow{D}W_1(Z).
\]    
Then, Slutsky's lemma yields the assertion.
\end{proof}

\subsection{Proof of Theorem \ref{consistency_ind}}

We begin with an auxiliary result establishing the convergence of the soft-thresholding operator to the sign function.

\begin{lemma}\label{l:soft}
Let $\hat{\theta}_n\in\mathbb{R}$ be a consistent estimator of $\theta\in\mathbb{R}$ such that
\[
\hat{\theta}_n-\theta=o_{\mathbb{P}}(a_n)
\]
for some sequence $a_n\searrow 0$. Then, as $n\to\infty$,
\[
\soft_{a_n}(\hat{\theta}_n)\xrightarrow{\mathbb{P}}\sign(\theta).
\]
\end{lemma}

\begin{proof}
We distinguish two cases.

\begin{enumerate}
\item \emph{Case $\theta=0$.}
In this case, $\sign(\theta)=0$. Since $\hat{\theta}_n-\theta=o_{\mathbb{P}}(a_n)$, we have
\[
\mathbb{P}\bigl(|\hat{\theta}_n|\le a_n\bigr)\to 1,
\]
which implies
\(
\mathbb{P}\bigl(\soft_{a_n}(\hat{\theta}_n)=0\bigr)\to 1.
\)
Hence, $\soft_{a_n}(\hat{\theta}_n)\xrightarrow{\mathbb{P}}0=\sign(\theta)$.

\item \emph{Case $\theta\neq 0$.}
Define the events
\[
A_n:=\Bigl\{|\hat{\theta}_n|\ge a_n,\ \sign(\hat{\theta}_n)=\sign(\theta),\ \text{and }|\hat{\theta}_n-\theta|\le a_n\Bigr\}.
\]
On $A_n$, we obtain
\[
\bigl|\soft_{a_n}(\hat{\theta}_n)-\sign(\theta)\bigr|
= \bigl||\hat{\theta}_n|-|\theta|-a_n\bigr|
\le a_n \to 0.
\]
Moreover, $\mathbb{P}(A_n)\to 1$ as a consequence of $\hat{\theta}_n-\theta=o_{\mathbb{P}}(a_n)$. Therefore,
\[
\soft_{a_n}(\hat{\theta}_n)\xrightarrow{\mathbb{P}}\sign(\theta).\qedhere
\]
\end{enumerate}
\end{proof}

\begin{proof}[Proof of Theorem~\ref{consistency_ind}] 
By Lemma~\ref{l:soft}, we have, for $s,r\in\{1,\ldots,L\}$, 
\[
\soft_{a_n}(\hat\mu_s-\hat\mu_r) \xrightarrow{\mathbb{P}}\sign(\mu_s-\mu_r)\quad\text{and}\quad
\soft_{a_n}(\hat\lambda_{r-1}-\hat\lambda_{s-1}) \xrightarrow{\mathbb{P}}\sign(\lambda_{r-1}-\lambda_{s-1}).
\]
Thus, as $n\to\infty$,
\begin{multline*}
\Lambda_n(\hat\theta_n) = \frac{1}{L(L-1)}\sum_{s=1}^{L-1}\sum_{r=s+1}^L
\Bigl(
\soft_{a_n}(\hat\mu_s-\hat\mu_r)+\soft_{a_n}(\hat\lambda_{r-1}-\hat\lambda_{s-1})
\Bigr) \\
\xrightarrow{\mathbb{P}}\;  \frac{1}{L(L-1)}\sum_{s=1}^{L-1}\sum_{r=s+1}^L
\Bigl(
\sign(\mu_s-\mu_r)+\sign(\lambda_{r-1}-\lambda_{s-1})
\Bigr)
 = \Lambda(\theta).\qedhere
\end{multline*}
\end{proof}

\clearpage

\section{Proofs of results in Section \ref{s:estimation} }

\subsection{Proof of Lemma \ref{identifiability}}

\begin{proof}
Let $P_\delta(\theta)=e^{\delta R(\theta)}$. By Proposition~\ref{bd_proc}, the sum process $S_t$ is reversible for all $\theta \in \Theta$. Then $P_\delta(\theta)$ uniquely identifies $\delta R(\theta)$ for $\theta\in\Theta$, by \citet[Theorem 1]{Jia2016}. Further, as the map $\theta \mapsto \delta R(\theta)$ is injective (Theorem~\ref{Theorem2}\ref{i:th2:3}), we have $\theta \mapsto e^{\delta R(\theta)}$ is injective.
\end{proof}

\subsection{Proof of Theorem~\ref{HMM_consistency}}

\begin{proof}
The targeted consistency result can be proven using the classical approach of \citet{Wald49}. Here, we apply the general consistency result for maximum likelihood estimators in hidden Markov models established by \citet{LEROUX1992}. To this end, it suffices to verify Conditions~1--6 therein.

Condition~4 is immediate from the parametrization of the latent process $(S_t)_{t \ge 0}$, see Theorem~\ref{Theorem2}. Under the SD-HMM, the latent sum process $(S_t)_{t \ge 0}$ is irreducible by Lemma~\ref{l:ir:sd} and Proposition~\ref{bd_proc}\ref{i:bd_proc:1}, which implies that Condition~1 is satisfied. Conditions~2, 3, 5 and~6 follow directly from Assumption~\ref{a:em}. 

Thus, by \citet[Theorem~3]{LEROUX1992} and Proposition~\ref{HMM_ident_prop}, the assertion holds.
\end{proof}

\subsection{Proof of Theorem \ref{HMM_asymptotic_normal}}

\begin{proof}
We invoke the general result of \citet{Bickel98} and verify conditions (A1)--(A6) therein.

The assumed stationarity of the latent sum process $(S_t)_{t\ge0}$, together with its irreducibility (see Lemma~\ref{l:ir:sd} and Proposition~\ref{bd_proc}\ref{i:bd_proc:1}), implies that the sampled process $(S_{t_k})_{k=1}^n$ is ergodic \citep[Lemma~1]{LEROUX1992}. Hence, condition (A1) holds.

Condition (A2) follows from Assumption~\ref{a:reg:em2}\ref{HMM_consist_6} and the differentiability with respect to $\theta$ of the transition matrix $\exp\!\bigl(\delta R(\theta)\bigr)$, where $R(\theta)$ is defined in Theorem~\ref{Theorem2}\ref{i:th2:2}.

Conditions (A3) and (A4) are ensured by Assumption~\ref{a:reg:em2}\ref{HMM_consist_7}--\ref{HMM_consist_8}, respectively. Finally, conditions (A5) and (A6) are imposed as assumptions. The claim therefore follows from \citet[Theorem~1]{Bickel98}.
\end{proof}

\subsection{Auxiliary result}

The following result demonstrates that Assumption~\ref{a:em}\ref{HMM_consist_1}, which is slightly stronger than Condition~2 in \citet{LEROUX1992}, strengthens identifiability from identifiability in the quotient space to exact identifiability.

\begin{prop}\label{HMM_ident_prop}
Under the SD-HMM (Model~\ref{model:SDHMM}), suppose that Assumption~\ref{a:em}\ref{HMM_consist_1} holds, and let $\eta=(\theta,\phi)\in\Theta\times\Phi$ and $\eta'=(\theta',\phi')\in\Theta\times\Phi$. Then
\[
\eta \sim \eta' \quad \text{if and only if} \quad \eta=\eta'
\]
where $\sim$ denotes the equivalence relation defined in \citet{LEROUX1992}.
\end{prop}

\begin{proof}
The implication $\eta=\eta' \Rightarrow \eta\sim\eta'$ is immediate. 

Conversely, suppose that $\eta\sim\eta'$. Note that Assumption~\ref{a:em}\ref{HMM_consist_1} implies Condition~2 of \citet{LEROUX1992}. By Lemma~2 of \citet{LEROUX1992}, the joint stationary distribution of $(Y_{n-1},Y_n)$ coincides under $\eta=(\theta,\phi)$ and $\eta'=(\theta',\phi')$, namely,
\begin{equation}\label{e:two:stat}
\sum_{j,k=0}^L \pi^\star_j(\theta)\, p_{j,k}(\theta)\,
g_\phi(y_1\mid j)\, g_\phi(y_2\mid k)
=
\sum_{j,k=0}^L \pi^\star_j(\theta')\, p_{j,k}(\theta')\,
g_{\phi'}(y_1\mid j)\, g_{\phi'}(y_2\mid k)
\end{equation}
for $(\nu\times \nu)$-almost every $(y_1,y_2)$. Here $p_{j,k}(\theta)$ denotes the $(j,k)$-th entry of $e^{\delta R(\theta)}$, and $\pi^\star_j(\theta)$ the $j$-th component of the invariant distribution from Proposition~\ref{bd_proc}.

Integrating both sides of~\eqref{e:two:stat} with respect to $y_2$ and using that $g_\phi(\cdot\mid k)$ and $g_{\phi'}(\cdot\mid k)$ integrate to one for all $k$, we obtain
\[
\sum_{j=0}^L \pi^\star_j(\theta)\, g_\phi(y_1\mid j)
=
\sum_{j=0}^L \pi^\star_j(\theta')\, g_{\phi'}(y_1\mid j)
\]
for $\nu$-almost every $y_1$. By Assumption~\ref{a:em}\ref{HMM_consist_1}, this implies
\[
\pi^\star_j(\theta)=\pi^\star_j(\theta') \quad \text{for all } j=0,\dots,L,
\qquad \text{and} \qquad \phi=\phi'.
\]

Substituting these equalities back into \eqref{e:two:stat} yields
\[
\sum_{j=0}^L\sum_{k = 0}^L \pi^\star_j(\theta)\, p_{j,k}(\theta)\,
g_\phi(y_1\mid j)\, g_\phi(y_2\mid k)
=
\sum_{j=0}^L\sum_{k = 0}^L \pi^\star_j(\theta)\, p_{j,k}(\theta')\,
g_\phi(y_1\mid j)\, g_\phi(y_2\mid k)
\]
which is equivalent to 
\[
\sum_{j=0}^L g_\phi(y_1\mid j)\left( \sum_{k = 0}^L\alpha_{j,k}\,
g_\phi(y_2\mid k)\right)=0
\quad \text{with } \alpha_{j,k}:=\pi^\star_j(\theta)\bigl(p_{j,k}(\theta)-p_{j,k}(\theta')\bigr).
\]
By Assumption~\ref{a:em}\ref{HMM_consist_1}, we have, for every $j\in\{0,\ldots, L\}$,
\[
\sum_{k = 0}^L\alpha_{j,k}\,
g_\phi(y_2\mid k) = 0\quad\text{for $\nu$-almost every } y_2.
\]
Again by Assumption~\ref{a:em}\ref{HMM_consist_1}, it holds that $\alpha_{j,k} = 0$ for all $j,k\in\{0,\dots,L\}$. 
Hence $e^{\delta R(\theta)}=e^{\delta R(\theta')}$, which implies $\theta=\theta'$ by Lemma~\ref{identifiability}. Together with $\phi=\phi'$, this yields $\eta=\eta'$.
\end{proof}

\subsection{Verifying the conditions from Theorems~\ref{HMM_consistency} and~\ref{HMM_asymptotic_normal} for Example \ref{HMM_example2}} \label[appendix]{verify} 
In Example~\ref{HMM_example2}, we have $\Theta=(0,\infty)^{2L}$, the set of emission parameters $\Phi= \mathbb{R}\times (0,\infty)\times (0,\infty)^{L+1}$, and the emission density
\begin{equation}\label{e:emd}
g_\phi(y\mid i)=\frac{1}{\sqrt{2\pi}\sigma_i}\exp\left(-\frac{(y-b-\nu i)^2}{2\sigma_i^2}\right),
\end{equation} 
where $\phi=(b,\nu,\sigma_0,\dots,\sigma_L) \in \Phi$.  
We assume the stationarity of the hidden chain and examine the conditions from  \Cref{HMM_consistency,HMM_asymptotic_normal} as follows:
\begin{description}
    \item[\bf Assumption~\ref{a:em}\ref{HMM_consist_1}] 
    Let  $a_0,\dots,a_{L}, a_0',\dots,a_L'$ be nonnegative weights such that 
    $$
    \sum_{i=0}^L a_i= \sum_{i=0}^L a_i' = 1.
    $$
By the identifiability of finite normal mixtures \citep{Teicher1963} and \eqref{e:emd}, if
 \[
\sum_{i = 0}^{L}a_i g_\phi (y\;\vert\;i)=\sum_{i=0}^{L}a_i' g_{\phi'} (y\;\vert\;i )
 \]
 for almost every $y\in \mathbb{R}$, then $\phi'=\phi$ and $a_i'=a_{i}$ for all $i \in \{0, 1, \ldots, L\}$.
 
 \smallskip
 
 \item[\bf Assumption~\ref{a:em}\ref{HMM_consist_2}] 
 By \eqref{e:emd}, for each $y \in \mathbb{R}$ and $i \in \{0,\dots,L\}$,  the map $\phi \mapsto g_\phi(y\;\vert\; i )$ is continuous and  $\lim_{\norm{\phi}\to \infty}g_\phi(y\;\vert\; i )=0$.

\item[\bf Assumption~\ref{a:em}\ref{HMM_consist_3}] 
For any $i \in \{0,\dots,L\}$ and $\phi\in\Phi$, there exists a constant $C>0$ such that
\begin{align*}
    \bigl\lvert{\log(g_{\phi}(Y_1\mid i))}\bigr\rvert & = \left\lvert-\frac{1}{2}\log(2\pi)-\log(\sigma_i)-\frac{(Y_1-b-\nu i)^2}{2\sigma_i^2}\right\rvert\\
    &\le C+\frac{(Y_1-b-\nu i)^2}{2\sigma_i^2} \le C+\frac{2Y_1^2}{2\sigma_i^2}+\frac{2(b+\nu i)^2}{2\sigma_i^2}.
\end{align*}
Thus, the condition holds, since for any $\eta \in \Theta \times \Phi$, 
\begin{align*}
    \mathbb{E}_{\eta}[Y_1^2]&= \mathbb{E}_{\eta}[ \mathbb{E}_{\eta}[Y_1^2\mid S_{t_1}]]\\
    &=\sum_{i=0}^L\prob{S_{t_1}=i} \bigl(\sigma_i^2+(b+\nu i)^2\bigr) <\infty.
\end{align*}

 \smallskip

\item[\bf Assumption~\ref{a:em}\ref{HMM_consist_4}] 
The condition holds, since for any $\eta \in \Theta \times \Phi$,
\begin{align*}
    \mathbb{E}_{\eta}\left[\sup_{\substack{\phi'\in \Phi: \norm{\phi-\phi'}<\delta \\ i \in \{0,\dots,L\}} }\max\{\log g_{\phi'}(Y_1\mid j), 0\}\right]
\le      \sup_{\substack{\sigma_0',\dots,\sigma_L'>0: \sum_{j=0}^L\abs{\sigma_j'-\sigma_j}<\delta \\ i\in\{0,\dots,L\}} }\abs{\log(\sigma'_i)}
    <\infty,
\end{align*}
whenever $\delta \in (0, \min_{i} \sigma_i)$.

 \smallskip

\item[\bf Assumption~\ref{a:reg:em2}\ref{HMM_consist_6}] 
By \eqref{e:emd}, for each $i \in \{0,1,\ldots, L\}$ and $y\in \mathbb{R}$, the map $\phi \mapsto g_\phi(y\mid i)$ is twice continuously differentiable on $\Phi$.

 \smallskip

\item[\bf Assumption~\ref{a:reg:em2}\ref{HMM_consist_7}]
There exists $\varepsilon>0$ such that for all $\|\phi-\phi_\circ\|<\varepsilon$,
\begin{equation}\label{e:lobnd}
\min_{0\le s\le L} \sigma_s \ge c > 0.
\end{equation}
The first- and second-order partial derivatives of $\log g_\phi(y\mid s)$ with respect to $\phi$ are polynomials in $(y-b-\nu s)$ of degree at most $2$, multiplied by powers of $\sigma_s^{-1}$. Hence, uniformly over $\|\phi-\phi_\circ\|<\varepsilon$,
	\[
	\left|
	\frac{\partial}{\partial \phi_k}\log g_\phi(y\mid s)
	\right|
	+
	\left|
	\frac{\partial^2}{\partial \phi_k\partial \phi_l}
	\log g_\phi(y\mid s)
	\right|
	\;\le\;
	C\,(1+|y|^2),
	\]
	for some finite constant $C>0$. Squaring the bound for the first derivative yields 
	$$
	\left| \frac{\partial}{\partial \phi_k}\log g_\phi(y\mid s) \right|^2 \le C^2(1+|y|^2)^2 \le 2C^2(1+|y|^4).
	$$
	Under $\eta_\circ$, the random variable $Y_1$ has a normal mixture
	distribution and therefore possesses finite moments of all orders.
	It follows that, for all $s\in\{0,\dots,L\}$ and $k,l\in\{1,\dots,d\}$,
	\begin{align*}
	&\mathbb{E}_{\eta_\circ}
	\left[
	\sup_{\|\phi-\phi_\circ\|<\varepsilon}
	\left|
	\frac{\partial}{\partial \phi_k}
	\log g_\phi(Y_1\mid s)
	\right|^2
	\right]
	<\infty,\\
	\text{and}\quad &
	\mathbb{E}_{\eta_\circ}
	\left[
	\sup_{\|\phi-\phi_\circ\|<\varepsilon}
	\left|
	\frac{\partial^2}{\partial \phi_k\partial \phi_l}
	\log g_\phi(Y_1\mid s)
	\right|
	\right]
	<\infty.
	\end{align*}
	Similarly, using \eqref{e:lobnd} and the boundedness of $\{\phi:\|\phi-\phi_\circ\|<\varepsilon\} $, we can obtain, for $l\in\{1,2\}$,  $1\le k_1,\dots,k_l\le d$ and $y\in \mathbb{R}$, 
\[
\sup_{\phi:  \|\phi-\phi_\circ\|<\varepsilon}\left| \frac{\partial^l}{\partial \phi_{k_1}\dots \partial \phi_{k_l}}g_\phi(y\mid s) \right| \le \bar{C}(1+|y|^4) \exp\left(-\frac{y^2}{4\bar{\sigma}^2}\right) =: \bar{g}(y),
\]
with some finite constants $\bar{\sigma},\bar{C} >0$. Consequently, 
\[
\int_{\mathbb{R}} \sup_{\phi:\norm{\phi-\phi_{\circ}}<\varepsilon}\abs{\frac{\partial^l}{\partial \phi_{k_1}\dots \partial \phi_{k_l}}g_{\phi}(y\mid s)}\, \mathrm{d} y \le \int_{\mathbb{R}} \bar{g}(y) \, \mathrm{d}y <\infty. 
\]			

\smallskip

\item[\bf Assumption~\ref{a:reg:em2}\ref{HMM_consist_8}] 
For $i,j\in\{0,1,\ldots, L\}$, it holds 
\[
\frac{g_\phi(y\mid i)}{g_\phi(y\mid j)}
=
\frac{\sigma_j}{\sigma_i}
\exp\!\left(
-\frac{(y-b-\nu i)^2}{2\sigma_i^2}
+
\frac{(y-b-\nu j)^2}{2\sigma_j^2}
\right).
\]
Choose $\varepsilon>0$ such that $\sigma_s \ge c>0$ for all $s$ whenever $\|\phi-\phi_\circ\|<\varepsilon$. Then the above ratio is finite for all $y\in\mathbb{R}$ and all such $\phi$. In particular,
\[
\mathbb{P}_{\eta_\circ}\!\left(
\sup_{\|\phi-\phi_\circ\|<\varepsilon}
\max_{0\le i,j\le L}
\frac{g_\phi(Y_1\mid i)}{g_\phi(Y_1\mid j)}
= \infty
\,\Big|\, S_{t_1}=s
\right) = 0,
\quad s\in \{0,\dots,L\}.
\]
\end{description} 

\smallskip

\noindent 
Therefore, the conditions of Theorems~\ref{HMM_consistency} and~\ref{HMM_asymptotic_normal} are all satisfied.

\subsection{Proofs for \Cref{sec:testing}}\label{ss:pr:fwer}

We begin with some notation and technical preparations. For a covariance matrix
\(R\in\mathbb R^{d\times d}\) with \(d\in\mathbb N\), define
\begin{equation}\label{e:calpha}
    c_{\alpha}(R)
    :=
    \inf\left\{
    t\ge 0:
    \mathbb P\left(\max_{1\le r\le d}|Z_r|\le t\right)\ge 1-\alpha
    \right\},
    \qquad Z\sim\mathcal N(0,R).
\end{equation}

\begin{lemma}\label[lemma]{lem:gaussian_max_continuous}
For \(d\in\mathbb N\), let \(R\in\mathbb R^{d\times d}\) be a positively semidefinite matrix with $\text{rank}(R) \ge 1$. Let also $c_\alpha(R)$ and $Z = (Z_1, \ldots,Z_d)$ be given in \eqref{e:calpha}. Then, the distribution function of \(M := \max_{1\le j\le d}|Z_j|\) is continuous and strictly increasing on \([0,\infty)\). In particular, for every \(\alpha\in(0,1)\),
\[
\prob{M\le c_{\alpha}(R)}=1-\alpha,
\quad\text{ and equivalently }\quad
\prob{M> c_{\alpha}(R)}=\alpha.
\]
\end{lemma}

\begin{proof}
   Let $F$ be the distribution function of \(M\), which is continuous on $[0,\infty)$ since
\[
F(t) - F(t-) = \prob{M=t}
\le \sum_{j=1}^d \prob{\abs{Z_j}=t}
=0.
\]
The continuity of $F$ implies that $F(c_\alpha(R))=1-\alpha$. 

We next show that $F$ is strictly increasing on $[0,\infty)$. There exists a matrix $B\in\mathbb R^{d\times r}$ with full column rank $r = \text{rank}(R)$ such that 
\(
R=BB^\top.
\)
Then $Z\stackrel{D}=BX$ with $X$ a standard $r$-dimensional normal random vector, and $M\stackrel{D}=N(X)$ for $N(x)=\max_{1\le j\le d}\abs{(Bx)_j}$. Hence,
\[F(t)-F(a)=\prob{ X\in\{x:N(x)\le t\}\setminus \{x:N(x)\le a\} }>0,\]as $X$ has a strictly positive Lebesgue density. Namely, $F$ is strictly increasing.
\end{proof}
For \((i,j)\in\mathcal H\), let
\(a_{(i,j)}:=e_i-e_j\in\mathbb R^{2L}\), with \(e_i\) the
\(i\)-th canonical basis vector. Then
\[
\hat s_{n,(i,j)}^2
=
a_{(i,j)}\hat\Sigma_n a_{(i,j)}^\top
=
(\hat\Sigma_n)_{i,i}
-2(\hat\Sigma_n)_{i,j}
+(\hat\Sigma_n)_{j,j}.
\]
Similarly, introduce
\[
s_{\circ,(i,j)}^2
:=
a_{(i,j)}\Sigma_\circ a_{(i,j)}^\top
=
(\Sigma_\circ)_{i,i}
-2(\Sigma_\circ)_{i,j}
+(\Sigma_\circ)_{j,j},
\]
where
\(\Sigma_\circ
:=
(\mathcal I_{\eta_\circ}^{-1})_{\{1,\ldots,2L\},\{1,\ldots,2L\}}.\) We further introduce 
\begin{align*}
    \Gamma_\circ &
:=
\bigl(a_{(i,j)} \Sigma_\circ a_{(i',j')}^\top\bigr)_{(i,j),(i',j')\in\mathcal H},&
\qquad
\hat\Gamma_n&
:=
\bigl(a_{(i,j)} \hat\Sigma_n a_{(i',j')}^\top\bigr)_{(i,j),(i',j')\in\mathcal H},\\
S_\circ &
:=
\operatorname{diag}\bigl(s_{\circ,(i,j)}^2\bigr)_{(i,j)\in\mathcal H},&
\qquad
\hat S_n&
:=
\operatorname{diag}\bigl(\hat s_{n,(i,j)}^2\bigr)_{(i,j)\in\mathcal H},
\end{align*}
as well as the normalized matrices
\begin{equation}\label{e:RnRo}
 R_\circ
:=
S_\circ^{-1/2}\Gamma_\circ S_\circ^{-1/2},
\qquad
\hat R_n
:=
(\hat S_n^{\dagger})^{1/2}\hat\Gamma_n (\hat S_n^{\dagger})^{1/2}.
\end{equation}
Here $\hat S_n^{\dagger}$ denotes the Moore--Penrose pseudoinverse of $\hat S_n$.

\begin{lemma}\label[lemma]{lem:joint_limit_T}
Assume the conditions of Theorem~\ref{HMM_asymptotic_normal} hold. 
Then
\[
\hat\Gamma_n\xrightarrow{\mathbb{P}}\Gamma_\circ,
\qquad
\hat S_n\xrightarrow{\mathbb{P}}S_\circ,
\qquad
\hat R_n\xrightarrow{\mathbb{P}}R_\circ.
\]
Let $\mathcal H_0$ and $\mathcal H_1$ be defined in \Cref{hypothesis_test_main_thm}. 
For \((i,j)\in\mathcal H\), define
\[
\hat U_{n,(i,j)}
:=
\begin{cases}
{\sqrt n(\hat D_{n,(i,j)}-D_{\circ,(i,j)})}/{\hat s_{n,(i,j)}}, & \text{ if } \hat s_{n,(i,j)}>0,\\
0, & \text{ if } \hat s_{n,(i,j)}=0,
\end{cases}
\]
where $D_{\circ,(i,j)}:=\theta_{\circ,i}-\theta_{\circ,j}
$ and $\hat D_{n,(i,j)}:=\hat\theta_{n,i}-\hat\theta_{n,j}$. 
Then
\[
(\hat U_{n,(i,j)})_{(i,j)\in\mathcal H}
\xrightarrow{D}
\mathcal N(0,R_\circ)\quad\text{and particularly}\quad (\hat U_{n,(i,j)})_{(i,j)\in\mathcal H_0}
\xrightarrow{D}
\mathcal N_{|\mathcal H_0|}(0,R_{\circ,\mathcal H_0}),
\]
where \(R_{\circ,\mathcal H_0}\) is the principal submatrix of \(R_\circ\) indexed by \(\mathcal H_0\). Further, it holds for \((i,j)\in\mathcal H_0\), 
\[
\hat T_{n,(i,j)} \xrightarrow{D}Z,
\qquad Z\sim\mathcal N(0,1),
\]
and for \((i,j)\in\mathcal H_1\),
\[
\frac{|\hat T_{n,(i,j)}|}{\sqrt n}
\xrightarrow{\mathbb P}
\frac{|D_{\circ,(i,j)}|}{s_{\circ,(i,j)}}\quad \text{and hence}\quad
|\hat T_{n,(i,j)}|\xrightarrow{\mathbb P}\infty.
\]
\end{lemma}

\begin{proof}
 Conditions (A1)--(A4) from \citet{Bickel98} are  satisfied under the assumptions of Theorem~\ref{HMM_asymptotic_normal}, see the proof of Theorem~\ref{HMM_asymptotic_normal}.
 Thus, Lemma 2 in \cite{Bickel98} yields
\[
-\frac{1}{n}\nabla_\eta^2 \log p_\eta(Y_1,\dots,Y_n)\big|_{\eta=\hat\eta_n}
\xrightarrow{\mathbb P}
\mathcal I_{\eta_\circ}.
\]
Since the Moore--Penrose pseudoinverse is continuous at the nonsingular matrix $\mathcal{I}_{\eta_\circ}$, the continuous mapping theorem yields
\(\hat\Sigma_n
\xrightarrow{\mathbb P}
\Sigma_\circ.\)
Let
\[
A:=\bigl(a_{(i,j)}\bigr)_{(i,j)\in\mathcal H}\in\mathbb R^{|\mathcal H|\times 2L}.
\]
Then, $\Gamma_\circ=A\Sigma_\circ A^\top$ and $
\hat\Gamma_n=A\hat\Sigma_nA^\top$. By the continuous mapping theorem,
\[
\hat\Gamma_n\xrightarrow{\mathbb P}\Gamma_\circ\quad\text{and}\quad
\hat S_n=\operatorname{diag}(\hat\Gamma_n)\xrightarrow{\mathbb P}\operatorname{diag}(\Gamma_\circ)=S_\circ.
\]
Note that \(s_{\circ,(i,j)}^2=a_{(i,j)}\Sigma_\circ a_{(i,j)}^\top>0\) for all \((i,j)\in\mathcal H\). Hence, the map 
\[
M\mapsto (\operatorname{diag}(M)^{\dagger})^{1/2}M\,(\operatorname{diag}(M)^{\dagger})^{1/2}
\]
is continuous at \(\Gamma_\circ\). As \((\operatorname{diag}(\Gamma_\circ)^{\dagger})^{1/2}
=
\operatorname{diag}(\Gamma_\circ)^{-1/2}\), we obtain
\[
\hat R_n=(\hat S_n^{\dagger})^{1/2}\hat\Gamma_n(\hat S_n^{\dagger})^{1/2}\xrightarrow{\mathbb P}R_\circ,
\]
by the continuous mapping theorem. 

Let  
\(\hat D_n:=(\hat D_{n,(i,j)})_{(i,j)\in\mathcal H}=A\hat\theta_n\) and \(D_\circ:=(D_{\circ,(i,j)})_{(i,j)\in\mathcal H}=A\theta_\circ.\) By Theorem~\ref{HMM_asymptotic_normal},
\[
\sqrt n\,(\hat D_n-D_\circ)
=
\sqrt n\,A(\hat\theta_n-\theta_\circ)
\xrightarrow{D}
\mathcal N(0,\Gamma_\circ).
\]
By Slutsky's lemma,
\[
(\hat U_{n,(i,j)})_{(i,j)\in\mathcal H} = \sqrt n\,(\hat S_n^{\dagger})^{1/2}(\hat D_n-D_\circ)
\xrightarrow{D}
\mathcal N(0,R_\circ).
\]
Selecting only the coordinates indexed by \(\mathcal H_0\), where \(D_{\circ,(i,j)}=0\), yields
\[
(\hat U_{n,(i,j)})_{(i,j)\in\mathcal H_0}
\xrightarrow{D}
\mathcal N_{|\mathcal H_0|}(0,R_{\circ,\mathcal H_0}).
\]
This further implies, for 
\((i,j)\in\mathcal H_0\),
\(
\hat U_{n,(i,j)}\xrightarrow{D}\mathcal N(0,1)
\) and 
\[
\hat T_{n,(i,j)}
=
\hat U_{n,(i,j)}
\xrightarrow{D}Z,
\qquad Z\sim\mathcal N(0,1).
\]

Consider now \((i,j)\in\mathcal H_1\). The consistency of \(\hat\theta_n\) and \(\hat\Sigma_n\) yields
\[
\hat D_{n,(i,j)}\xrightarrow{\mathbb P}D_{\circ,(i,j)},
\qquad
\hat s_{n,(i,j)}\xrightarrow{\mathbb P}s_{\circ,(i,j)}.
\]
Since \(s_{\circ,(i,j)}>0\), it follows that
\(\prob{\hat s_{n,(i,j)}>0}\to 1.\)
Therefore,
\[
\frac{|\hat T_{n,(i,j)}|}{\sqrt n}
=
\frac{|\hat U_{n,(i,j)}|}{\sqrt n}
\xrightarrow{\mathbb P}
\frac{|D_{\circ,(i,j)}|}{s_{\circ,(i,j)}}.
\]
Since \(D_{\circ,(i,j)}\neq 0\), we have
\(|\hat T_{n,(i,j)}|\xrightarrow{\mathbb P}\infty.\)
\end{proof}

For any nonempty subset $\mathcal K \subseteq \mathcal H$, let $R_{\circ,\mathcal K}$ and $\hat R_{n,\mathcal K}$ denote the principal submatrices of $R_{\circ}$ and $\hat R_n$ in \eqref{e:RnRo}, respectively, indexed by $\mathcal K$. Then we have
\begin{equation}\label{e:chca}
 \hat c_{n,\alpha}(\mathcal K) = c_{\alpha}\bigl(\hat R_{n,\mathcal K}\bigr),   
\end{equation}
where $\hat c_{n,\alpha}(\cdot)$ and $c_{\alpha}(\cdot)$ are as defined in \eqref{eq:thd:joint} and \eqref{e:calpha}, respectively.

\begin{lemma}\label[lemma]{lem:plugin_quantile}
Under the conditions of Theorem~\ref{HMM_asymptotic_normal}, it holds that
\[
    \max_{\emptyset\neq \mathcal K\subseteq\mathcal H}
    \left|
    c_{\alpha}\bigl(\hat R_{n,\mathcal K}\bigr)
    -
    c_{\alpha}\bigl(R_{\circ,\mathcal K}\bigr)
    \right|
    \xrightarrow{\mathbb P}0.
    \]
\end{lemma}

\begin{proof}
Let $R_m\to R$, where each $R_m \in \R^{d\times d}$ is a positively semidefinite matrix, and let
\[
    Z_m\sim \mathcal N(0,R_m),
    \qquad
    Z\sim \mathcal N(0,R).
\]    
Then, for every $u\in\mathbb R^d$,
\[
    \mathbb E\exp(iu^\top Z_m)
    =
    \exp\left(-\frac{u^\top R_m u}{2}\right)
    \to
    \exp\left(-\frac{u^\top R u}{2}\right)
    =
    \mathbb E\exp(iu^\top Z),
\]
and hence $Z_m \xrightarrow{D} Z$. By the continuous mapping theorem,
\[
   M_m :=  \max_{1\le j\le d}|Z_{m,j}|
    \;\xrightarrow{D}\;
    M:=\max_{1\le j\le d}|Z_j|.
\]
As the distribution functions of $M_m$ and $M$ are continuous and strictly increasing by \Cref{lem:gaussian_max_continuous}, their quantile functions are continuous on $(0,1)$ by \citet[Lemma~21.2]{v}. Thus,
\[
c_{\alpha}(R_m)
    \to
    c_{\alpha}(R),
\]
namely, the map $R\mapsto c_{\alpha}(R)$ is continuous. 

By Lemma~\ref{lem:joint_limit_T},
\(\hat R_n\xrightarrow{\mathbb P} R_\circ \) and thus 
\( \hat R_{n,\mathcal K}
    \xrightarrow{\mathbb P}
    R_{\circ,\mathcal K}\), for every nonempty
$\mathcal K\subseteq\mathcal H$. Then the continuity of $c_{\alpha}(\cdot)$ established above yields
\[
    c_{\alpha}\bigl(\hat R_{n,\mathcal K}\bigr)
    -
    c_{\alpha}\bigl(R_{\circ,\mathcal K}\bigr)
    \xrightarrow{\mathbb P}0.
\]

Finally, since $\mathcal H$ is finite, the number of nonempty subsets
$\mathcal K\subseteq\mathcal H$ is finite. Hence,
\[
    \max_{\emptyset\neq \mathcal K\subseteq\mathcal H}
    \left|
    c_{\alpha}\bigl(\hat R_{n,\mathcal K}\bigr)
    -
    c_{\alpha}\bigl(R_{\circ,\mathcal K}\bigr)
    \right|
    \xrightarrow{\mathbb P}0.\qedhere
\]
\end{proof}

Now we are ready to prove the theoretical results in \Cref{sec:testing}.

\begin{proof}[Proof of \Cref{hypothesis_test_main_thm}]
\textbf{Part~\ref{test_proof1}}
If \(\mathcal H_1=\emptyset\), the assertion is immediate. Hence assume that
\(\mathcal H_1\neq\emptyset\). By \Cref{lem:joint_limit_T} and the finiteness of $\mathcal H_1$,
\begin{equation}\label{e:testinH1}
    B_n := \min_{(i,j)\in\mathcal H_1}|\hat T_{n,(i,j)}|
\xrightarrow{\mathbb P}\infty.
\end{equation}
Next, by \eqref{e:chca} and \Cref{lem:plugin_quantile},  
\[
\max_{\emptyset\neq \mathcal K\subseteq\mathcal H}
\left|
\hat c_{n,\alpha}(\mathcal K)
-
c_{\alpha}(R_{\circ,\mathcal K})
\right|
\xrightarrow{\mathbb P}0 .
\]
Consequently, due to the finiteness of $\mathcal{H}$,
\begin{align*}
E_n := \max_{\emptyset\neq \mathcal K\subseteq\mathcal H}\hat c_{n,\alpha}(\mathcal K)
&\le
\max_{\emptyset\neq \mathcal K\subseteq\mathcal H}
\left|
\hat c_{n,\alpha}(\mathcal K)-c_{\alpha}\bigl(R_{\circ,\mathcal K}\bigr)
\right|
+
\max_{\emptyset\neq \mathcal K\subseteq\mathcal H}
c_{\alpha}\bigl(R_{\circ,\mathcal K}\bigr) \\
&=o_{\mathbb P}(1)+O(1)=O_{\mathbb P}(1).
\end{align*}
This together with \eqref{e:testinH1} implies 
\[
\prob{B_n>E_n}\to 1 .
\]

On the event \(\{B_n>E_n\}\), every false null hypothesis is rejected in the
first step of \Cref{alg:stepdown}. Indeed, at the first step the active set is
\(\mathcal K=\mathcal H\), and for every \((i,j)\in\mathcal H_1\),
\[
|\hat T_{n,(i,j)}|
\ge B_n
>
E_n
\ge
\hat c_{n,\alpha}(\mathcal H).
\]
Therefore, each \((i,j)\in\mathcal H_1\) belongs to the first-step rejection set.
Since the final rejection set \(\widehat{\mathcal H}_{1,n}\) contains all
hypotheses rejected in the first step, it follows that
\[
\prob{\mathcal H_1\subseteq \widehat{\mathcal H}_{1,n}}
\ge
\prob{B_n>E_n}
\to 1 .
\]

\smallskip

\textbf{Part~\ref{test_proof2}}
We apply the general result in Theorem~4(b) of \citet{RoWo05} on the stepdown procedure as considered in \Cref{alg:stepdown}.

First, we verify their Assumption~A1. By Lemma~\ref{lem:joint_limit_T} and the continuous mapping theorem, 
\[
    \max_{(i,j)\in\mathcal H_0}|\hat U_{n,(i,j)}|
    \xrightarrow{D}
    \max_{(i,j)\in\mathcal H_0}|Z_{(i,j)}|,
    \qquad
    Z\sim
    \mathcal N_{|\mathcal H_0|}\bigl(0,R_{\circ,\mathcal H_0}\bigr).
\]
By Lemma~\ref{lem:gaussian_max_continuous}, the limiting distribution function
is continuous and strictly increasing at
$c_{\alpha}\bigl(R_{\circ,\mathcal H_0}\bigr)$. Thus, Assumption~A1 of
\citet{RoWo05} holds.

Next, their required monotonicity condition follows directly from the definition of the
critical values. If
$\emptyset\neq \mathcal K_1\subseteq\mathcal K_2\subseteq\mathcal H$, then,
for the corresponding Gaussian vector,
\[
    \max_{(i,j)\in\mathcal K_1}|Z_{(i,j)}|
    \le
    \max_{(i,j)\in\mathcal K_2}|Z_{(i,j)}|
    \qquad \text{a.s.}
\]
and thus
\(\hat c_{n,\alpha}(\mathcal K_1)
    \le
    \hat c_{n,\alpha}(\mathcal K_2),\)
so condition~(30) of \citet{RoWo05} is satisfied.

It remains to verify their condition~(31), which follows,  by \eqref{e:chca} and \Cref{lem:plugin_quantile}, from
\[
    \prob{
    \hat c_{n,\alpha}(\mathcal H_0)
    \ge
    c_{\alpha}\bigl(R_{\circ,\mathcal H_0}\bigr)-\varepsilon
    }
    \to 1,\qquad \varepsilon > 0.
\]

Now all assumptions of Theorem~4(b) in \citet{RoWo05} are  satisfied, and
hence
\[
    \limsup_{n\to\infty}
    \prob{
    \mathcal H_0\cap\widehat{\mathcal H}_{1,n}\neq\emptyset
    }
    \le \alpha. \qedhere
\]
\end{proof}

\begin{proof}[Proof of \Cref{test_proof3}]
Let
\[
    C_n
    :=
    \bigcap_{(i,j)\in\mathcal H_1}
    \left\{
    \sign(\hat\theta_{n,i}-\hat\theta_{n,j})
    =
    \sign(\theta_{\circ,i}-\theta_{\circ,j})
    \right\}.
\]
If $\mathcal H_1=\emptyset$, we interpret $C_n$ as the whole sample space such that $\prob{C_n} = 1$.

Since $\hat\theta_n\xrightarrow{\mathbb P}\theta_\circ$ by \Cref{HMM_asymptotic_normal}, for every
$(i,j)\in\mathcal H$,
\[
    \hat\theta_{n,i}-\hat\theta_{n,j}
    \xrightarrow{\mathbb P}
    \theta_{\circ,i}-\theta_{\circ,j}.
\]
For $(i,j)\in\mathcal H_1$, the limit satisfies
$\theta_{\circ,i}-\theta_{\circ,j}\neq 0$. Hence,
\[
    \prob{
    \sign(\hat\theta_{n,i}-\hat\theta_{n,j})
    =
    \sign(\theta_{\circ,i}-\theta_{\circ,j})
    }
    \to 1 .
\]
Because $\mathcal H_1$ is finite, the union bound gives
\begin{equation}\label{e:Cnc}
  \prob{C_n^c}
    \le
    \sum_{(i,j)\in\mathcal H_1}
    \prob{
    \sign(\hat\theta_{n,i}-\hat\theta_{n,j})
    \neq
    \sign(\theta_{\circ,i}-\theta_{\circ,j})
    }
    \to 0.  
\end{equation}

Now define
\[
    W_n:=\{\mathcal H_1\subseteq\widehat{\mathcal H}_{1,n}\},
    \qquad
    G_n:=\{\mathcal H_0\cap\widehat{\mathcal H}_{1,n}=\emptyset\}.
\]
We claim that, on $W_n\cap G_n\cap C_n$, for every $(i,j)\in\mathcal H$,
\begin{equation}\label{hypothesis2_final}
    \mathbb I\bigl((i,j)\in\widehat{\mathcal H}_{1,n}\bigr)
    \sign(\hat\theta_{n,i}-\hat\theta_{n,j})
    =
    \mathbb I(\theta_{\circ,i}\neq\theta_{\circ,j})
    \sign(\theta_{\circ,i}-\theta_{\circ,j}).
\end{equation}
Indeed, if $(i,j)\in\mathcal H_0$, then the right-hand side of
\eqref{hypothesis2_final} is zero. On $G_n$, no true null hypothesis is
rejected, so $(i,j)\notin\widehat{\mathcal H}_{1,n}$ and the left-hand side is
also zero. If $(i,j)\in\mathcal H_1$, then on $W_n$ we have
$(i,j)\in\widehat{\mathcal H}_{1,n}$, while on $C_n$ the estimated sign equals
the true sign. Thus \eqref{hypothesis2_final} holds in both cases.
Consequently, 
\(
    \hat\Lambda_{n,\alpha}
    =
    \Lambda(\theta_\circ),
\) on $W_n\cap G_n\cap C_n$.
Therefore,
\[
    \prob{\hat\Lambda_{n,\alpha}=\Lambda(\theta_\circ)}
    \ge
    \prob{W_n\cap G_n\cap C_n}
    \ge
    1-\prob{W_n^c}-\prob{G_n^c}-\prob{C_n^c}.
\]
By \Cref{hypothesis_test_main_thm},
\(\prob{W_n^c}\to 0\) and \(\limsup_{n\to\infty}\prob{G_n^c}\le \alpha\). Hence, by \eqref{e:Cnc},
\[
    \liminf_{n\to\infty}
    \prob{\hat\Lambda_{n,\alpha}=\Lambda(\theta_\circ)}
    \ge 1-\alpha.
\]If now $\mathcal{H}_0=\emptyset$, it follows that $G_n^c=\emptyset$ and $\prob{G_n}=1$. As a direct consequence,
 \[
 1\ge \prob{\hat\Lambda_{n,\alpha}=\Lambda(\theta_\circ)}
    \ge
    \prob{W_n\cap C_n}
    \ge
    1-\prob{W_n^c}-\prob{C_n^c}\to 1.\qedhere
    \]
\end{proof}

\clearpage

%
%
%
%
%
%
%
%

\section{Implementation details and additional results}
\label{s:imp}

\subsection{Maximum likelihood estimation by a modified Baum--Welch algorithm}
\label{s:imp:baumwelch}

Let $Y_1,\dots,Y_n$ denote observations generated by an SD-HMM (Model~\ref{model:SDHMM}), and let $(S_t)_{t\ge 0}$ be the hidden sum process of an SDMC with transition parameter
\[
    \theta=(\lambda_0,\ldots,\lambda_{L-1},\mu_1,\ldots,\mu_L)    \in (0,\infty)^{2L}.
\]
We write $\eta=(\theta,\phi)\in\Theta\times\Phi$, where $\phi$ denotes the emission parameter.  The maximum likelihood estimator $\hat\eta_n$ in \eqref{e:mle} is computed using a modified version of the Baum--Welch algorithm by \citet{Baum1970}. Given a current parameter value  $\eta^{(r)}=(\theta^{(r)},\phi^{(r)})$, the algorithm alternates between an \emph{expectation} step, in which the posterior state probabilities are computed by the forward--backward algorithm, and a \emph{maximization} step, in which the expected complete-data log-likelihood is maximized with respect to the transition and emission parameters. See Algorithm~\ref{alg:modified_baum_welch}.

For $\eta = (\theta, \phi)$, we have
\begin{align*}
\mathcal{F}(\eta\mid \eta^{(r)})&:=\mathbb{E}_{\eta^{(r)}}[\log p_\eta (Y_1,\dots,Y_n,S_{t_1},\dots,S_{t_n})\mid Y_1,\dots,Y_n ]
        \\
        &=\sum_{k=2}^{n}\sum_{s=0}^L\sum_{s' = 0}^L\log (e^{\delta R(\theta)})_{s,s'}\mathbb{P}_{\eta^{(r)}}(S_{t_k}=s',S_{t_{k-1}}=s\mid Y_1,\dots,Y_n)\\
        &\quad+\sum_{k=1}^{n}\sum_{s=0}^L \log g_{\phi}(Y_k\mid s)\mathbb{P}_{\eta^{(r)}}(S_{t_k}=s\mid Y_1,\dots,Y_n)\\
        &\quad+\sum_{s=0}^L \log \pi_{\theta}(s) \mathbb{P}_{\eta^{(r)}}(S_{t_1}=s\mid Y_1,\dots,Y_n).
\end{align*}
Thus, $\max_{\eta}  \mathcal{F}(\eta\mid \eta^{(r)})$ is equivalent to solving \eqref{e:opt:theta} and \eqref{e:opt:phi}. 

\begin{algorithm}[t]
\caption{Modified Baum--Welch algorithm for SD-HMMs}
\label{alg:modified_baum_welch}
\begin{algorithmic}[1]
\Require Observations $Y_1,\ldots,Y_n$; number of channels $L$; sampling
interval $\delta$
\State Set initial values $\theta^{(1)}\in(0,\infty)^{2L}$ and $\phi^{(1)}\in\Phi$ as in Appendix~\ref{s:imp:initialization}
\For{$r = 1,\ldots, r_{\max} (=100)$}
    \State {\bf Expectation:} Compute the posterior state probabilities via the forward--backward algorithm
    \begin{align*}
     \gamma_k^{(r)}(s) & \gets \mathbb{P}_{\eta^{(r)}}(S_{t_k}=s\mid Y_1,\dots,Y_n), && s\in\{0,1,\ldots,L\}\\
     \xi_k^{(r)}(s,s') & \gets \mathbb{P}_{\eta^{(r)}}(S_{t_k}=s',S_{t_{k-1}}=s\mid Y_1,\dots,Y_n), &&s,s'\in\{0,1,\ldots,L\}
    \end{align*}
    \State {\bf Maximization:} Update the parameter $(\theta^{(r)},\phi^{(r)})$ by solving 
    \begin{align}
      \theta^{(r+1)}\gets & 
        \mathop{\arg\max}_{\theta\in(0,\infty)^{2L}}
        \sum_{k=2}^{n}\sum_{s=0}^{L}\biggl[\gamma_1^{(r)}(s)\log \pi_\theta(s) + \sum_{s'=0}^L
        \xi_k^{(r)}(s,s')
        \log \bigl(\exp\{\delta R(\theta)\}\bigr)_{s,s'}
        \biggr] \label{e:opt:theta} \\
    \phi^{(r+1)}\gets & \mathop{\arg\max}_{\phi \in \Phi}   \sum_{k=1}^{n}\sum_{s=0}^{L}
        \gamma_k^{(r)}(s)\log g_\phi(Y_k\mid s)\label{e:opt:phi}
    \end{align}
    \If{$\max\{\|\theta^{(r+1)}-\theta^{(r)}\|_\infty,\|\phi^{(r+1)}-\phi^{(r)}\|_\infty\}<\varepsilon (= 10^{-6})$}
        \State \textbf{break}
    \EndIf
\EndFor
\State \Return (approximate) maximum likelihood estimate $\hat\eta_n=(\theta^{(r+1)},\phi^{(r+1)})$
\end{algorithmic}
\end{algorithm}

\subsection{Initialization}
\label{s:imp:initialization}

Since the log-likelihood
\(
    \eta\mapsto \log p_\eta(Y_1,\dots,Y_n)
\)
is generally non-concave, the modified Baum--Welch algorithm (\Cref{alg:modified_baum_welch}) is not guaranteed to converge to a global maximum. It may instead converge to a local maximum. Consequently, the choice of the initial value $\eta^{(1)}=(\theta^{(1)},\phi^{(1)})$ is important.

In our numerical experiments and real-data applications, for sufficiently small sampling intervals $\delta>0$, the initialization of the SDMC rate parameter $\theta^{(0)}$ has negligible influence on the resulting estimate. In the implementation, a simple default choice is
\[
    \theta^{(1)}=(10,\ldots,10)\in(0,\infty)^{2L}.
\]
By contrast, the initial value of the emission parameter has a more pronounced effect on the procedure.

For the Gaussian emission model (\Cref{HMM_example2}), a natural initialization is based on a histogram of the observations. Let $Y_{\min}$ and $Y_{\max}$ denote suitable lower and upper empirical signal levels. Then the initial emission parameters are chosen as
\[
    b^{(1)}=Y_{\min},
    \qquad
    a^{(0)}=\frac{Y_{\max}-Y_{\min}}{L},
\]
so that the initial means are placed linearly across the $L+1$ levels,
\[
    b^{(1)}+a^{(1)}s,
    \qquad s\in \{0,\ldots,L\}.
\]
If the observations are too noisy for the histogram to reveal the levels reliably, one may first idealize the data using an idealization algorithm (e.g.\ \citealp{liu2024multiscale}), and then apply the histogram-based initialization to the idealized observations.

The initial standard deviation $\sigma_s^{(1)}$ is chosen as the sample standard deviation of 
$$
\{Y_1,\ldots, Y_n\} \cap  \bigl[b^{(1)}+a^{(1)}(s-1/2), \,  b^{(1)}+a^{(1)}(s+1/2)\bigr)
$$
for $s \in \{0,\ldots, L\}$. 

\subsection{Numerical implementation}
\label{s:imp:numerical_implementation}

The numerical implementation used for the simulations and data applications is
provided in the accompanying {R} package \texttt{SDMC} (available at \url{https://gitlab.gwdg.de/requadt/sdmc}). The package contains routines for constructing SDMC transition matrices, simulating the sum process, fitting Gaussian-emission SD-HMMs, evaluating the observed HMM likelihood, computing Viterbi paths, and performing the SCoT for testing cooperativity.

The maximization step of \Cref{alg:modified_baum_welch} is computed using R package \texttt{stats}. The optimization over parameter $\theta$ in \eqref{e:opt:theta} is computed by the Nelder--Mead algorithm \citep{NeMe65} via function \texttt{optim}, and the optimization over parameter $\phi$ in \eqref{e:opt:phi}, involving a constraint that levels are equally spaced, is solved by function \texttt{constrOptim}. 

In \Cref{alg:modified_baum_welch}, the initial distribution is estimated freely and updated by the usual Baum--Welch update. As an alternative, we also allow the user to impose stationary initial distribution. More precisely, in the $r$-th iteration, the   initial distribution is set as the stationary distribution with respect to $\theta^{(r)}$, according to \Cref{bd_proc}, i.e.
\[
    \pi_\theta^{(r)}(j)
    =
    \frac{
        \binom{L}{j}
        \prod_{\ell=0}^{j-1}
        \lambda_\ell/^{(r)}\mu_{\ell+1}^{(r)}
    }{
        1+
        \sum_{k=1}^{L}
        \binom{L}{k}
        \prod_{\ell=0}^{k-1}
        \lambda_\ell^{(r)}/\mu_{\ell+1}^{(r)}
    },\qquad j\in \{0,1,\ldots, L\},
\]
with the empty product interpreted as one.

The critical values in \eqref{eq:thd:joint} for SCoT are computed using function \texttt{qmvnorm} from R package \texttt{mvtnorm}. This function is based on a stochastic root-finding algorithm using local linear regression \citep{Bornkamp2018}. As an alternative, we also tested a bisection search based on the distribution function of the maximum of a multivariate Gaussian vector. In our simulations, however, this alternative approach showed performance comparable to that of \texttt{qmvnorm}.

\subsection{Additional numerical results}
\label{s:imp:additional_results}

Additional data analysis results for \Cref{section8.4} are reported in
\Cref{fig:bic_both,fig:BIC_gram,ryr_dataset2}.

\begin{figure}[ht]
    \centering
    \begin{subfigure}[b]{0.49\linewidth}
        \centering
\includegraphics[width=\linewidth, trim={1cm 0.5cm 1cm 0.5cm}, clip]{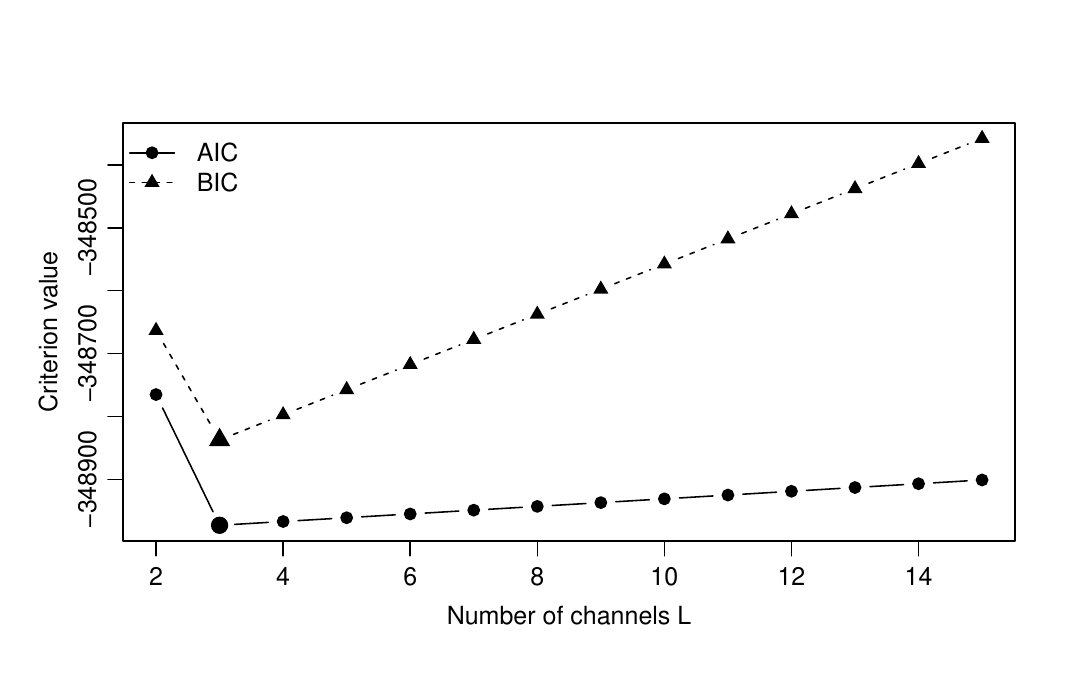}       
\caption{RyR2 Dataset 1}
        \label{fig:bic_y1}
    \end{subfigure}
    \hfill
    \begin{subfigure}[b]{0.49\linewidth}
        \centering
\includegraphics[width=\linewidth, trim={1cm 0.5cm 1cm 0.5cm}, clip]{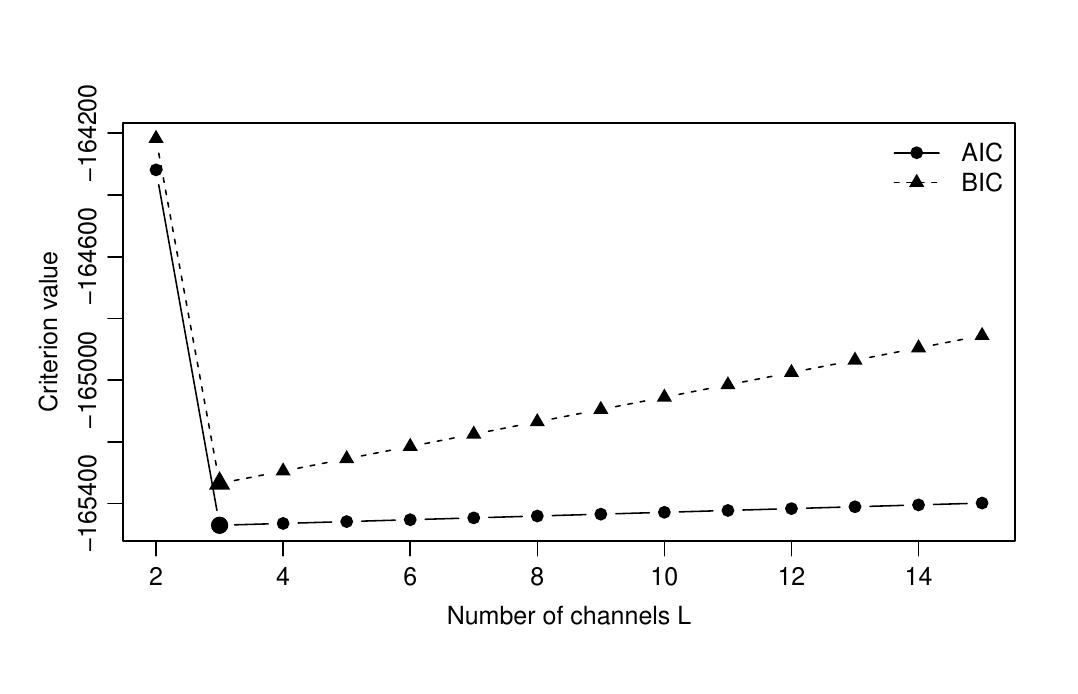}
\caption{RyR2 Dataset 2}
        \label{fig:bic_y2}
    \end{subfigure}
    \caption{BIC and AIC criteria for the two datasets of RyR2 channels. Both criteria are minimized at $L=3$ for the two datasets.}
    \label{fig:bic_both}
\end{figure}

\begin{figure}[hb]
    \centering
    \includegraphics[width=0.7\linewidth]{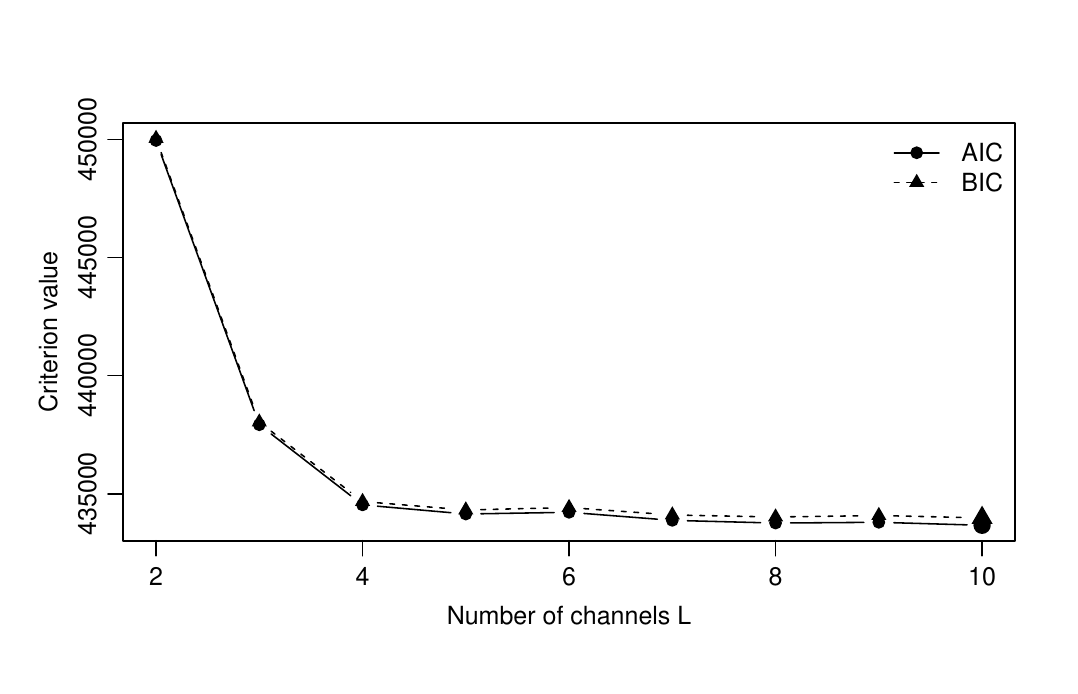}
    \caption{BIC and AIC criteria for the gramicidin D channels. Both criterions are minimal at $L=10$.}
    \label{fig:BIC_gram}
\end{figure}

\begin{figure}[ht]
    \centering
    \includegraphics[width=0.9\linewidth]{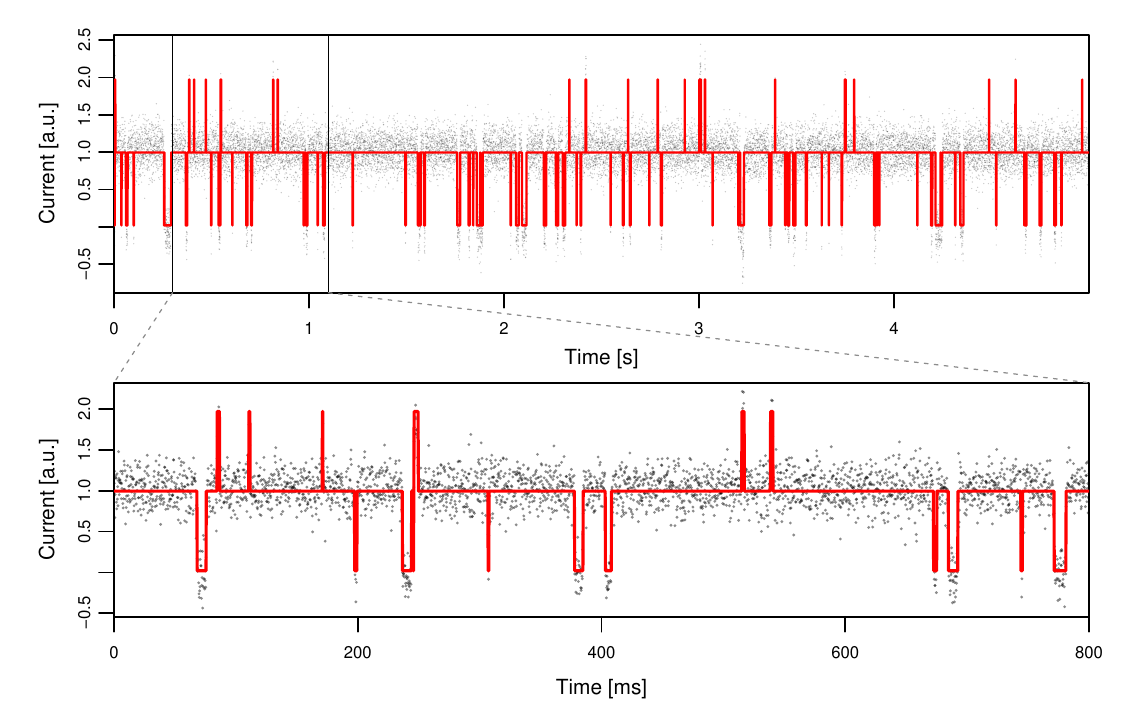}
    \caption{Dataset~2 of RyR2 channels:  current trace, consisting of $20000$ observations, is shown as black dots, and its Viterbi path is plot as a red line.}
    \label{ryr_dataset2}
\end{figure}

\end{document}